\documentclass[lettersize,journal]{IEEEtran}
\usepackage{cite}
\usepackage[T1]{fontenc}
\usepackage{bm}
\usepackage{amsmath}
\usepackage[table,xcdraw]{xcolor}

\usepackage{amsthm}
\usepackage{float}
\usepackage{setspace}
\usepackage{amssymb}
\usepackage{stfloats}
\usepackage{cite}
\usepackage{ragged2e}
\usepackage[ruled,vlined,linesnumbered]{algorithm2e}
\usepackage{amsfonts}
\usepackage{mathrsfs}
\usepackage{amsmath,amsthm}
\usepackage{array,booktabs}
\usepackage{subfigure}
\usepackage{multirow}
\usepackage{cuted}
\usepackage{multicol}
\usepackage{graphicx}
\usepackage{subfigure}
\usepackage{graphicx,xcolor,bm}
\usepackage{hyperref}
\usepackage{threeparttable}
\usepackage{dcolumn}
\usepackage{setspace}
\usepackage{makecell}
\usepackage{lipsum}
\usepackage{enumerate}
\usepackage{mathrsfs}
\usepackage{bbm}

\usepackage[protrusion=true,expansion=true]{microtype}
\pdfoutput=1

\newtheorem{Defn}{Definition}

\newtheorem{Prop}{Proposition}

\usepackage{xfrac}

\usepackage{cite,bm}
\graphicspath{{figures/}}
\def\BibTeX{{\rm B\kern-.05em{\sc i\kern-.025em b}\kern-.08em
		T\kern-.1667em\lower.7ex\hbox{E}\kern-.125emX}}
\usepackage{balance}
\everymath{\footnotesize} 
\everydisplay{\small} 
\usepackage{nicefrac}
\allowdisplaybreaks
\begin{document}
\newcommand\semiHuge{\fontsize{21.7}{31.38}\selectfont}
	\title{\semiHuge Future Resource Bank for ISAC: Achieving Fast and Stable Win-Win Matching for Both Individuals and Coalitions	}
\author{Houyi Qi, Minghui Liwang, \IEEEmembership{Senior Member}, \IEEEmembership{IEEE}, Seyyedali Hosseinalipour, \IEEEmembership{Senior Member}, \IEEEmembership{IEEE}, \\Liqun Fu, \IEEEmembership{Senior Member}, \IEEEmembership{IEEE}, Sai Zou, \IEEEmembership{Senior Member}, \IEEEmembership{IEEE}, and Wei Ni, \IEEEmembership{Fellow}, \IEEEmembership{IEEE}
	\thanks{H. Qi (houyiqi@tongji.edu.cn) and M. Liwang (minghuiliwang@tongji.edu.cn) are with State Key Laboratory of Autonomous Intelligent Unmanned Systems, Frontiers Science Center for Intelligent Autonomous Systems, Ministry of Education, Shanghai Key Laboratory of Intelligent Autonomous Systems, Tongji University, Shanghai, China. L. Fu (liqun@xmu.edu.cn) is with School of Informatics, Xiamen University, Xiamen, China. S. Hosseinalipour (alipour@buffalo.edu) is with Department of Electrical Engineering, University at Buffalo--SUNY, NY, USA. S. Zou (dr-zousai@foxmail.com) is with College of Big Data and Information Engineering, Guizhou University, Guiyang, China. W. Ni (Wei.Ni@ieee.org) is with Data61, CSIRO, Sydney, Australia. 
    
    Corresponding author: Minghui Liwang (minghuiliwang@tongji.edu.cn)}}

	\IEEEtitleabstractindextext{
		\begin{abstract}
			\justifying
Future wireless networks must support emerging applications where environmental awareness is as critical as data transmission. Integrated Sensing and Communication (ISAC) enables this vision by allowing base stations (BSs) to allocate bandwidth and power to mobile users (MUs) for communications and cooperative sensing. However, this resource allocation is highly challenging due to: \textit{(i)} dynamic resource demands from MUs and resource supply from BSs, and \textit{(ii)} the selfishness of MUs and BSs. To address these challenges, existing solutions rely on either real-time (online) resource trading, which incurs high overhead and failures, or static long-term (offline) resource contracts, which lack flexibility. To overcome these limitations, we propose the \textit{Future Resource Bank for ISAC}, a hybrid trading framework that integrates offline and online resource allocation through a level-wise client model, where MUs and their coalitions negotiate with BSs. We introduce two mechanisms: \textit{(i)} Offline Role-Friendly Win-Win Matching (offRFW$^2$M), leveraging overbooking to establish risk-aware, stable contracts, and \textit{(ii)} Online Effective Backup Win-Win Matching (onEBW$^2$M), which dynamically reallocates unmet demand and surplus supply. We theoretically prove stability, individual rationality, and weak Pareto optimality of these mechanisms. Through comprehensive experiments, we show that our framework improves social welfare, latency, and energy efficiency compared to existing methods.
		\end{abstract}
		
}
	
\maketitle
 \IEEEdisplaynontitleabstractindextext

\IEEEpeerreviewmaketitle

\section{Introduction}
\IEEEPARstart{T}{he} rapid advancement of wireless communications and artificial intelligence (AI) is driving innovations in autonomous driving, smart manufacturing, and remote healthcare, increasing the demand for real-time environmental awareness\cite{SURVEY 1}. Integrated Sensing and Communication (ISAC) has emerged as a key enabler of these innovations, seamlessly integrating sensing, positioning, and communication to support these complex applications\cite{SURVEY 2, SURVEY 3}.
Despite its transformative potential, ISAC faces critical resource allocation challenges. A key obstacle is the need to meet diverse and stringent requirements of various applications despite limited resource availability. Also, the inherently dynamic nature of ISAC networks, driven by mobile devices and evolving service demands, calls for frequent interactions between base stations (BSs), acting as resource providers, and mobile users (MUs), serving as requestors. These continuous interactions impose substantial overhead,  complicating the delivery of communication and sensing services. 
Moreover, the selfish behavior of BSs and MUs complicates coordination, as both parties prioritize individual gains over overall network efficiency\cite{SURVEY 1}. To address these challenges, integration of a \textit{resource trading market} within ISAC networks is promising. By incorporating economic incentives, such as pricing and trading mechanisms\cite{price1,price2,price3}, this trading-based approach can foster a mutually beneficial ecosystem, where providers (sellers) maximize their revenues while requestors (buyers) optimize their service-related profits.

\subsection{Challenges and Motivations}
While market-driven ISAC networks hold promise, their practical design present several unresolved challenges. To address them, we identify key research questions that shape our approach and define the core motivations driving this work.

\noindent
$\bullet$ \textit{Research Question 1: What is Online Resource Trading in A Market-Driven Network, 	and How Can We Overcome Its Limitations?}
Online resource trading, widely adopted in existing studies \cite{Spot 1,Spot 2,Spot 3}, refers to a real-time procedure where buyers and sellers establish agreements in real-time based on current resource supply/demand and network conditions. Despite being flexible and adaptive, this approach has notable drawbacks in dynamic ISAC environments.
A key limitation is its high operational overhead \cite{future 1,future 2}, as negotiating resource quantities, pricing, and agreements consumes time and resources, diverting attention from actual service delivery. Also, fluctuating resource availability increases the risk of buyers failing to secure services despite prolonged negotiations \cite{future 1}.

To mitigate these limitations, a promising alternative is offline trading, which leverages historical/predictive data to establish long-term contracts, reducing the overhead of real-time decisions\cite{RW Matching3}. These contracts improve time efficiency by enabling instant execution when needed. However, a key challenge remains: ensuring their applicability in time-varying ISAC networks, leading to our next research question.

\noindent
$\bullet$ \textit{Research Question 2: Can Offline Trading Fit into Dynamic ISAC Networks? If Not, How Can Its Advantages Be Integrated into a Broader Framework for Greater Adaptability?}
Offline trading involves pre-signing long-term contracts ahead of actual resource allocation. However, ISAC networks are inherently dynamic and subject to uncertainties arising from mobile device movement, time-varying channel conditions, and fluctuating resource supply and demand. These factors can limit the effectiveness of pre-established/long-term contracts\cite{future 1}.
To address this, one solution is to assess key sources of uncertainty in ISAC networks, and then employ advanced risk management techniques to ensure that these contracts remain viable. 
Additionally, overbooking\footnote{Overbooking is a widely used economic strategy in airlines \cite{overbook1}, hotels \cite{overbook2}, and telecommunications \cite{overbook3} industries that addresses the limitations of booking methods in handling fluctuating resource demand and supply.}, a widely used strategy in other industries\cite{RW Matching3}, represents a promising tool in resource allocation. It allows sellers to allocate more resources than their theoretical supply, thereby improving resource availability and mitigating fluctuations in supply and demand. Subsequently, \textit{combining online and risk-controlled, overbooking-empowered offline resource allocation in ISAC networks can form a hybrid market with significant advantages.}

Nevertheless, when offline and online modes coexist, a significant challenge arises: the overlapping service demands across buyers can complicate resource allocation and create redundancies, e.g., multiple buyers may target the same sensing objective (e.g., air quality monitoring or hazard detection). Therefore, efficient coordination, whether through competition or collaboration, is crucial for optimizing resource utilization. Further complicating this issue is the lack of well-defined revenue distribution mechanisms in collaborative settings. In particular, if buyers form sensing coalitions to share resources, ensuring fair cost allocation and incentive structures becomes a critical yet complex issue. These considerations make designing a framework that balances cooperation, fairness, and efficiency while adapting to dynamic ISAC conditions an open research problem, leading to our next research question.

\noindent
$\bullet$ \textit{Research Question 3: In a Dynamic ISAC Network, Should Resource Buyers Compete for Resources Individually or Collaborate for Mutual Benefits?}
In a resource trading market over ISAC networks, maximizing utility for both buyers and sellers while ensuring efficient resource utilization is essential. Existing efforts primarily focus on individual MUs acquiring resources (e.g., bandwidth and power) from BSs to fulfill their communication and sensing demands \cite{ISAC}. However, in many scenarios, multiple MUs share the same sensing objectives, where performing separate sensing tasks results in redundancy and inefficient resource utilization. To this end, in this work, we propose \textit{sensing coalitions} allowing for MUs with identical sensing targets to form cooperative groups: instead of competing for resources individually, a coalition collectively requests a sensing service from BSs. The sensing results can then be shared among all coalition members. This cooperation benefits both MUs and BSs: MUs receive higher-quality sensing services, while BSs optimize resource utilization by handling aggregated requests rather than redundant demands, establishing a win-win paradigm for profit enhancement for both parties.

Addressing the above three key research questions serves as the core motivation of this work. In a nutshell, we study the resource trading problem in dynamic ISAC networks, where multiple BSs act as resource sellers and multiple MUs serve as buyers, each requiring two types of services, communication and sensing, both demanding bandwidth and power resources. To optimize resource allocation and service efficiency, we propose the \textit{Future Resource Bank for ISAC Networks}, a novel framework that \textit{interprets resource trading as a structured financial system with both offline and online modes}. In this model, BSs function as bank presidents, while MUs are classified into two levels of clients. For communication services, each MU acts as an individual client, directly requesting resources from the appropriate BS. For sensing services, MUs with a shared sensing objective can form a coalition, functioning as a team client that collectively requests resources. 

\subsection{Related Works}
In the following, we provide a review of related work from three interrelated domains: \textit{(i)} resource allocation in conventional wireless communication and sensor networks, \textit{(ii)}  resource allocation in ISAC networks, and \textit{(iii)} matching-based resource provisioning. We also highlight the novelty of our proposed framework and compare it with the existing literature in these domains.

\begin{table*}[b!] 
\vspace{-5mm}
	{\footnotesize
		\caption{A summary of related studies (Sta.: Stability, Ind.: Individual rationality, Fai.: Fairness, NonW.: Non-wastefulness)}  
		\begin{center}\label{table_RW}
				\begin{tabular}{|c|c|c|c|c|p{2.8cm}|c|c|c|c|}
					\hline
					\multirow{2}{*}{\textbf{Reference}} & \multicolumn{2}{c|}{\textbf{Network service type}} & \multicolumn{2}{c|}{\textbf{Trading mode}}&\multirow{2}{*}{\textbf{Resource type}}&\multirow{2}{*}{\textbf{Matching mechanism}}&\multicolumn{3}{c|}{\textbf{Key property}} \\  \cline{2-5}\cline{8-10} 
					&\makecell[c]{Commication}&\makecell[c]{Sensing}&\makecell[c]{Online}&\makecell[c]{Offline}&&&Sta.&\makecell[c]{Ind.}&Fai./NonW.\\ \hline
					\makecell[l]{\cite{TWN1}} &$\surd$& & &$\surd$&Power, time& & &$\surd$&\\ \hline
					\makecell[l]{\cite{TWN2}} &$\surd$& &$\surd$& &Power& & &$\surd$&\\ \hline
					\makecell[l]{\cite{TWN3}} &$\surd$& &$\surd$& &Spectrum, energy, encoding rate& & & &\\ \hline
					\makecell[l]{\cite{TWSN1}} & &$\surd$& &$\surd$&Transmission power, duration& & & &\\ \hline
				\makecell[l]{\cite{TWSN2}} & &$\surd$& &$\surd$&Antenna, beam direction, transmission power& & & &\\ \hline
				\makecell[l]{\cite{RW ISAC1}} &$\surd$&$\surd$&$\surd$&&Power, bandwidth& & &$\surd$&$\surd$\\ \hline
				\makecell[l]{\cite{SURVEY 3}} &$\surd$&$\surd$&$\surd$&$\surd$&Power, time, antenna& & & &\\ \hline
				\makecell[l]{\cite{RW ISAC3}} &$\surd$ &$\surd$&$\surd$ &&Computing resources, bandwidth& && $\surd$&\\ \hline
				\makecell[l]{\cite{RW ISAC4}} &$\surd$&$\surd$&$\surd$ && Power& & & &\\ \hline
				\makecell[l]{\cite{RW Matching1}} & $\surd$&& $\surd$&&Resource block&$\surd$&$\surd$& &\\ \hline
				\makecell[l]{\cite{RW Matching2}} &$\surd$&&$\surd$& &Travel resources, prices, service quality&$\surd$&$\surd$&$\surd$&$\surd$\\ \hline
				\makecell[l]{\cite{RW Matching3}} &$\surd$& &$\surd$&$\surd$&Computing resources&$\surd$&$\surd$&$\surd$&$\surd$\\ \hline
				\makecell[l]{\cite{RW Matching4}, \cite{RW Matching5}} & $\surd$&& $\surd$&&Computing resources&$\surd$&$\surd$& $\surd$&\\ \hline
				\makecell[l]{Our work} &$\surd$&$\surd$&$\surd$&$\surd$&Power, bandwidth&$\surd$ &$\surd$&$\surd$&$\surd$\\ \hline
			\end{tabular}
	\end{center}}
\end{table*}           
{\subsubsection{Resource Allocation in Conventional Wireless Communication and Sensor Networks}

Resource allocation has been a cornerstone in both traditional wireless communication systems and wireless sensor networks, serving as a foundational pillar for optimizing communication and sensing services under diverse network operational constraints. \textit{Focusing on wireless communication systems}, existing studies have primarily focused on throughput enhancement, interference mitigation, and energy/time efficiency improvement. For example, \textit{Yao et al.} \cite{TWN1} proposed a heterogeneous composite fading channel model tailored to 5G-drone emergency wireless communications, and derived closed-form solutions for power allocation and bandwidth optimization to enhance the network throughput and energy efficiency in post-disaster scenarios. \textit{Lee et al.} \cite{TWN2} developed a deep learning-based resource allocation framework for device-to-device communication in multi-channel cellular networks, achieving near-optimal spectral efficiency via joint channel assignment and power control. \textit{Zhang et al.} \cite{TWN3} explored a semantic-aware resource allocation scheme for energy harvesting cognitive radio networks, employing deep reinforcement learning to jointly optimize transmit power, time slot allocation, and semantic compression, thereby maximizing long-term quality of experience in task-oriented semantic communications.
\textit{Focusing on wireless sensor networks}, the research emphasis has been on energy-efficient sensing, prolonged network operation, and high-precision target tracking. For instance, \textit{Huang et al.} \cite{TWSN1} proposed a joint polarization mode, beamforming, and power allocation strategy for radar systems, integrating radar cross-section prediction with posterior Cramér-Rao lower bound (PCRLB)-based optimization to improve multi-target tracking accuracy.  \textit{Yang et al.} \cite{TWSN2} tackled joint antenna selection and transmit–receive beamforming optimization under signal interference, designing a block coordinate descent algorithm to minimize antenna usage while maintaining tracking performance based on PCRLB.

Although the above efforts have contributed valuable insights into resource allocation/management for either \textit{communication} or \textit{sensing}, they generally treat these two functionalities in isolation and adopt decoupled optimization strategies. Nevertheless, ISAC networks necessitate the joint allocation of resources such as bandwidth and power to simultaneously support communication and sensing functionalities/services. This dual-service integration results in tightly coupled resource demands, more complex performance requirements, and novel resource management challenges that have led to a new wave of research dedicated to resource allocation in ISAC networks.
}

\subsubsection{Resource Allocation in ISAC Networks} {To accommodate the unique features of ISAC networks, recent studies have explored resource allocation strategies that concurrently address the requirements of both communication and sensing services \cite{RW ISAC1,SURVEY 3,RW ISAC3,RW ISAC4}}. For instance,  \textit{Li et al.} \cite{RW ISAC1} introduced a value-of-service-oriented resource allocation scheme for multi-MU collaborative ISAC networks, facilitating concurrent heterogeneous service provisioning. \textit{Dong et al.} \cite{SURVEY 3} developed an ISAC resource allocation framework, integrating sensing quality of service (QoS) to optimize resources for diverse applications. \textit{Hu et al.}\cite{RW ISAC3} proposed an ISAC-aided edge computing framework, addressing the rapid proliferation of vehicles and the growing demand for integrating vehicular communications within computing services. \textit{Wang et al.}\cite{RW ISAC4} explored a resource-optimized ISAC architecture for green ad-hoc networks, examining the intricate interplay between environmental sensing and data transmission performance.

Despite offering valuable contributions, the above efforts have primarily relied on onsite (i.e., online) decision-making approaches for resource allocation, which often lead to extended delays and potential service disruptions. Subsequently, our work introduces a hybrid resource trading framework that integrates \textit{offline} and \textit{online} resource trading modes, leveraging \textit{(i)} long-term contracts to reduce real-time decision-making overhead and \textit{(ii)} online trading as a backup mechanism to handle dynamic fluctuations in resource demand and supply. 

\subsubsection{Matching-Based Resource Provisioning}
This work develops one of the first efficient matching/assignment mechanisms between clients (MUs and coalitions) and BSs to ensure responsive and cost-effective resource provisioning in ISAC networks. Similar methods have gained attention in other domains such as edge computing and Internet of Things (IoT), where matching-driven resource allocation is used for balancing diverse demands with available resource supplies \cite{RW Matching1,RW Matching2,RW Matching3,RW Matching4,RW Matching5}. \textit{Ye et al.} \cite{RW Matching1} combined deep reinforcement learning with stable matching to enable adaptive resource allocation in mobile crowdsensing. \textit{Xu et al.} \cite{RW Matching2} introduced a three-sided stable matching with an optimal pricing scheme for distributed vehicular networks. \textit{Qi et al.} \cite{RW Matching3} explored cross-layer pre-matching mechanisms to achieve cost-effective resource trading in cloud-aided edge networks. \textit{Sharghivand et al.} \cite{RW Matching4} considered QoS in terms of service response time and proposed a matching model between cloudlets and IoT applications. \textit{Du et al.} \cite{RW Matching5} developed a matching-based approach for computing resource management in small-cell networks, optimizing resource allocation and service pricing.

{Although the above prior works have made notable advancements, they  are designed to handle isolated resources (e.g., bandwidth, computation, or service-specific capabilities) across participants in a one-to-one fashion. Nevertheless, in emerging ISAC networks, communication and sensing services coexist and compete for the limited physical resources, which fundamentally changes the matching landscape as follows. \textit{(i) Heterogeneous resource coupling:} Achieving both reliable data transmission and high-accuracy sensing requires the joint allocation of bandwidth and power. \textit{(ii) Bidirectional incentives:} For BSs, allocating extra power to improve sensing may sacrifice spectral efficiency and reduce communication performance/revenue; for MUs (or their coalitions), higher sensing fidelity often implies paying more for bandwidth–power bundles. These cross-coupled incentives make it difficult to find a stable matching. \textit{(iii) Stringent trade-offs:} The dual-service provisioning in ISAC networks  forces tighter latency and reliability constraints than in conventional communication and sensor networks: any feasible matching must balance multi-objective utilities while accounting for multi-dimensional risks faced by both MUs and BSs.}
{These intertwined factors complicate the resource-to-user matching problem in ISAC networks, rendering traditional single-resource matching algorithms inadequate and motivating the coordinated and risk-aware matching mechanisms developed in this paper.} To further clarify our contributions,  we provide a comparison between this study and a representative set of  existing studies in Table~\ref{table_RW}.

{In summary, by incorporating \textit{(i)} joint optimization of communication and sensing services, \textit{(ii)} stable demand-resource matching for coupled resources across these two services, and \textit{(iii)} a hybrid offline-online resource-trading mechanism, we pioneer a low-overhead and economically viable resource allocation framework for ISAC networks.}

\subsection{Overview and Summary of Contributions}

The core principle of our proposed approach in this work is to develop one of the first \textit{time-efficient} and \textit{mutually beneficial} matching mechanisms between BSs and both individual and coalition clients in dynamic ISAC networks. In this regard, our key contributions can be summarized as follows:

\noindent
~$\bullet$ To enable efficient service delivery in dynamic ISAC networks, we introduce the \textit{Future Resource Bank for ISAC}, a novel framework that integrates both offline and online trading modes to balance long-term stability with real-time adaptability. 
Further, unlike existing approaches that treat MUs as independent requestors, our framework incorporates \textit{both} individual and coalition-based resource allocation.

\noindent
~$\bullet$ Our proposed offline trading mode enables BSs and clients to \textit{pre-sign risk-aware long-term contracts} for both communication and sensing services, reducing real-time decision-making overhead. Further, in this mode, we introduce the concept of \textit{overbooking} to the ISAC literature, designed to address the uncertainty in MU resource demands. This concept allows BSs to allocate more resources than their theoretical supply, effectively mitigating fluctuations in resource demand.

\noindent
~$\bullet$
To establish feasible and mutually beneficial long-term  contracts between BSs and clients, we propose the \textit{Role-Friendly Win-Win Matching for Offline Trading} (offRFW$^2$M). We further prove that it satisfies key properties of \textit{stability}, \textit{individual rationality}, and \textit{weak Pareto optimality}, ensuring sustainability, incentivization, and efficiency in ISAC networks.

\noindent
~$\bullet$ Noting that the practical implementation of long-term contracts may suffer from performance degradations due to the intermittent participation of MUs and unpredictable resource demands, we introduce an online trading mode as a \textit{complementary backup mechanism to ensure resource allocation adaptability to the real-time ISAC network conditions}. 

\noindent
~$\bullet$ For the online trading mode, we first 
design a greedy-based algorithm that selects \textit{volunteer clients} willing to forgo their services in exchange for compensation, balancing excessive real-time demand without disrupting overall network performance.
We then enable individual MUs or coalitions with unmet demands to compete for available resources at BSs with surplus supply in real-time through proposing the Effective Backup Win-Win Matching for Online Trading (onEBW$^2$M).
We theoretically prove that onEBW$^2$M retains key properties of stability, individual rationality, and weak Pareto optimality.

\noindent
~$\bullet$ We conduct extensive experiments using both synthetic simulations and real-world datasets. The results demonstrate the superior performance of our method across key evaluation metrics, including social welfare, time efficiency, and energy efficiency compared to the existing baselines.

\noindent\textbf{\textit{Notations:}} $(\cdot)^\mathsf{(c)}$ and $(\cdot)^\mathsf{(s)}$ represent super-scripts related to communication services and sensing services, respectively;
$(\cdot)^\top$ and $(\cdot)^\mathsf{H}$ describe transpose and Hermitian, respectively. $j = \sqrt{-1}$; $\|\mathbf{a}\|$ is the Euclidean 2-norm of vector $\mathbf{a}$; $\{\cdot\}^\text{Re}$ is the real component of complex variable. $\operatorname{tr}(\mathbf{A})$ denotes the trace of matrix $\mathbf{A}$ and $\operatorname{diag}\{\mathbf{a}\}$ indicates the diagonal matrix built upon the elements of vector $\mathbf{a}$.

\section{Overview and Key System Components}
\noindent In this section, we first offer an overview of our methodology in Sec.~\ref{sec:over}. Then, we present the modeling of MUs and BSs in
Sec.~\ref{sec:baseModel}, and finally introduce the ISAC model in
Sec.~\ref{sec:modelISAC}.

\begin{figure*}[]
	\centering
	\setlength{\abovecaptionskip}{-0.5 mm}
	\includegraphics[width=2\columnwidth]{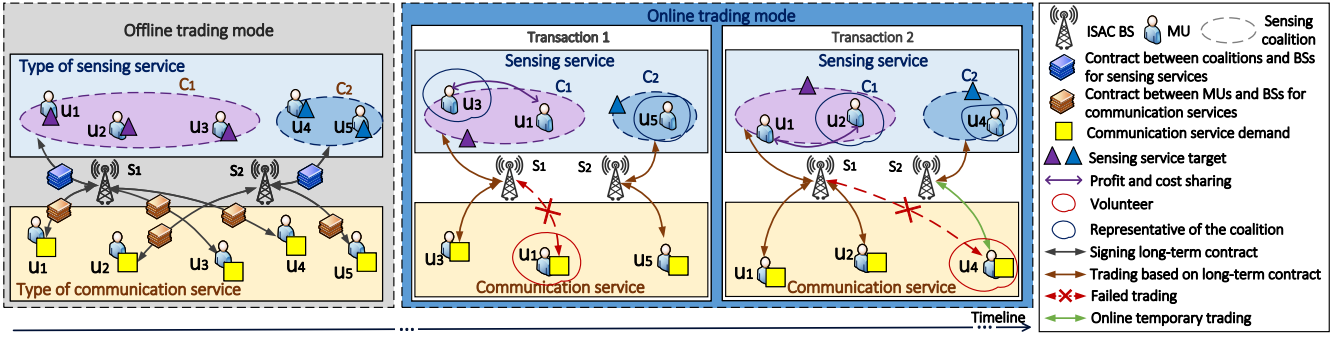}
	\caption{Schematic of our future resource bank for ISAC networks, illustrating two types of services (communication and sensing) and two levels of clients (individual MUs and sensing coalitions).}
	\label{SYSTEM MODEL}
	\vspace{-0.3cm}
\end{figure*}
\subsection{Overview of Our Methodology}\label{sec:over}
We consider a dynamic ISAC network, where resource provisioning involves two key parties: \textit{(i)} multiple MUs, denoted by $\bm{\mathcal{U}} = \{u_1, \dots, u_{\bm{|\mathcal{U}|}}\}$, each requiring both sensing and communication services, and \textit{(ii)} multiple BSs, denoted by $\bm{\mathcal{S}} = \{s_1, \dots, s_{\bm{|\mathcal{S}|}}\}$, which provide bandwidth and power resources to MUs.
We further define $\bm{\mathcal{Q}}$ as the set of sensing targets/objectives. When multiple MUs share the same sensing objective, they can form coalitions for cooperative sensing, represented as $\bm{\mathcal{C}} = \{\bm{c}_1, \dots, \bm{c}_{\bm{|\mathcal{C}|}}\}$ (each sensing coalition may contain one MU or more). Our \textit{future resource bank for ISAC} classifies clients into two levels: for communication services, each MU acts as an \textit{individual client}, directly requesting resources from BSs; for sensing services, MUs with common sensing targets form a coalition, functioning as a \textit{team client} that collectively requests resources, reducing redundancy. {We presume that each sensing coalition is served by a \textit{single} BS. This assumption has also been adopted in existing ISAC and distributed sensing literature (e.g., \cite{SURVEY 3,RW ISAC1}), as it enables tractable problem formulations and facilitates the design of resource allocation mechanisms\footnote{While cooperative sensing involving multiple BSs, such as through joint beamforming or signal fusion, has the potential to improve sensing accuracy and robustness, it also increases system complexity in terms of coordination overhead, utility modeling, and contract formation. Integrating such cooperative paradigms into our Future Resource Bank framework represents a promising direction, which we plan to explore in future work.}.}

Our future resource bank for ISAC further integrates both \textit{offline} and \textit{online} trading modes\footnote{In this work, we employ a two-stage, hybrid resource allocation framework that is \textit{distributed} in nature (i.e., there is no need for a centralized coordinator/controller as the decisions are made locally at the BSs and MUs). Thus, compared to centralized control architectures, the proposed decentralized approach offers several critical advantages:
\textit{(i)} Reduced signaling overhead and latency, as it obviates the need for continuous and real-time collection of global system states; \textit{(ii)} Superior scalability, by avoiding the exponential computational burden inherent in solving centralized NP-hard optimization problems, particularly those involving discrete decision variables; \textit{(iii)} Improved system robustness, as the elimination of a single point of failure ensures that faults are contained locally, without cascading effects on the global system; and
\textit{(iv)} Enhanced privacy preservation, since sensitive data are processed and retained locally at the network edge, thereby minimizing exposure to centralized entities.}. The offline mode allows individual MUs and coalitions to establish long-term contracts with BSs for future communication and sensing services, securing resource commitments in advance of actual demand requests. In particular, for communication services, a contract between an individual MU $u_i\in\bm{\mathcal{U}}$ and a BS $s_j\in\bm{\mathcal{S}}$ is represented as {\footnotesize $\mathbb{C}^\mathsf{(c)}_{i,j} = \{\mathbbm{c}^\mathsf{(c),B}_{i,j}, \mathbbm{c}^\mathsf{(c),Pow}_{i,j}, \mathbbm{c}^\mathsf{(c),Pay}_{i,j}, \mathbbm{c}^\mathsf{(c),PelU}_{i,j}, \mathbbm{c}^\mathsf{(c),PelS}_{i,j}\}$}, where {\footnotesize$\mathbbm{c}^\mathsf{(c),B}_{i,j}$} and {\footnotesize$\mathbbm{c}^\mathsf{(c),Pow}_{i,j}$} denote the allocated bandwidth and power, {\footnotesize$\mathbbm{c}^\mathsf{(c),Pay}_{i,j}$} represents the unit price, and {\footnotesize$\mathbbm{c}^\mathsf{(c),PelU}_{i,j}$} and {\footnotesize$\mathbbm{c}^\mathsf{(c),PelS}_{i,j}$} are clauses specifying penalties if the MU or BS breaches the contract.
Similarly, for sensing services, a long-term contract between a sensing coalition $\bm{c}_k\in \bm{\mathcal{C}}$ and a BS $s_j$ as {\footnotesize$\mathbb{C}^\mathsf{(s)}_{k,j} = \{\mathbbm{c}^\mathsf{(s),B}_{k,j}, \mathbbm{c}^\mathsf{(s),Pow}_{k,j}, \mathbbm{c}^\mathsf{(s),Pay}_{k,j}, \mathbbm{c}^\mathsf{(s),PelU}_{k,j}, \mathbbm{c}^\mathsf{(s),PelS}_{k,j}\}$}, where {\footnotesize$\mathbbm{c}^\mathsf{(s),B}_{k,j}$} and {\footnotesize$\mathbbm{c}^\mathsf{(s),Pow}_{k,j}$} specify the bandwidth and power demands, {\footnotesize$\mathbbm{c}^\mathsf{(s),Pay}_{k,j}$} represents the unit price, and {\footnotesize$\mathbbm{c}^\mathsf{(s),PelU}_{k,j}$} and {\footnotesize$\mathbbm{c}^\mathsf{(s),PelS}_{k,j}$} define penalty clauses for contract violations. We presume that MUs within a coalition 
{\footnotesize$\bm{c}_k$} share the same pricing scheme, where the cost per MU is determined as: {\footnotesize$p^\mathsf{(s)}_{i,j} = \mathbbm{c}^\mathsf{(s),Pay}_{k,j} / |\bm{c}_k|$}, where {\footnotesize$|\bm{c}_k|$} is the number of MUs within the coalition. {To simplify intra-coalition coordination, we assume that only one representative MU will executes the sensing task on behalf of the coalition, and the obtained sensing utility can be shared across all coalition members, where the payment is evenly divided. Such a consideration provides a practical abstraction of cooperative sensing behavior and ensures low coordination overhead, as supported by existing works such as \cite{RW ISAC1}.} These long-term contracts are directly executed during practical resource transactions\footnote{A practical transaction is a resource trading event between clients and BSs.}.
 {To capture the real-world decision-making behavior of participants in ISAC networks, we assume that both clients and BSs behave \textit{selfishly} (i.e., each client/BS only accepts to be matched with a BS/client if that improves its own utility). For example, a BS rejects contracts that offer low payments although such contracts may lead to a better collective utility for the BSs. This mutual self-interest underpins our win-win matching design, ensuring that every contract is individually rational and voluntarily accepted by both parties.}

Since long-term contracts may experience performance degradation during execution, we then introduce an online trading mode as a backup mechanism. This mode addresses two key scenarios: when resource demand at a BS exceeds its available supply (due to overbooking as we will  formalize in Sec.~\ref{sec:probForm}), certain clients (called \textit{volunteers}) voluntarily relinquish their commitments in exchange for compensation. Additionally, for MUs with unmet demands, online trading enables temporary contracts with BSs that have surplus resources.
{In summary, to enable scalable and role-aware association between BSs and clients with diverse service needs, we design a two-stage matching mechanism.  In the offline stage, a proactive offRFW$^2$M is employed to assign both individual MUs and sensing coalitions to BSs based on long-term role-specific utility.  In the online stage, the onEBW$^2$M ensures responsive resource reallocation for residual demands and overbooked commitments.  This design achieves seamless and fast service delivery in dynamic ISAC networks.}

Fig. \ref{SYSTEM MODEL} illustrates the timeline of our methodology. The timeline is divided into two segments: \textit{(i)} offline trading, where long-term contracts are established before practical demands arise, and \textit{(ii)} online trading, which manages practical transactions. Also, MUs are categorized into two client levels: \textit{(i)} individual MUs request communication services directly from BSs, and \textit{(ii)} sensing coalitions, formed by MUs with shared sensing objectives, request sensing services collectively. {Note that a single MU may act as a dual-role participant, simultaneously requesting communication services as an individual, while also participating in sensing coalitions.}
In the offline trading mode (gray box in Fig. \ref{SYSTEM MODEL}), long-term contracts are signed between individual MUs, sensing coalitions, and BSs. During online trading (blue box), these contracts are executed, while MUs with unmet demands negotiate temporary contracts with BSs that have surplus resources.
The figure illustrates a scenario with two BSs (\(s_1, s_2\)) and five MUs (\(u_1\) to \(u_5\)). Initially, MUs form sensing coalitions based on shared sensing objectives, where \(\bm{c}_1 = \{u_1, u_2, u_3\}\) and \(\bm{c}_2 = \{u_4, u_5\}\). Each coalition participates in the resource trading market as a single entity, while individual MUs also negotiate contracts for communication services. Using our later offline matching strategy, called offRFW$^2$M, both individual MUs and coalitions establish long-term contracts with BSs, securing bandwidth and power allocations in advance. 
During practical Transaction 1, \(s_1\) selects \(u_1\) as a volunteer, relinquishing its reserved service to balance demand and supply. In practical Transaction 2, \(s_1\) selects \(u_4\) as a volunteer, since total demand exceeds its available resources. This highlights the effectiveness of overbooking, which enables BSs to handle resource demand fluctuations efficiently. For instance, in Transaction 1, when \(u_2\) and \(u_4\) are absent, \(s_1\) still fulfills contracts with \(u_3\) and coalition \(\bm{c}_1\), while \(s_2\) serves \(u_5\) and \(\bm{c}_2\), preventing resource waste. 
Nevertheless, dynamics of client participation can still lead to surplus resources. For example, in Transaction 2, \(s_1\) is unable to fulfill its contract with \(u_4\) due to an imbalance in supply and demand. To resolve this, our later online matching strategy, called onEBW$^2$M, efficiently reassigns resources, ensuring optimal resource utilization. For example, in Transaction 2, \(s_2\) reallocates communication resources to serve \(u_4\).

\subsection{Modeling of MUs and BSs}\label{sec:baseModel}
\subsubsection{Modeling of MUs} We model each MU $u_i \in \bm{\mathcal{U}}$ as a 6-tuple $\{Q_i, R_i^\mathsf{req}, S_i^\mathsf{req}, N_i^\mathsf{R}, l_i^\mathsf{U},\alpha_i\}$, where $Q_i\in\mathcal{Q}$ represents its sensing target, $R_i^\mathsf{req}$ and $S_i^\mathsf{req}$ denote its requirement on data rate (bits/s) and sensing accuracy, respectively, and $N_i^\mathsf{R}$ represents the number of receiving antennas on $u_i$, while $l_i^\mathsf{U} = [x_i^\mathsf{U}, y_i^\mathsf{U}]^\top$ denotes the $(x,y)$ location of $u_i$. To reflect the dynamic nature of ISAC networks, we model MU participation uncertainty using a Bernoulli random variable $\alpha_i$\footnote{In this work, the dynamic nature of resource demand is modeled through the uncertainty of user participation. This uncertainty reflects whether an MU is willing/able to engage in a given sensing or communication service at practical transactions. It encompasses a range of real-world factors, including device-level constraints (e.g., battery depletion or high processor load), network instability, and behavioral preferences (e.g., low perceived utility). }, $\alpha_i \sim {\bf{B}} \Big((1, 0), (\mathbbm{a}_i, 1-\mathbbm{a}_i)\Big)$, where $\alpha_i=1$ which occurs with probability $\mathbbm{a}_i$ implies that the MU joins the market during a practical transaction; otherwise $\alpha_i=0$. MUs are formally categorized as follows.
\begin{Defn}(Individual MU) An individual MU  participates independently in resource trading for communication services.
\end{Defn}
\begin{Defn}(Sensing Coalition) MUs with a shared sensing target form a coalition, where each coalition acts as a team client in resource trading. A sensing coalition is defined through a two-way mapping/matching function $\mu$ between the MU set $ \bm{\mathcal{U}} $ and coalition set $ \bm{\mathcal{C}} $, which satisfies the following properties:
	
	\noindent
	$\bullet$ For each MU $ u_i \in \bm{\mathcal{U}}, \mu(u_i) \subseteq \bm{\mathcal{C}} $, and $|\mu(u_i)|=1$, meaning each MU belongs to exactly one coalition.
	
\noindent 
$\bullet$ For each coalition $\bm{c}_k \in \bm{\mathcal{C}}$, its assigned MUs satisfy $\mu(\bm{c}_k) \subseteq \bm{\mathcal{U}}$, and every MU $u_i \in \bm{c}_k$ shares the same sensing target $Q_i$.
	
	\noindent
	$\bullet$ An MU \( u_i \) belongs to coalition \( \bm{c}_k \) if and only if \( \bm{c}_k \) includes \( u_i \), i.e., \( u_i \in \mu(\bm{c}_k) \) if and only if \( \bm{c}_k \in \mu(u_i) \).
\end{Defn}

\subsubsection{Modeling of BSs} We model each BS $s_j \in \bm{\mathcal{S}}$ as a 4-tuple $\{B_j, P_j, l_j^\mathsf{S}, N_j^\mathsf{T}\}$, where $B_j$ and $P_j$ refer to the total bandwidth and power resources of $s_j$, respectively, $l_i^\mathsf{S} = [x_i^\mathsf{S}, y_i^\mathsf{S}]^\top$ represents the geographic location of $s_j$, and $N_j^\mathsf{T}$ indicates the number of transmitting antennas of BS $s_j$. We assume that both transmitting (Tx) and receiving (Rx) antennas are implemented using Uniform Linear Arrays (ULA) \cite{RW ISAC1}.

\subsection{Modeling of ISAC}\label{sec:modelISAC}
\subsubsection{Communication Model}
We assume that the total bandwidth of BS \( s_j \) is divided into \( N^\mathsf{B}_j \) sub-channels, represented as $\mathcal{B}_j = \{1, \dots, N^\mathsf{B}_j\}$, where each sub-channel has an equal bandwidth of \( B_0 \). In the communication process, when the \( n \)-th sub-channel (\( 1 \leq n \leq N^\mathsf{B}_j \)) of BS \( s_j \) is allocated to MU \( u_i \), the downlink transmitted symbols are denoted as \( s_{i,j,n} = [s_{i,j,n,1}, s_{i,j,n,2}, \dots, s_{i,j,n,N^\star}]^\top \in \mathbb{C}^{N^\star \times 1} \) where \( N^\star \) represents the number of transmitted symbols. Subsequently, the beamforming matrix is  {\footnotesize \( F_{i,j,n} = F^\mathsf{RF}_{i,j,n} F^\mathsf{BB}_{i,j,n} \in \mathbb{C}^{N_j^\mathsf{T} \times N^\star} \)} where {\footnotesize\( F^\mathsf{RF}_{i,j,n} \in \mathbb{C}^{N_j^\mathsf{T} \times N^\mathsf{RF}_i} \)} is the analog precoding matrix, and {\footnotesize\( F^\mathsf{BB}_{i,j,n} \in \mathbb{C}^{N^\mathsf{RF}_j \times N^\star} \)} is the digital beamforming matrix. Here, {\footnotesize\( N^\mathsf{RF}_j \)} denotes the number of radio frequency (RF) chains at the transmitter of BS \( s_j \). {We assume that inaccurate location information of users are available at BS, which can be used for beamforming schemes\cite{beamforming 1,beamforming 2,beamforming 3}.}
For the communication channel, we assume \( L \) distinct propagation paths between BSs and MUs, indexed by $\ell \in\{0, 1, 2, \dots, L - 1\} $ where \( \ell = 0 \) represents the Line-of-Sight (LoS) path, and the remaining paths correspond to Non-LoS (NLoS) components, originating from single-bounce reflections\cite{NLoS}. Consequently, the MU-BS channel matrix \( \mathbf{H}_{i,j,n} \) of size \( N_i^\mathsf{R} \times N_j^\mathsf{T} \) can be expressed as 
\begin{equation}\label{equ.1}
	\mathbf{H}_{i,j,n} = \mathbf{A}_{i,j}^\mathsf{r} \boldsymbol{\Gamma}_{i,j,n} \left(\mathbf{A}_{i,j}^\mathsf{t}\right)^\mathsf{H},
\end{equation}
where
\begin{equation}\label{equ.2}
	\begin{aligned}
		&\boldsymbol{\Gamma}_{i,j,n} = \sqrt{N_i^\mathsf{R} N_j^\mathsf{T} p_{i,j,n}} \, \text{diag} \Bigg\{ \frac{h_{i,j,n,0}}{\sqrt{\rho_{i,j,n,0}}} e^{-j \frac{2\pi n \tau_{i,j,0}}{N^\mathsf{B}_j T_s}}, \dots,\\& \frac{h_{i,j,n,\ell}}{\sqrt{\rho_{i,j,n,\ell}}} e^{-j \frac{2\pi n \tau_{i,j,\ell}}{N^\mathsf{B}_j T_s}}, \dots,\frac{h_{i,j,n,L-1}}{\sqrt{\rho_{i,j,n,L-1}}} e^{-j \frac{2\pi n \tau_{i,j,L-1}}{N^\mathsf{B}_j T_s}} \Bigg\}.
	\end{aligned}
\end{equation}
Here, $p_{i,j,n}$ is the power allocated to MU \( u_i \) on sub-channel \( n \) of BS \( s_j \), while \( h_{i,j,n,\ell} \) and \( \rho_{i,j,n,\ell} \) denote the channel gain and path loss for the \( \ell \)-th path, respectively. The sampling period is represented by \( T_s \), and \(\tau_{i,j,\ell} \) is the time delay of the \( \ell \)-th path. Also, the steering vector \( \mathbf{A}^\mathsf{t}_{i,j} \) and response vector \( \mathbf{A}^\mathsf{r}_{i,j} \) for a ULA are given by:
\begin{equation}\label{AOD}
	     \hspace{-8mm}
\resizebox{0.465\textwidth}{!}{$ \mathbf{A}_{i,j}^\mathsf{t}(\theta_{i,j}) {=} \frac{1}{\sqrt{N_j^\mathsf{T}}} 
	\begin{bmatrix}
		e^{-j \frac{N_j^\mathsf{T} - 1}{2} \frac{2\pi}{\lambda_n} d \sin(\theta_{i,j})} ,\cdots,
		e^{j \frac{N_j^\mathsf{T} - 1}{2} \frac{2\pi}{\lambda_n} d \sin(\theta_{i,j})}
	\end{bmatrix}^\top\hspace{-3mm},
    $}\hspace{-5mm}
\end{equation}
\begin{equation}\label{AOA}
\hspace{-8mm}
\resizebox{0.465\textwidth}{!}{$ 
	 \mathbf{A}_{i,j}^\mathsf{r}(\theta_{i,j}) {=} \frac{1}{\sqrt{N_i^\mathsf{R}}} 
	\begin{bmatrix}
		e^{-j \frac{N_i^\mathsf{R} - 1}{2} \frac{2\pi}{\lambda_n} d \sin(\theta_{i,j})} ,\cdots,
		e^{j \frac{N_i^\mathsf{R} - 1}{2} \frac{2\pi}{\lambda_n} d \sin(\theta_{i,j})}
	\end{bmatrix}^\top\hspace{-3mm},
    $}\hspace{-5mm}
\end{equation}
where \( \theta_{i,j} \) is the angle of departure (AoD), \( d \) is the distance between antenna elements (\( d = \lambda_n / 2 \)), and \( \lambda_n \) is the signal wavelength. Assuming BS \( s_j \) and MU \( u_i \) are synchronized, they can compensate for the Doppler effect and time delay. Subsequently, the received communication signal is given by
\begin{equation}
	\mathbf{y}_{i,j,n}^\mathsf{(c)} = \mathbf{W}_{i,j,n}^\mathsf{H} \mathbf{H}_{i,j,n} \mathbf{F}_{i,j,n} \mathbf{s}_n + \mathbf{W}_{i,j,n}^\mathsf{H} \mathbf{n}_{i,j,n},
\end{equation}
where \(\mathbf{W}_{i,j,n} \) is the receiver combiner, and \( \mathbf{n}_{i,j,n} \) is the zero-mean additive white Gaussian noise (AWGN) with power \( \sigma_c^2 \). 

Using Shannon’s capacity formula, the achievable data rate for MU \( u_i \) on sub-channel \( n \) of BS \( s_j \) is thus given by
\begin{equation}
	R_{i,j,n} = B_0 \log_2 \left( 1 + \text{SNR}_{i,j,n} \right),
\end{equation}
{which is obtained  under the assumption of orthogonal (i.e., non-overlapping) sub-channel allocation, where each sub-channel is assigned to at most a single MU or coalition for either communication or sensing,  avoiding both inter-user and inter-functionality interference\cite{SURVEY 3,RW ISAC1}.} Correspondingly, the signal-to-noise ratio (SNR) is defined as 
	\begin{equation}
	    \text{SNR}_{i,j,n} = p_{i,j,n} \xi_{i,j,n}
	\end{equation}
with the effective channel gain-to-noise ratio\footnote{We assume that each BS is equipped with a hybrid analog-digital beamforming architecture. Instead of explicitly optimizing the precoding matrices, we follow prior ISAC literature~\cite{RW ISAC1} by leveraging coarse user location information (e.g., estimated angle-of-departure) to construct analog steering vectors $\mathbf{A}_{i,j}^\mathsf{t}$ and $\mathbf{A}_{i,j}^\mathsf{r}$. The resulting beamforming gain is abstracted into the effective channel-to-noise ratio term $\xi_{i,j}$, which directly influences the communication data rate. This abstraction enables compact modeling while aligning with practical ISAC assumptions.}
 of
\begin{equation}
	\xi_{i,j,n} = {\left| \mathbf{W}_{i,j,n}^\mathsf{H} \mathbf{A}_{i,j}^\mathsf{r} \boldsymbol{\Gamma}_{i,j,n} \left(\mathbf{A}_{i,j}^\mathsf{t}\right)^\mathsf{H} \mathbf{F}_{i,j,n} \right|^2}\Big/{\sigma_c^2}.
\end{equation}
Subsequently, the overall achievable data rate between MU \( u_i \) and BS \( s_j \) is given by
\begin{equation}\label{equ. 9}
	R_{i,j}^\mathsf{(c)} = \sum_{n=1}^{N^\mathsf{B}_j} a_{i,j,n} R_{i,j,n},
\end{equation}
{where \( a_{i,j,n} \in \{0,1\} \) is a binary variable indicating whether sub-channel \( n \) of BS \( s_j \) is assigned to MU \( u_i \). In particular, \( a_{i,j,n} = 1 \) means the assignment is made, whereas \( a_{i,j,n} = 0 \) implies otherwise. To avoid inter-user interference, we enforce the following constraint to ensure that each sub-channel is assigned to at most one MU: $\sum_{u_i \in \mathcal{U}} a_{i,j,n} \leq 1, \quad \forall s_j \in \bm{\mathcal{S}},\ \forall n \in \{1, \dots, N^\mathsf{B}_j\}$.}
Let \( \mathcal{N}^\mathsf{(c)}_{i,j}=\{n| a_{i,j,n}=1 \} \) with size ${N}^\mathsf{(c)}_{i,j}=|\mathcal{N}^\mathsf{(c)}_{i,j}|$ denote the set of allocated subchannels to MU \( u_i \) from BS \( s_j \). Also, let \(P^\mathsf{(c)}_{i,j} \)
denote the
total power allocated by BS \( s_j \) to MU \( u_i \), with per-sub-channel power allocation of \( {P^\mathsf{(c)}_{i,j}}/{N^\mathsf{(c)}_{i,j}} \). Assuming identical channel gains and noise powers across sub-channels \cite{RW ISAC1}, the achievable data rate in (\ref{equ. 9}) simplifies to
\begin{equation}\label{equ. 10}
	R_{i,j}^\mathsf{(c)} = N_{i,j}^\mathsf{(c)} B_0 \log_2 \left( 1 + {P_{i,j}^\mathsf{(c)} \xi_{i,j}}\Big/{N_{i,j}^\mathsf{(c)}} \right),
\end{equation}
where \( \xi_{i,j}=\xi_{i,j,n}~\forall n, \) is the channel gain-to-noise power ratio across all sub-channels allocated to MU \( u_i \). Defining \( B^\mathsf{(c)}_{i,j} = N^\mathsf{(c)}_{i,j} B_0 \), the achievable data rate can be rewritten as
\begin{equation}\label{equ. 11}
	R_{i,j}^\mathsf{(c)} = B_{i,j}^\mathsf{(c)} \log_2 \left( 1 + {B_0 P_{i,j}^\mathsf{(c)} \xi_{i,j}}\Big/{B_{i,j}^\mathsf{(c)}} \right).
\end{equation}

Finally, we define the \textit{value} that MU \( u_i \) obtains from utilizing the communication service provided by BS \( s_j \) as
\begin{equation}\label{comm value}
	V_{i,j}^\mathsf{(c)} =\omega_1 R_{i,j}^\mathsf{(c)}= \omega_1 B_{i,j}^\mathsf{(c)} \log_2 \left( 1 + {B_0 P_{i,j}^\mathsf{(c)} \xi_{i,j}}\Big/{B_{i,j}^\mathsf{(c)}} \right),
\end{equation}
where \( \omega_1>0 \) is a weighting coefficient, scaling the contribution of data rate \(R^\mathsf{(c)}_{i,j} \) to the perceived service value.
\begin{figure*}
\begin{equation}\label{equ. 21}
 \hspace{-3mm}
\resizebox{0.975\textwidth}{!}{$
\begin{aligned}
    &f(\mathbf{y}_{i,j}^\mathsf{\mathsf{(s)}} | \boldsymbol{\eta}_{i,j}) = \frac{1}{(2\pi \sigma_s^2)^{N^\mathsf{(s)}_{i,j}}} \exp \left\{ \sum_{n=\mathbbm{l}_{i,j}^\mathsf{\mathsf{(s)}}}^{\mathbbm{l}_{i,j}^\mathsf{\mathsf{(s)}}+N^\mathsf{(s)}_{i,j}} \left[ \frac{1}{2\sigma_s^2} \left( \mathbf{H}_{i,j,n}^\mathsf{(s)} \mathbf{F}_{i,j,n}\mathbf{s}_\text{ref} \right)^\mathsf{H} \mathbf{y}_{i,j}^\mathsf{\mathsf{(s)}} - \frac{1}{\sigma_s^2} \left| \mathbf{H}_{i,j,n}^\mathsf{(s)} \mathbf{F}_{i,j,n}\mathbf{s}_\text{ref} \right|^2 \right]^\mathsf{Re} \right\},~
\mathbf{H}_{i,j,n}^\mathsf{(s)} = \sqrt{N_j^\mathsf{T} N_i^\mathsf{R} p_{i,j,n}} \frac{h_{i,j,n}}{\sqrt{\rho_{i,j,n}}}
		e^{-j \frac{2 \pi n \tau_{i,j}}{N_j^\mathsf{B} T_s}} \mathbf{A}_{i,j}^\mathsf{r} \phi_{i,j,k}^\mathsf{(s)} 	\left(\mathbf{A}_{i,j}^\mathsf{t}\right)^\mathsf{H} \theta_{i,j}^\mathsf{(s)}
\end{aligned}
$}
\hspace{-3mm}
\vspace{-0.4mm}
\end{equation}
\vspace{-.5mm}
\hrulefill
\end{figure*}

\subsubsection{Sensing Model}
Sensing services encompass various applications, including detection, localization, and tracking, with this paper primarily focusing on the localization service as a representative example. In the localization service, in a semi-static sensing environment\footnote{We assume that the sensing target moves slowly, ensuring that the channel environment experienced by the signal remains approximately constant\cite{RW ISAC1}.}, each BS transmits sensing reference signals to the target, while MUs decode the target's position by processing the reflected signals. 
We assume that the approximate positions of targets within the region of interest are known \cite{known}.
We also presume that sensing services operate in two modes. In \textit{device-free sensing}, the targets are objects or entities that do not carry electronic devices \cite{RW ISAC1, device sensing}. In contrast, \textit{device-based sensing} refers to scenarios where the sensing targets are MUs equipped with electronic devices \cite{RW ISAC1}. 

\noindent~\textit{(a) Device-free sensing:}
Suppose that MU \( u_i \) has a device-free sensing task and is assigned to subchannel \( n \) of BS \( s_j \). Accordingly, the signal received by \( u_i \) from \( s_j \), after being \textit{reflected} by the target \( Q_i \), can be expressed as
\begin{equation}
	\begin{aligned}\label{sensing sign}
		&\mathbf{y}_{i,j,n}^\mathsf{(s)} =
		\sqrt{N_j^\mathsf{T} N_i^\mathsf{R} p_{i,j,n}} \frac{h_{i,j,n}}{\sqrt{\rho_{i,j,n}}}
		e^{-j \frac{2 \pi n \tau_{i,j}}{N_j^\mathsf{B} T_s}}
		\\&~~~~\times \mathbf{W}_{i,j,n}^\mathsf{H} \mathbf{A}_{i,j}^\mathsf{r} \phi_{i,j,k}^\mathsf{(s)} 	\left(\mathbf{A}_{i,j}^\mathsf{t}\right)^\mathsf{H} \theta_{i,j}^\mathsf{(s)} \times\mathbf{F}_{i,j,n}\mathbf{s}_\text{ref}
		+ \mathbf{W}_{i,j,n}^\mathsf{H} \mathbf{w}_{i,j,n},
	\end{aligned}
    \hspace{-4mm}
\end{equation}
where $p_{i,j,n}$, $h_{i,j,n}$, $\rho_{i,j,n}$ are the transmit power, complex channel gain\footnote{Following existing works~\cite{RW ISAC1,ISAC}, we do not explicitly introduce the Radar Cross Section (RCS) as a separate variable. Instead, its physical impact (e.g., target reflectivity) is implicitly embedded in the composite sensing channel gain $h_{i,j,n}$, which combines path loss, beamforming gain, and object reflection characteristics. 
}, and path loss, respectively. Also, $\mathbf{w}_{i,j,n}$ denotes the AWGN with power $\sigma_s^2$, $\mathbf{s}_\text{ref}$ is the sensing reference signal, and $\tau_{i,j}$ is the transmission time from BS $s_j$ to MU $u_i$, expressed as
\begin{equation}\label{equ. transmission time}
	\tau_{i,j} = ({d_{i}^\mathsf{Q} + d_{i,j}^\mathsf{Q}})\big/{c},
\end{equation}
where $c$ is the speed of light, $d_{i}^\mathsf{Q}$ and $d_{i,j}^\mathsf{Q}$ are the distances between target $Q_i$ and MU $u_i$ and BS $s_j$, which are given by
\begin{equation}\label{equ. 15}
	d_{i}^\mathsf{Q} = \sqrt{(x_i^\mathsf{U} - x_i^\mathsf{Q})^2 + (y_i^\mathsf{U} - y_i^\mathsf{Q})^2},
\end{equation}
\begin{equation}\label{equ. 16}
	d_{i,j}^\mathsf{Q} = \sqrt{(x_{i}^\mathsf{Q} - x_{j}^\mathsf{S})^2 + (y_i^\mathsf{Q} - y_{j}^\mathsf{S})^2}.
\end{equation}
Here, the position of MU $u_i$, BS $s_j$, and $u_i$'s sensing target $Q_i$ are denoted by $l_i^\mathsf{U} = [x_i^\mathsf{U}, y_i^\mathsf{U}]^\top$, $l^\mathsf{S}_{j} = [x^\mathsf{S}_{j}, y^\mathsf{S}_{j}]^\top$, $l_i^\mathsf{Q} = [x_i^\mathsf{Q}, y_i^\mathsf{Q}]^\top$, respectively. Moreover, the steering and response vector functions in (\ref{sensing sign}) follow the same structures as (\ref{AOD}) and (\ref{AOA}), where \( \theta^\mathsf{(s),Q}_{i,j} \) represents the AoD from BS \( s_j \) to \( Q_i \), and \( \phi^\mathsf{(s),Q}_{i} \) denotes the AoA from \( Q_i \) to BS \( s_j \), which are given by
\begin{equation}
	\theta_{i}^\mathsf{\mathsf{(s),Q}} = \arccos \left( ({x_i^\mathsf{U} - x_{i}^\mathsf{Q}})\big/{\| l_{i}^\mathsf{Q} - l^\mathsf{U}_i \|_2} \right),
\end{equation}
\begin{equation}\label{equ. 18}
	\phi_{i,j}^\mathsf{\mathsf{(s),Q}} = \pi - \arccos \left( {(x_i^\mathsf{Q} - x_{j}^\mathsf{S})}\big/{\| l_i^\mathsf{Q} - l_{j}^\mathsf{S} \|_2} \right).
\end{equation}

\noindent~\textit{(b) Device-based sensing:}
When considering MU \( u_i \) as the sensing target \( Q_i \), the expression of the received signal at MU \( u_i \) follows a similar form to (\ref{sensing sign}), with the angles \( \theta^\mathsf{(s),Q}_i \) and \( \phi^\mathsf{(s),Q}_{i,j} \) replaced by \( \theta_{i,j} \), which represents the AoD and AoA from BS \( s_j \) to MU \( u_i \). 

To quantify the quality of sensing in both device-free and device-based sensing scenarios, we use a common  sensing metric, called Position Error Bound (PEB), which is derived from the Fisher Information Matrix (FIM, \( \mathbf{J}(l^\mathsf{Q}_i) \)) as follows:
\begin{equation}
	\mathbb{E}\left[ \| \hat{l}_{i}^\mathsf{Q} - l_{i}^\mathsf{Q} \|^2 \right] \geq \text{tr} \{ \mathbf{J}^\mathsf{-1}(l_{i}^\mathsf{Q}) \},
\end{equation}
where $\hat{l}_{i}^\mathsf{Q}$ and $ l_{i}^\mathsf{Q}$ are the estimation and true location of sensing target  $Q_i \in\bm{\mathcal{Q}}$, while $\mathbb{E}[\cdot]$ denotes the expectation operator.
Instead of deriving the PEB directly, a more practical approach is to compute the Cramér-Rao Lower Bound (CRLB) for \( \tau_{i,j} \), \( \theta^\mathsf{(s),Q}_i \), and \( \phi^\mathsf{(s)}_{i,j} \) in device-free sensing (or \(\tau_{i,j} \), \( \theta_{i,j} \) in device-based sensing)\cite{RW ISAC1}. We define \( {\boldsymbol{\eta}_{i,j}} = [\tau_{i,j}, \theta^\mathsf{(s),Q}_i, \phi^\mathsf{(s)}_{i,j}]^\mathsf{T} \), and the CRLB\footnote{Note that the CRLB expression inherently depends on the true values of the angular parameters (e.g., $\theta^\mathsf{(s),Q}_i$, $\phi^\mathsf{(s)}_{i,j}$), which are typically unavailable during real-time operation. To address this, we follow a standard practice~\cite{RW ISAC1}, where the CRLB is treated as an offline service quality indicator. Specifically, these angular parameters are configured using representative deployment scenarios, prior knowledge, or worst-case estimates. The resulting bounds are abstracted into a coefficient $\zeta_{i,j}$ later in~\eqref{eq:PEBmain}, enabling tractable modeling of sensing utility without requiring real-time access to ground-truth parameters.} (i.e., the variance of the estimate: \( \text{var}(\bm{\hat{\eta}}_{i,j}) \), \( {\boldsymbol{\hat{\eta}}_{i,j}} = [\hat{\tau}_{i,j}, \hat{\theta}^\mathsf{(s),Q}_i, \hat{\phi}^\mathsf{(s)}_{i,j}]^\mathsf{T} \)) serves as a lower bound for the FIM of \( {\boldsymbol{\eta}_{i,j}} \). Specifically, we have \( \text{var}(\bm{\hat{\eta}}_{i,j}) \leq \mathbf{J}_{\boldsymbol{\eta}_{i,j}} \), where $\mathbf{J}_{\boldsymbol{\eta}_{i,j}}$ denotes the FIM of \( {\boldsymbol{\eta}_{i,j}} \),  given by
\begin{equation}
	\mathbf{J}_{\boldsymbol{\eta}_{i,j}} = \mathbb{E}_{\boldsymbol{\eta}_{i,j}} \left[ \nicefrac{- \partial^2 \ln f(\mathbf{y}_{i,j}^\mathsf{\mathsf{(s)}} | \boldsymbol{\eta}_{i,j})}{\partial \boldsymbol{\eta}_{i,j} \partial (\boldsymbol{\eta}_{i,j})^\mathsf{T}} \right],
\end{equation}
where $f(\mathbf{y}_{i,j}^\mathsf{\mathsf{(s)}} | \boldsymbol{\eta}_{i,j})$ is the likelihood function of $\mathbf{y}_{i,j }^\mathsf{\mathsf{(s)}}$ with respect to $\boldsymbol{\eta}_{i,j}$, shown in (\ref{equ. 21}).
{Let $\mathcal{N}_{i,j}^\mathsf{\mathsf{(s)}} = \{\mathbbm{l}_{i,j}^\mathsf{\mathsf{(s)}}, \dots, \mathbbm{l}_{i,j}^\mathsf{\mathsf{(s)}} + N_{i,j}^\mathsf{\mathsf{(s)}} \}$ contain the index of subchannels allocated by BS \( s_j \) to MU \( u_i \) for sensing target \( Q_i \), where \( \mathbbm{l}^\mathsf{(s)}_{i,j} \) is the initial subchannel index\cite{RW ISAC1} and $N_{i,j}$ is the number of assigned subchannels (i.e., $B_{i,j}^\mathsf{(s)} = N_{i,j}^\mathsf{(s)} B_0$ is the total bandwidth of sensing).}
The derivation of \(\mathbf{J}_{\boldsymbol{\eta}_{i,j}} \) with respect to \( B^\mathsf{(s)}_{i,j} \) and the  power allocated for sensing \( P^\mathsf{(s)}_{i,j} \), is provided in Appendix A, where the PEB for the MU-BS pair \( i,j \) can finally be approximated as
\begin{equation}\label{eq:PEBmain}
	\text{PEB}_{i,j}=\sqrt{ \text{tr} \left\{ \mathbf{J}_{\boldsymbol{\eta}_{i,j}}^{-1} \right\}} \approx {\zeta_{i,j}}\Big/{\sqrt{N_{i,j}^\mathsf{(s)} P_{i,j}^\mathsf{(s)}}},
\end{equation}
where \( \zeta_{i,j} \) is influenced by multiple factors, including bandwidth, antennas configuration, channel gain, noise power, and the relative positions of the sensing target and the MU (detailed derivations are provided in Appendix A). We further introduce \( \omega_2,\omega_3 \in[0, 1] \) as the power and bandwidth coefficients to simply model the PEB as follows:
\begin{equation}
	\frac{1}{\text{PEB}_{i,j}} \approx \kappa_{i,j} \left( P_{i,j}^\mathsf{(s)} \right)^\mathsf{\omega_2} \left( B_{i,j}^\mathsf{(s)} \right)^\mathsf{\omega_3},
\end{equation}
{where \( \kappa_{i,j} = {\vartheta_{i,j}}/{\zeta_{i,j}} \)  (\( \vartheta_{i,j}\in [0,1] \)) represents the normalized relative sensing capability between BS $s_j$ and MU $u_i$, reflecting heterogeneous sensing quality due to spatial and device-level variations\cite{RW ISAC1}. }The coefficient \( \kappa_{i,j} \) is influenced by system configurations, imperfections in beamforming gain, filtering gain, and the signal processing algorithm\cite{sensing factor}.

Finally, we define the \textit{value} of sensing for each MU \( u_i \) as
\begin{equation}\label{sensing value}
	V_{i,j}^\mathsf{(s)}=\omega_4\frac{1}{\text{PEB}_{i,j}}\approx\omega_4\kappa_{i,j} \left( P_{i,j}^\mathsf{(s)} \right)^\mathsf{\omega_2} \left( B_{i,j}^\mathsf{(s)} \right)^\mathsf{\omega_3},
\end{equation}
where \( \omega_4 >0\) is a weighting coefficient, determining the relative importance of sensing accuracy in the utility function.

\section{Role-Friendly Win-Win Matching for Offline Trading (offRFW$^2$M)}
\noindent In this section, we introduce offRFW$^2$M, which achieves mutually beneficial and risk-aware long-term contracts for both communication and sensing services. 


\subsection{Key Definitions}
We first introduce the key definitions of our many-to-many (M2M) matching framework between MUs and BSs, designed for the offline trading mode. Unlike conventional matching mechanisms, this framework is crafted to mitigate potential risks, ensuring the robustness of long-term contracts. 

\begin{Defn}(M2M Matching for Communication Services in offRFW$^2$M)
	An M2M matching \( \varphi^\mathsf{(c)} \)  for communication services in offRFW$^2$M constitutes a two-way function/mapping between BSs \( \bm{\mathcal{S}} \) and MUs \(\bm{\mathcal{U}} \), satisfying the following properties:
	
	\noindent
	$\bullet$ For each BS $ s_{j} \in \bm{\mathcal{S}},\varphi^\mathsf{(c)}\left( s_j \right) \subseteq \bm{\mathcal{U}} $, meaning that a BS can potentially provide communication services to multiple MUs.
	
	\noindent
	$\bullet$ For each MU $ u_i \in \bm{\mathcal{U}}, \varphi^\mathsf{(c)}\left( u_i \right) \subseteq \bm{\mathcal{S}} $, and $|\varphi^\mathsf{(c)}\left( u_i \right)|=1$, ensuring that each MU is assigned to exactly one BS.
	
	\noindent
	$\bullet$ For each BS $ s_j $ and MU $ u_i $, $ s_j\in\varphi^\mathsf{(c)}(u_i)$ if and only if $ u_i\in\varphi^\mathsf{(c)}\left(s_j\right) $, indicating that a valid matching occurs only when both the MU and the BS mutually accept the contract.
\end{Defn}

\begin{Defn}(M2M Matching for Sensing Services in offRFW$^2$M)
	An M2M matching \( \varphi^\mathsf{(s)} \)  for sensing services in offRFW$^2$M constitutes a two-way function/mapping between BSs  \( \bm{\mathcal{S}} \) and coalitions \( \bm{\mathcal{C}} \), satisfying the following properties:
	
	\noindent
	$\bullet$ For each BS $ s_{j} \in \bm{\mathcal{S}},\varphi^\mathsf{(s)}\left( s_j \right) \subseteq \bm{\mathcal{C}} $, meaning that a BS can provide sensing resources to multiple sensing coalitions.
	
	\noindent
	$\bullet$ For each MUs' coalition $ \bm{c}_k \in \bm{\mathcal{C}}, \varphi^\mathsf{(s)}\left( \bm{c}_k \right) \subseteq \bm{\mathcal{S}} $ and $|\varphi^\mathsf{(s)}\left( \bm{c}_k \right)|=1$, ensuring that each coalition is assigned to one or more BSs, allowing cooperative resource provisioning.
	
	\noindent
	$\bullet$ For each $ s_j $ and $ \bm{c}_k $, $ s_j\in\varphi^\mathsf{(s)}(\bm{c}_k)$ if and only if $ \bm{c}_k\in\varphi^\mathsf{(s)}\left(s_j\right) $, ensuring that a matching occurs only when both parties mutually accept the 
    resource-sharing agreement.
\end{Defn}

\subsection{Utility, Expected Utility, and Risk of Clients}\label{sec:UtilityClients}
\subsubsection{Utility of Individual MUs for Communication Services}
We obtain the utility of communication services for an MU \( u_i \) through three key components: \textit{(i)} The valuation derived from utilizing communication services provided by BS \( s_j \), subtracting the payment made by \( u_i \) to \( s_j \). \textit{(ii)} The penalty incurred when \( u_i \) breaks the contract (e.g., when \( \alpha_i = 0 \) in a practical transaction). \textit{(iii)} The compensation received if \( u_i \) is selected as a volunteer. We formulate the utility function as
\begin{equation}\label{equ. comm utility}
	\begin{aligned}
		u^\mathsf{(c),U}\left(u_i,\varphi^\mathsf{(c)}(u_i),\mathbb{C}^\mathsf{(c)}_{i,j}\right)&=\alpha_i(1-\mathbbm{v}^\mathsf{(c)}_{i,j})\left(V_{i,j}^\mathsf{(c)}-\mathbbm{c}^\mathsf{(c),Pay}_{i,j}\right)\\[-.1em]&\hspace{-4mm}-(1-\alpha_i)\mathbbm{c}^\mathsf{(c),PelU}_{i,j}+\alpha_i\mathbbm{v}^\mathsf{(c)}_{i,j}\mathbbm{c}^\mathsf{(c),PelS}_{i,j},
	\end{aligned}
\end{equation}
where \( \mathbbm{v}^\mathsf{(c)}_{i,j} = 1 \) indicates that \( u_i \) is selected by \( s_j \) as a volunteer in a practical transaction, and \(\mathbbm{v}^\mathsf{(c)}_{i,j} = 0 \) otherwise. In~\eqref{equ. comm utility}, \( V_{i,j}^\mathsf{(c)} \), defined in (\ref{comm value}), is related to the long-term contract \( \mathbb{C}^\mathsf{(c)}_{i,j} \)  between MU $u_i$ and BS \( s_j\in\varphi^\mathsf{(c)}(u_i) \), where the  bandwidth \( \mathbbm{c}^\mathsf{(c),B}_{i,j} \) and power \( \mathbbm{c}^\mathsf{(c),Pow}_{i,j} \) traded are specified (see Sec.~\ref{sec:over}).

\subsubsection{Utility of Coalitions for Sensing Services} We obtain the utility that a sensing coalition \( \bm{c}_k \) obtains through three key components: \textit{(i)} The valuation of sensing services provided by BS \( s_j \), subtracting the payment made by \( \bm{c}_k \) to \( s_j \). \textit{(ii)} The penalty incurred when \( \bm{c}_k \) breaks the contract. \textit{(iii)} The compensation received if \( \bm{c}_k \) is selected as a volunteer. We thus formulate the utility function for sensing coalitions as
\begin{equation}\label{equ. sensing utility}
	\hspace{-3mm}\begin{aligned}
 &u^\mathsf{(s),U}\left(\bm{c}_k,\varphi^\mathsf{(s)}(\bm{c}_k),\mathbb{C}^\mathsf{(s)}_{k,j}\right)=\beta_k\mathbbm{v}^\mathsf{(s)}_{k,j}\mathbbm{c}^\mathsf{(s),PelS}_{k,j}
        \\[-.2em]
        &-(1-\beta_k)\mathbbm{c}^\mathsf{(s),PelU}_{k,j}+(1-\mathbbm{v}^\mathsf{(s)}_{k,j})\beta_k\left(|\bm{c}_k|V^\mathsf{(s),max}_{k,j}-\mathbbm{c}^\mathsf{(s),Pay}_{k,j}\right),
	\end{aligned}
    \hspace{-3mm}
\end{equation}
where $\mathbbm{v}^\mathsf{(s)}_{k,j} = 1$ implies that  $\bm{c}_k$ is selected by $s_j$ as a volunteer in a practical transaction, and $\mathbbm{v}^\mathsf{(s)}_{k,j} = 0$ otherwise.
In \eqref{equ. sensing utility}, \( V^\mathsf{(s),max }_{k,j} \) is the maximum sensing value of MUs in coalition \( \bm{c}_k \), defined as $ V^\mathsf{(s),max }_{k,j}=\max_{u_i\in \bm{c}_k}\{ V^\mathsf{(s)}_{i,j}\}$\footnote{In each coalition $\bm{c}_k$, the MU with the highest sensing value is selected as the representative to interact with BSs. The resulting sensing outcome (e.g., localization result) and its associated sensing value are equally shared among all coalition members. For modeling simplicity, we assume the total value obtained by the coalition is given by $|\bm{c}_k|V^\mathsf{(s),max}_{k,j}$, where $V^\mathsf{(s),max}_{k,j}$ denotes the representative’s sensing value. All coalition members jointly commit to the contract and evenly split the total payment required for sensing resource usage (i.e., $\mathbbm{c}^\mathsf{(s),Pay}_{k,j}/|\bm{c}_k|$). Similarly, when a coalition is selected as a volunteer, the compensation offered by the BS is also evenly divided among its members (i.e., $\mathbbm{c}^\mathsf{(s),PelS}_{k,j}/|\bm{c}_k|$). These rules ensure incentive compatibility and fairness within the coalition.}, and $V^\mathsf{(s)}_{i,j}$, defined in (\ref{sensing value}), is related to the traded bandwidth \( \mathbbm{c}^\mathsf{(s),B}_{k,j} \) and power \( \mathbbm{c}^\mathsf{(s),Pow}_{k,j} \)  through the long-term contract $\mathbb{C}^\mathsf{(s)}_{k,j}$ (see Sec.~\ref{sec:over}). Also, in \eqref{equ. sensing utility}, \( \beta_k = 1 \) captures that \( \bm{c}_k \) participates in a practical transaction (i.e.,  $\sum_{u_i\in \bm{c}_k}\alpha_i > 0$), and \( \beta_k = 0 \) otherwise\footnote{A coalition $\bm{c}_k$ participates in a practical transaction if and only if any $u_i$ in $\bm{c}_k$ participates: $\bm{c}_k$ will be absent from the transaction if $\mu(\bm{c}_k)=\emptyset$.}. Since uncertainties introduce challenges in directly maximizing the values of (\ref{equ. comm utility}) and (\ref{equ. sensing utility}) in the offline trading mode, we use their expected values over the uncertainties as follows:\footnote{{The expectations are taken with respect to the joint distribution of random variables $\alpha_i$, $\beta_k$, $\mathbbm{v}^\mathsf{(c)}_{i,j}$, and $\mathbbm{v}^\mathsf{(s)}_{k,j}$.} It is also assumed that $\alpha_i$ and $\mathbbm{v}^\mathsf{(c)}_{i,j}$, and $\beta_k$ and $\mathbbm{v}^\mathsf{(s)}_{k,j}$ are independent.} 

\begin{equation}\label{equ. expected comm utility}\hspace{-3mm}
\resizebox{0.405\textwidth}{!}{$
	\begin{aligned}
		&\mathbb{E}\left [u^\mathsf{(c),U}(u_i,\varphi^\mathsf{(c)}(u_i),\mathbb{C}^\mathsf{(c)}_{i,j})\right ]=\mathbb{E}[\alpha_i]\mathbb{E}[\mathbbm{v}^\mathsf{(c)}_{i,j}]\mathbbm{c}^\mathsf{(c),PelS}_{i,j}-\\[-.2em]
        &(1-\mathbb{E}[\alpha_i])\mathbbm{c}^\mathsf{(c),PelU}_{i,j}+\mathbb{E}[\alpha_i](1-\mathbb{E}[\mathbbm{v}^\mathsf{(c)}_{i,j}])(V_{i,j}^\mathsf{(c)}-\mathbbm{c}^\mathsf{(c),Pay}_{i,j}),
	\end{aligned}
    $}
\end{equation}
\begin{equation}\label{equ. expected sensing utility}
 \hspace{-3mm}
\resizebox{0.46\textwidth}{!}{$
	\begin{aligned}
		&\mathbb{E}\left [u^\mathsf{(s),U}(\bm{c}_k,\varphi^\mathsf{(s)}(\bm{c}_k),\mathbb{C}^\mathsf{(s)}_{k,j})\right ]=-(1-\mathbb{E}[\beta_k])\mathbbm{c}^\mathsf{(s),PelU}_{k,j}\\[-.2em]
        &+(1-\mathbb{E}[\mathbbm{v}^\mathsf{(s)}_{k,j}])\mathbb{E}[\beta_k](\mathbb{E}[V^\mathsf{(s),max }_{k,j}]-\mathbbm{c}^\mathsf{(s),Pay}_{k,j})+\mathbb{E}[\beta_k]\mathbb{E}[\mathbbm{v}^\mathsf{(s)}_{k,j}]\mathbbm{c}^\mathsf{(c),PelS}_{i,j},
	\end{aligned}
    $}
    \hspace{-3.5mm}
\end{equation}

\noindent {where variables $\alpha_i$ and $\beta_k$ denote the participation status of MU $u_i$ and coalition $c_k$ in the practical transaction, respectively. 
Similarly, $\mathbbm{v}^\mathsf{(c)}_{i,j}$ and $\mathbbm{v}^\mathsf{(s)}_{k,j}$ denote whether MU $u_i$ and coalition $c_k$ are selected by BS $s_j$ as volunteers, i.e., they voluntarily relinquish their participation in the transaction to balance resource allocation of BS $s_j$. The expectation values of these variables, used in (\ref{equ. expected comm utility})--(\ref{equ. expected sensing utility}), are formally obtained in Appendix B.}

Carefully investigating offline trading unveils potential risks due to the uncertainties involved. To evaluate these risks on the client side, we define two key risk assessment aspects:

\noindent~1) An individual MU \( u_i \) faces the risk of obtaining an undesired utility. We quantify this risk as the probability that the utility of \( u_i \) falls below a threshold \(u^\mathsf{(c)}_{\min} \), formulated as
\begin{equation}
\hspace{-3mm}
\resizebox{0.47\textwidth}{!}{$
	\begin{aligned}
			&R_1^\mathsf{U}\big(u_i,\varphi^\mathsf{(c)}(u_i),\mathbb{C}^\mathsf{(c)}_{i,j}\big)=\Pr\left(u^\mathsf{(c),U}\left(u_i,\varphi^\mathsf{(c)}(u_i),\mathbb{C}^\mathsf{(c)}_{i,j}\right)\leq u^\mathsf{(c)}_\mathsf{\min}\right).
	\end{aligned}
    $}\hspace{-5mm}
\end{equation}

\noindent~2) Similarly, a sensing coalition \( \bm{c}_k \) faces the risk of an unsatisfactory utility. We formulate this risk as the probability that the utility of \( \bm{c}_k \) falls below a tolerable threshold \( u^\mathsf{(s)}_{\min} \):
\begin{equation}
\hspace{-3mm}
\resizebox{0.47\textwidth}{!}{$
	\begin{aligned}			&R_2^\mathsf{U}(\bm{c}_k,\varphi^\mathsf{(s)}(\bm{c}_k),\mathbb{C}^\mathsf{(s)}_{k,j}))=\Pr\left(u^\mathsf{(s),U}\left(\bm{c}_k,\varphi^\mathsf{(s)}(\bm{c}_k),\mathbb{C}^\mathsf{(s)}_{k,j})\right)\leq u^\mathsf{(s)}_\mathsf{\min}\right).
	\end{aligned}
    $}\hspace{-5mm}
\end{equation}

The above risks should be managed when designing long-term contracts. Otherwise, clients may prefer online trading, opting out of long-term agreements with any BS.

\subsection{Utility, Expected Utility, and Risk of BSs}
\subsubsection{Utility of BSs for Communication Services} We obtain the utility of a BS for providing communication services through three components: \textit{(i)} payments received from MUs, {\textit{(ii)} power cost incurred for serving communication demands (i.e., $\mathbbm{c}^\mathsf{(c),Pow}_{i,j} = \omega_5 {P}_{i,j}^\mathsf{(c)}$, where $\omega_5$ denotes the unit power cost),} {\textit{(iii)} compensation received from MUs who break their contracts, and \textit{(iv)} compensation paid to MUs selected as volunteers.} Accordingly, the utility function of BS \( s_j \) is formulated by
\begin{equation}\label{equ. comm BS U}
	\hspace{-4mm}\begin{aligned}
		u^\mathsf{(c),S}\left (s_j,\varphi^\mathsf{(c)}(s_j),\mathbb{C}^\mathsf{(c)}_{i,j}\right )&=\hspace{-5mm}\sum_{u_i\in\varphi^\mathsf{(c)}(s_j)}\hspace{-4mm}(1-\mathbbm{v}^\mathsf{(c)}_{i,j})\alpha_i(\mathbbm{c}^\mathsf{(c),Pay}_{i,j}\hspace{-1mm}-\mathbbm{c}^\mathsf{(c),Pow}_{i,j})\\[-.1em]
        &\hspace{-4mm}+(1-\alpha_i)\mathbbm{c}^\mathsf{(c),PelU}_{i,j}-\alpha_i\mathbbm{v}^\mathsf{(c)}_{i,j}\mathbbm{c}^\mathsf{(c),PelS}_{i,j}.
	\end{aligned}
    \hspace{-4mm}
\end{equation}

\subsubsection{Utility of BSs for Sensing Services} Similarly, the utility of a BS for providing sensing services comprises: 
\textit{(i)} payments received from sensing coalitions, {\textit{(ii)} power cost incurred for serving sensing demands (i.e., $\mathbbm{c}^\mathsf{(s),Pow}_{k,j} = \omega_5 {P}_{k,j}^\mathsf{(s)}$, where $\omega_5$ denotes the unit power cost),} {\textit{(iii)} compensation received from coalitions that break their contracts, and \textit{(iv)} compensation paid to coalitions who are determined as volunteers}\footnote{The compensations $\mathbbm{c}_{i,j}^{\mathsf{(c),PelS}}$ and $\mathbbm{c}_{k,j}^{\mathsf{(s),PelS}}$ are considered to be fixed as determined by the contract, rather than being dynamically computed from user-specific utility losses. This design ensures implementability under incomplete information and mitigates the risk of strategic manipulation of compensations.}. Mathematically, we have 
\begin{equation}\label{equ. sensing BS U}
	 \hspace{-4mm}\begin{aligned}
	u^\mathsf{(s),S}\left(s_j,\varphi^\mathsf{(s)}(s_j), \mathbb{C}^\mathsf{(s)}_{k,j}\right)&=\hspace{-4mm}\sum_{\bm{c}_k\in\varphi^\mathsf{(c)}(s_j)}\hspace{-4mm}(1-\mathbbm{v}^\mathsf{(s)}_{k,j})\beta_k(\mathbbm{c}^\mathsf{(s),Pay}_{k,j}-\mathbbm{c}^\mathsf{(s),Pow}_{k,j})\\[-.1em]
    &\hspace{-4mm}+(1-\beta_k)\mathbbm{c}^\mathsf{(s),PelU}_{k,j}-\beta_k\mathbbm{v}^\mathsf{(s)}_{k,j}\mathbbm{c}^\mathsf{(s),PelS}_{k,j}.
	\end{aligned}
    \hspace{-6mm}
\end{equation}
Due to uncertainties that prevent the direct computation of values in (\ref{equ. comm BS U}) and (\ref{equ. sensing BS U}), we consider their expected values:
\begin{equation}
 \hspace{-9mm}
	\begin{aligned}
		&\mathbb{E}\left [u^\mathsf{(c),S}(s_j,\varphi^\mathsf{(c)}(s_j),\mathbb{C}^\mathsf{(c)}_{i,j})\right ]\\
        &~~~~~~~~~~~~=\hspace{-3mm}\sum_{u_i\in\varphi^\mathsf{(c)}(s_j)}(1-\mathbb{E}[\mathbbm{v}^\mathsf{(c)}_{i,j}])\mathbb{E}[\alpha_i](\mathbbm{c}^\mathsf{(c),Pay}_{i,j}-\mathbbm{c}^\mathsf{(c),Pow}_{i,j})\\[-.1em]
        &~~~~~~~~~~~~~~+(1-\mathbb{E}[\alpha_i])\mathbbm{c}^\mathsf{(c),PelU}_{i,j}-\mathbb{E}[\alpha_i]\mathbb{E}[\mathbbm{v}^\mathsf{(c)}_{i,j}]\mathbbm{c}^\mathsf{(c),PelS}_{i,j},
	\end{aligned}
     \hspace{-4mm}
\end{equation}
\begin{equation} \hspace{-4mm}
\hspace{-9mm}
	\begin{aligned}
		&\mathbb{E}\left [u^\mathsf{(s),S}(s_j,\varphi^\mathsf{(s)}(s_j),\mathbb{C}^\mathsf{(s)}_{k,j})\right ]\\&~~~~~~~~~~~~=\hspace{-3mm}\sum_{\bm{c}_k\in\varphi^\mathsf{(c)}(s_j)}(1-\mathbb{E}[\mathbbm{v}^\mathsf{(s)}_{k,j}])\mathbb{E}[\beta_k](\mathbbm{c}^\mathsf{(s),Pay}_{k,j}-\mathbbm{c}^\mathsf{(s),Pow}_{k,j})\\[-.1em]&~~~~~~~~~~~~~~~+(1-\mathbb{E}[\beta_k])\mathbbm{c}^\mathsf{(s),PelU}_{k,j}-\mathbb{E}[\beta_k]\mathbb{E}[\mathbbm{v}^\mathsf{(s)}_{k,j}]\mathbbm{c}^\mathsf{(s),PelS}_{k,j},
	\end{aligned}
     \hspace{-4mm}
\end{equation}
where derivations on how we obtain the expected values are detailed by Appendix B. Similar to Sec.~\ref{sec:UtilityClients}, we define two key risk assessment aspects for BSs as follows:

\noindent~1) A BS \( s_j \) risks not being able to provide the committed bandwidth during actual transactions, quantified as
\begin{equation}
	\begin{aligned}		&R_1^\mathsf{S}\big(s_j,\varphi^\mathsf{(c)}(s_j),\mathbb{C}^\mathsf{(c)}_{i,j},\varphi^\mathsf{(s)}(s_j),\mathbb{C}^\mathsf{(s)}_{k,j}\big)=\\[-.1em]
    &~~~~\Pr\Bigg(\sum_{u_i\in\varphi^\mathsf{(c)}(s_j)}\alpha_i\mathbbm{c}^\mathsf{(c),B}_{i,j}+\sum_{\bm{c}_k\in\varphi^\mathsf{(s)}(s_j)}\beta_k\mathbbm{c}^\mathsf{(s),B}_{k,j}> B_j\Bigg).
	\end{aligned}
\end{equation}


\noindent~2) A BS \( s_j \) risks failing to supply the stipulated power resources during actual transactions, quantified as
\begin{equation}
	\begin{aligned}		&R_2^\mathsf{S}(s_j,\varphi^\mathsf{(c)}(s_j),\mathbb{C}^\mathsf{(c)}_{i,j},\varphi^\mathsf{(s)}(s_j),\mathbb{C}^\mathsf{(s)}_{k,j})=\\[-.1em]
    &~~~~\Pr\Bigg(\sum_{u_i\in\varphi^\mathsf{(c)}(s_j)}\alpha_i\mathbbm{c}^\mathsf{(c),Pow}_{i,j}+\sum_{\bm{c}_k\in\varphi^\mathsf{(s)}(s_j)}\beta_k\mathbbm{c}^\mathsf{(s),Pow}_{k,j}> P_j\Bigg).
	\end{aligned}
\end{equation}

Managing these risks is crucial; otherwise, a BS may prefer online trading over committing to long-term contracts.

\subsection{Problem Formulation of Offline Trading Mode}\label{sec:probForm}
We formulate the resource trading in the offline trading mode to optimize the M2M matching and long-term contracts between clients and BSs. Accordingly, the objective of each BS \( s_j \in \bm{\mathcal{S}} \) is to \textit{maximize its overall utility} as follows:
\begin{subequations}
	\begin{align}
\hspace{-3mm}\bm{\mathcal{F}^\mathsf{S}} \hspace{-0.1mm}{:}&\hspace{-1.5mm}\underset{{\mathbb{C}^\mathsf{(c)}_{i,j},\mathbb{C}^\mathsf{(s)}_{k,j}}}{\max}\hspace{-1.25mm}\mathbb{E}\hspace{-0.75mm}\left[\hspace{-0.4mm}u^\mathsf{(c),S}(s_j,\hspace{-0.25mm}\varphi^\mathsf{(c)}(s_j),\hspace{-0.25mm}\mathbb{C}^\mathsf{(c)}_{i,j}){+}u^\mathsf{(s),S}(s_j,\hspace{-0.25mm}\varphi^\mathsf{(s)}(s_j),\hspace{-0.25mm}\mathbb{C}^\mathsf{(s)}_{k,j})\hspace{-0.45mm}\right] \label{equ. PF BS} \tag{37}\hspace{-3mm}\\
&\hspace{-4mm}\text{\textbf{s.t.}}~~~\hspace{-0mm}\varphi^\mathsf{(c)}\left(s_j\right)\subseteq\bm{\mathcal{U}},\varphi^\mathsf{(s)}\left(s_j\right)\subseteq\bm{\mathcal{C}}, \mu\left(\bm{c}_k\right)\subseteq\bm{\mathcal{U}}, \tag{37a}\label{equ. PF BS C1}\\
		&\hspace{-4mm}u_i\in \varphi^\mathsf{(c)}(s_j), \bm{c}_k\in \varphi^\mathsf{(s)}(s_j), u_i\in\mu(\bm{c}_k), \tag{37b}\label{equ. PF BS C2}\\
		&\hspace{-4mm}\sum_{u_i\in\varphi^\mathsf{(c)}(s_j)}B_{i,j}^\mathsf{(c)}+\sum_{\bm{c}_k\in\varphi^\mathsf{(s)}(s_j)}\mathbbm{c}^\mathsf{(s),B}_{k,j}\leq (1+O_j^\mathsf{B})B_j, \tag{37c}\label{equ. PF BS C3}\\
		&\hspace{-4mm}\sum_{u_i\in\varphi^\mathsf{(c)}(s_j)}P_{i,j}^\mathsf{(c)}+\sum_{\bm{c}_k\in\varphi^\mathsf{(s)}(s_j)}\mathbbm{c}^\mathsf{(s),Pow}_{k,j}\leq (1+O_j^\mathsf{Pow})P_j, \tag{37d}\label{equ. PF BS C4}\\
		&\hspace{-4mm}R_1^\mathsf{S}(s_j,\varphi^\mathsf{(c)}(s_j),\mathbb{C}^\mathsf{(c)}_{i,j},\varphi^\mathsf{(s)}(s_j),\mathbb{C}^\mathsf{(s)}_{k,j})\leq \rho_1,\label{equ. PF BS C5}\tag{37e}
		\\&\hspace{-4mm}R_2^\mathsf{S}(s_j,\varphi^\mathsf{(c)}(s_j),\mathbb{C}^\mathsf{(c)}_{i,j},\varphi^\mathsf{(s)}(s_j),\mathbb{C}^\mathsf{(s)}_{k,j})\leq \rho_2,\label{equ. PF BS C6}\tag{37f}
	\end{align}
\end{subequations}

\noindent where $ \rho_1 $ and $\rho_2 $ are risk thresholds falling in interval $ (0, 1] $, and $O_j^\mathsf{B}$ and $O_j^\mathsf{Pow}$ represent the overbooking rate of BS $s_j$ for bandwidth and power resource, respectively. In $ \bm{\mathcal{F}^\mathsf{S}} $, constraints (\ref{equ. PF BS C1}) and (\ref{equ. PF BS C2}) enforce the feasibility of communication and sensing M2M matchings  $\varphi^\mathsf{(c)}$ and $\varphi^\mathsf{(s)}$. 
 Constraints (\ref{equ. PF BS C3}) and (\ref{equ. PF BS C4}) ensure that the bandwidth and power resources sold by BS $s_j$ do not exceed its supply $(1+O_j^\mathsf{B})B_j$ and $(1+O_j^\mathsf{Pow})P_j$ after overbooking. Constraints (\ref{equ. PF BS C5}) and (\ref{equ. PF BS C6}) specify the risks of BSs with their tractable forms obtained in Appendix B through probabilistic analysis.
Further, each client (i.e., $u_i $ or $\bm{c}_k$) aims \textit{to maximize its utility} as follows:
\begin{subequations}
	\begin{align}
	\bm{\mathcal{F}^\mathsf{U}}:~&
	\left\{ \begin{matrix}
		\underset{{\mathbb{C}^\mathsf{(c)}_{i,j}}}{\max}~\mathbb{E}\left[u^\mathsf{(c),U}(u_i,\varphi^\mathsf{(s)}(u_i),\mathbb{C}^\mathsf{(c)}_{i,j})\right] \\[-.2em]
		\underset{{\mathbb{C}^\mathsf{(s)}_{k,j}}}{\max}~\mathbb{E}\left[u^\mathsf{(s),U}(\bm{c}_k,\varphi^\mathsf{(s)}(\bm{c}_k),\mathbb{C}^\mathsf{(s)}_{k,j})\right]
	\end{matrix}\right\}, \tag{38}\label{equ. PF MU}\\[-.1em]
		\text{\textbf{s.t.}}~~~
		&\varphi^\mathsf{(c)}\left(u_i\right)\subseteq\bm{\mathcal{S}}, \mu\left(u_i\right)\subseteq\bm{\mathcal{C}},\varphi^\mathsf{(s)}\left(\bm{c}_k\right)\subseteq\bm{\mathcal{S}} ,\label{equ. PF MU C1}\tag{38a}\\[-.2em]
		&s_j\in \varphi^\mathsf{(c)}(u_i), \bm{c}_k\in \mu(u_i), s_j\in\varphi^\mathsf{(s)}(\bm{c}_k), \tag{38b}\label{equ. PF MU C2}\\[-.2em]
		&V^\mathsf{(c)}_{i,j}\ge \mathbbm{c}^\mathsf{(c),Pay}_{i,j},V^\mathsf{(s)}_{i,j}=V^\mathsf{(s),max }_{k,j}\ge \mathbbm{c}^\mathsf{(s),Pay}_{i,j}, \tag{38c}\label{equ. PF MU C3}\\[-.2em]
		&V^\mathsf{(c)}_{i,j}\geq R^\mathsf{req},\tag{38d}\label{equ. PF MU C4}\\[-.2em]
		&V^\mathsf{(s),max }_{k,j}\geq S^\mathsf{req},\tag{38e}\label{equ. PF MU C5}\\[-.2em]
            & B_{\min} \leq B^\mathsf{(c)}_{i,j}, B^\mathsf{(s)}_{k,j}\leq B_{\max},\tag{38f}\label{equ. PF MU C8}\\[-.2em]
		&P_{\min} \leq P^\mathsf{(c)}_{i,j}, P^\mathsf{(s)}_{k,j}\leq P_{\max},\tag{38g}\label{equ. PF MU C9}\\[-.2em] 
		&R_1^\mathsf{U}(u_i,\varphi^\mathsf{(c)}(u_i),\mathbb{C}^\mathsf{(c)}_{i,j})\leq \rho_3,\tag{38h}\label{equ. PF MU C6}\\[-.2em]
		&R_2^\mathsf{U}(\bm{c}_k,\varphi^\mathsf{(s)}(\bm{c}_k),\mathbb{C}^\mathsf{(s)}_{k,j})\leq \rho_4,\tag{38i}\label{equ. PF MU C7}
	\end{align}
\end{subequations}

\noindent where \( B_{\min} \) and \( B_{\max} \) represent the lower and upper bounds of bandwidth resources that each MU can request, while \( P_{\min} \) and \( P_{\max} \) denote the minimum and maximum power resources that can be allocated to each MU. $ \rho_3,\rho_4\in (0, 1] $ are risk thresholds. In $ \bm{\mathcal{F}^\mathsf{U}} $, constraints (\ref{equ. PF MU C1}) and (\ref{equ. PF MU C2}) guarantee the feasibility of communication and sensing M2M matchings  $\varphi^\mathsf{(c)}$ and $\varphi^\mathsf{(s)}$. Constraint (\ref{equ. PF MU C3}) ensures that the obtained valuation of $u_i$ exceeds its payments, while constraints (\ref{equ. PF MU C4}) and (\ref{equ. PF MU C5}) ensure that the communication and sensing service quality of each MU or sensing coalition meets the corresponding requirements. Constraints (\ref{equ. PF MU C8}) and (\ref{equ. PF MU C9}) guarantee the bandwidth and power resources requested by each client for services are constrained within a certain range. Also, constraints (\ref{equ. PF MU C6}) and (\ref{equ. PF MU C7}) specify the risks of each MU with their tractable forms obtained in Appendix B.

The offline trading mode thus presents a \textit{multi-objective optimization (MOO) problem} involving both $\bm{\mathcal{F}^\mathsf{S}}$ and $\bm{\mathcal{F}^\mathsf{U}}$, where the conflicting utilities of different parties make designing a win-win solution a complex task. Furthermore, the probabilistic nature of risks further complicates the problem. To address this MOO problem, we propose offRFW$^2$M, which obtains long-term contracts while achieving mutually beneficial expected utilities and managing risks for both parties. 

\subsection{Solution Design: Structure of offRFW$^2$M}

In a nutshell, offRFW$^2$M enables BSs and clients to negotiate the quantity and pricing of bandwidth and power resources for two distinct service types: \textit{(i)} individual MUs engage in resource trading for communication services; \textit{(ii)} coalitions participate in resource trading for sensing services. Given the interdependencies between resource demands and pricing strategies, offRFW$^2$M is an iterative method that we describe below, with its details outlined in Alg.~\ref{Alg:1}.
\begin{algorithm}[t!] 
	{\scriptsize \setstretch{0.4}\caption{{Proposed Role-Friendly Win-Win Matching for Offline Trading (offRFW$^2$M)}\label{Alg:1}}
		\LinesNumbered 
		\textbf{Initialization:} $ \mathcal{X} \leftarrow 1 $, $ \mathbbm{c}^\mathsf{(c),Pay}_{i,j}\left\langle 1 \right\rangle \leftarrow p^\mathsf{\min}_{i,j}$, $ \mathbbm{c}^\mathsf{(s),Pay}_{k,j}\left\langle 1 \right\rangle \leftarrow p^\mathsf{\min}_{k,j}$, $\mathbbm{c}^\mathsf{(c),PelU}_{i,j}$, $\mathbbm{c}^\mathsf{(c),PelS}_{i,j}$, $\mathbbm{c}^\mathsf{(s),PelU}_{k,j}$, $\mathbbm{c}^\mathsf{(s),PelS}_{k,j}$, ${flag}_{j} \leftarrow 1 $, $\mathbb{Y}^\mathsf{(c)}\left( u_i \right)\leftarrow \varnothing$, $\mathbb{Y}^\mathsf{(c)}\left( s_{j} \right)\leftarrow \varnothing$, $\mathbb{Y}^\mathsf{(s)}\left( \bm{c}_k \right)\leftarrow \varnothing$, $\mathbb{Y}^\mathsf{(s)}\left( s_{j} \right)\leftarrow \varnothing$\ 
		
		\For{$\forall u_i\in\bm{\mathcal{U}}$}{
		$\bm{c}_k\leftarrow u_i$ forms coalitions based on shared sensing target, where $\bm{c}_k\in\bm{\mathcal{C}}$
		}
		\While{$ \sum_{u_i\in\bm{\mathcal{U}}}{flag}_{i} $ and $\sum_{\bm{c}_k\in\bm{\mathcal{C}}}{flag}_{k}$}{
			\textbf{$ {flag}_{i} \leftarrow {\bf False} $, $ {flag}_{k} \leftarrow {\bf False} $}
			
			\textbf{Calculate:} $\overrightarrow{L^\mathsf{(c)}_i}$ and $\overrightarrow{L^\mathsf{(s)}_k}$ under constraints (\ref{equ. PF MU C4}) - (\ref{equ. PF MU C7})
            
            $ F^\mathsf{(c),\star}_i\left\langle \mathcal{X} \right\rangle\leftarrow \overrightarrow{L^\mathsf{(c)}_i}$, $ F^\mathsf{(s),\star}_k\left\langle \mathcal{X} \right\rangle\leftarrow \overrightarrow{L^\mathsf{(s)}_k}$      
			 $\mathbb{Y}^\mathsf{(c)}\left( u_i \right), B_{i,j}^\mathsf{(c)}\left\langle \mathcal{X} \right\rangle, P_{i,j}^\mathsf{(c)}\left\langle \mathcal{X} \right\rangle, \mathbbm{c}^\mathsf{(c),Pay}_{i,j}\left\langle \mathcal{X} \right\rangle  \leftarrow F^\mathsf{(c),\star}_i\left\langle \mathcal{X} \right\rangle $, $ \mathbb{Y}^\mathsf{(s)}\left( \bm{c}_k \right), B_{k,j}^\mathsf{(s)}\left\langle \mathcal{X} \right\rangle, P_{k,j}^\mathsf{(s)}\left\langle \mathcal{X} \right\rangle, \mathbbm{c}^\mathsf{(s),Pay}_{k,j}\left\langle \mathcal{X} \right\rangle \} \leftarrow F^\mathsf{(s),\star}_k\left\langle \mathcal{X} \right\rangle $

			\If{$ \forall\mathbb{Y}^\mathsf{(c)}\left( u_i \right) \neq \varnothing $ or $ \forall\mathbb{Y}^\mathsf{(s)}\left( \bm{c}_k \right) \neq \varnothing $}{
				\For{$\forall u_i \in \bm{\mathcal{U}}$ }{
				$ u_i $ sends a proposal to $ s_j $, where $s_j\in\mathbb{Y}^\mathsf{(c)}\left( u_i \right)$}
				\For{$\forall \bm{c}_k \in \bm{\mathcal{C}}$ }{
				$ \bm{c}_k $ sends a proposal  to $ s_j $, where $s_j\in\mathbb{Y}^\mathsf{(s)}\left( \bm{c}_k \right)$}
				
				\While{
					$ \Sigma_{u_i\in \bm{\mathcal{U}}}{flag}_{i} > 0 $}{
					$ {\widetilde{\mathbb{Y}}}\left(s_j\right) \leftarrow$ collect proposals from clients
					
					$ \mathbb{Y}^\mathsf{(c)}(s_j)$, $\mathbb{Y}^\mathsf{(s)}(s_j) \leftarrow $ choose MUs from $ {\widetilde{\mathbb{Y}}}\left(s_j\right) $ that can achieve the maximization of the expected utility of BS $s_j$ (i.e., (\ref{equ. PF BS})) by using DP under (\ref{equ. PF BS C3}), (\ref{equ. PF BS C4}), (\ref{equ. PF BS C5}), and (\ref{equ. PF BS C6})
					
					$ s_j $ temporally accepts the clients in $ \mathbb{Y}^\mathsf{(c)}(s_j) $ and $ \mathbb{Y}^\mathsf{(s)}(s_j) $, and rejects the others
				}
				
				\For{
					$ \forall u_i \in \mathbb{Y}^\mathsf{(c)}\left( s_j \right) $
				}{
					\If{$ u_i $ is rejected by $ s_j $, $V^\mathsf{(c)}_{i,j}\ge \mathbbm{c}^\mathsf{(c),Pay}_{i,j}$ and constraints (\ref{equ. PF MU C4}) and (\ref{equ. PF MU C6}) are met}{
						$ \mathbbm{c}^\mathsf{(c),Pay}_{i,j}\left\langle {\mathcal{X} + 1} \right\rangle \leftarrow \min\left\{ \mathbbm{c}^\mathsf{(c),Pay}_{i,j}\left\langle \mathcal{X} \right\rangle + \mathrm{\Delta}p~,{ V}^\mathsf{(c)}_{i,j} \right\} $}
					\Else{$ \mathbbm{c}^\mathsf{(c),Pay}_{i,j}\left\langle {\mathcal{X} + 1} \right\rangle \leftarrow \mathbbm{c}^\mathsf{(c),Pay}_{i,j}\left\langle \mathcal{X} \right\rangle $}
				}
				
				\For{
					$ \forall \bm{c}_k \in \mathbb{Y}^\mathsf{(s)}\left( s_j \right) $
				}{
					\If{$ u_i $ is rejected by $ s_j $, $V^\mathsf{(s)}_{k,j}\ge \mathbbm{c}^\mathsf{(s),Pay}_{k,j}$ and constraints (\ref{equ. PF MU C5}) and (\ref{equ. PF MU C7}) are met}{
						$ \mathbbm{c}^\mathsf{(s),Pay}_{k,j}\left\langle {\mathcal{X} + 1} \right\rangle \leftarrow \min\left\{ \mathbbm{c}^\mathsf{(s),Pay}_{k,j}\left\langle \mathcal{X} \right\rangle + \mathrm{\Delta}p~,{ V}^\mathsf{(s)}_{k,j} \right\} $}
					\Else{$ \mathbbm{c}^\mathsf{(s),Pay}_{k,j}\left\langle {\mathcal{X} + 1} \right\rangle \leftarrow \mathbbm{c}^\mathsf{(s),Pay}_{k,j}\left\langle \mathcal{X} \right\rangle $}
				}
                $ p_{i,\mathbbm{n}}^\mathsf{(c)}\leftarrow \mathbbm{c}^\mathsf{(c),Pay}_{i,j}\left\langle \mathcal{X}+1 \right\rangle, p_{i,\mathbbm{n}}^\mathsf{(c)} \in F^\mathsf{(c),\star}_{i}\left\langle \mathcal{X} \right\rangle$, $
                    p_{k,\mathbbm{m}}^\mathsf{(s)}\leftarrow \mathbbm{c}^\mathsf{(c),Pay}_{i,j}\left\langle \mathcal{X}+1 \right\rangle, p_{i,\mathbbm{n}}^\mathsf{(c)} \in F^\mathsf{(c),\star}_{i}\left\langle \mathcal{X} \right\rangle$
                    
					\If{$\mathcal{X}\le2$ and there exists $F^\mathsf{(c),\star}_{i}\left\langle \mathcal{X}-1 \right\rangle \neq F^\mathsf{(c),\star}_{i}\left\langle \mathcal{X} \right\rangle $ or $F^\mathsf{(s),\star}_{k}\left\langle \mathcal{X}-1 \right\rangle \neq F^\mathsf{(s),\star}_{k}\left\langle \mathcal{X} \right\rangle $}{
					$ {flag}_{i} \leftarrow {\bf True} $, $ {flag}_{k} \leftarrow {\bf True} $	}	\      
				$ \mathcal{X}\leftarrow \mathcal{X}+1 $
			}

		}

		$\varphi^\mathsf{(c)}(s_j)\leftarrow\mathbb{Y}^\mathsf{(c)}(s_j)$, $\varphi^\mathsf{(c)}(u_i)\leftarrow \mathbb{Y}^\mathsf{(c)}(u_i)$,
		$\varphi^\mathsf{(s)}(s_j)\leftarrow\mathbb{Y}^\mathsf{(s)}(s_j)$, $\varphi^\mathsf{(s)}(\bm{c}_k)\leftarrow \mathbb{Y}^\mathsf{(s)}(\bm{c}_k)$ , $\mathcal{X} \leftarrow \mathcal{X}-1$

		\textbf{Return:} $\mathbb{C}_{i,j}^\mathsf{(c)} =\{ B_{i,j}^\mathsf{(c)}\left\langle \mathcal{X} \right\rangle, P_{i,j}^\mathsf{(c)}\left\langle \mathcal{X} \right\rangle, \mathbbm{c}^\mathsf{(c),Pay}_{i,j}\left\langle \mathcal{X} \right\rangle, 		\mathbbm{c}^\mathsf{(c),PelU}_{i,j}, \mathbbm{c}^\mathsf{(c),PelS}_{i,j} \}$, $\mathbb{C}_{k,j}^\mathsf{(s)} =\{ B_{k,j}^\mathsf{(s)}\left\langle \mathcal{X} \right\rangle, P_{k,j}^\mathsf{(s)}\left\langle \mathcal{X} \right\rangle, \mathbbm{s}^\mathsf{(s),Pay}_{k,j}\left\langle \mathcal{X} \right\rangle, \mathbbm{c}^\mathsf{(s),PelU}_{k,j}, \mathbbm{c}^\mathsf{(s),PelS}_{k,j} \} $}
\end{algorithm}

\noindent~\textbf{Step 1. Initialization} (line 1, Alg.~\ref{Alg:1}): The negotiation process involves multiple rounds, indexed by $\mathcal{X}$. The initial (i.e., $\mathcal{X}=1$) payment of communication services by an individual MU $u_i$ is set as $p^\mathsf{(c)}_{i,j}\left\langle 1 \right\rangle = p^\mathsf{(c),\min}_{i,j}$, while the initial payment for sensing services by each coalition $\bm{c}_k$ is set as $p^\mathsf{(s)}_{k,j}\left\langle 1 \right\rangle = p^\mathsf{(s),\min}_{k,j}$ (line 1). We also define the set $\mathbb{Y}^\mathsf{(c)}(u_i)$ to include the BSs that $u_i$ is interested in for communication service, and $\mathbb{Y}^\mathsf{(c)}(s_j)$ to capture the MUs temporarily selected by $s_j$ for communication service. Similarly, set $\mathbb{Y}^\mathsf{(s)}(\bm{c}_k)$ comprises the BSs that $\bm{c}_k$ is interested in for sensing service, while $\mathbb{Y}^\mathsf{(s)}(s_j)$ covers the sensing coalitions temporarily selected by $s_j$.

\noindent~\textbf{Step 2. Establishment of MU coalitions and preference lists} (lines 2-6, Alg.~\ref{Alg:1}): Before the matching process, MUs form sensing coalitions based on their sensing targets. At the beginning of each round $\mathcal{X}$, each individual MU announces its resource requests for communication service to BSs according to its preference list, defined in Definition \ref{def 5}.

\begin{Defn}(Preference List of each MU in offRFW$^2$M)\label{def 5}
Consider the set of feasible long-term contract solutions for MU $ u_i$ that satisfy (\ref{equ. PF MU C8}), (\ref{equ. PF MU C9}), and (\ref{equ. PF MU C4}) as {\footnotesize \( \bm{\mathcal{C}^\mathsf{(c),F}}_{i} = \{ \bm{F^{(c)}}_{i,1}, \dots, \bm{F^{(c)}}_{i,\mathbbm{n}}, \dots, \bm{F^{(c)}}_{i,|\bm{\mathcal{C}^\mathsf{(c),F}_{i}}|} \} \)}, where {\footnotesize\( \bm{F^{(c)}}_{i,\mathbbm{n}} = \{s^\mathsf{(c),F}_j, \mathbb{C}^\mathsf{(c),F}_{i,\mathbbm{n}} \} \)} represents a feasible solution, which consists of: (i) the BS {\footnotesize\( s^\mathsf{(c),F}_j \in \bm{\mathcal{S}} \)} to which the request should be sent, (ii) the feasible five-tuple contract item {\footnotesize\( \mathbb{C}^\mathsf{(c),F}_{i,n}= \{ B_{i,\mathbbm{n}}^\mathsf{(c)}, P_{i,\mathbbm{n}}^\mathsf{(c)}, p^\mathsf{(c)}_{i,\mathbbm{n}}, \mathbbm{c}^\mathsf{(c),PelS}_{i,j}, \mathbbm{c}^\mathsf{(c),PelU}_{i,j} \} \)}, where {\footnotesize\( B_{i,\mathbbm{n}}^\mathsf{(c)} \)} and {\footnotesize\( P_{i,\mathbbm{n}}^\mathsf{(c)} \)} represent the required bandwidth and power resources, and {\footnotesize\( p^\mathsf{(c)}_{i,\mathbbm{n}} \)} represents the bid for the current solution, initialized with {\footnotesize$p_{i,j}^\mathsf{(c),min}$}. The preference list \( \overrightarrow{L^\mathsf{(c)}_i} \) of an MU \( u_i \) regarding feasible solutions {\footnotesize\( \bm{F^{(c)}}_{i,\mathbbm{n}} \in \bm{\mathcal{C}^\mathsf{(c),F}}_{i} \)} is a vector of tuples, sorted in non-ascending order based on their expected utility:
\begin{equation}
\hspace{-3mm}
	\begin{aligned}
			\overrightarrow{L^\mathsf{(c)}_i} = [\bm{F^{(c)}}_{i,\mathbbm{n}}\in \bm{\mathcal{C}^\mathsf{(c),F}}_{i}, \text{sorted in non-ascending order  based on (\ref{equ. expected comm utility})}].
	\end{aligned}
    \hspace{-3mm}
\end{equation}
\end{Defn}
Meanwhile, each coalition announces its resource requests to BSs according to its preference list, defined in Definition \ref{def 6}.

\begin{Defn}(Preference List of Sensing Coalition in offRFW$^2$M)\label{def 6} Consider the set of feasible long-term contract solutions for coalition $ \bm{c}_k$ that satisfy  (\ref{equ. PF MU C8}), (\ref{equ. PF MU C9}), and (\ref{equ. PF MU C5})  as {\footnotesize\( \bm{\mathcal{C}^\mathsf{(s),F}}_{k} = \{ \bm{F^{(s)}}_{k,1}, \dots, \bm{F^{(s)}}_{k,\mathbbm{m}}, \dots, \bm{F^{(s)}}_{k,|\bm{\mathcal{C}^\mathsf{(s),F}_{k}}|} \} \)}, where {\footnotesize\( \bm{F^{(s)}}_{k,\mathbbm{m}} = \{s^\mathsf{(s),F}_j, \mathbb{C}^\mathsf{(s),F}_{k,\mathbbm{m}} \} \)} represents a feasible solution, which consists of: (i) the BS {\footnotesize\( s^\mathsf{(s),F}_j \in \bm{\mathcal{S}} \)} to which the request should be sent, 
(ii) the feasible five-tuple contract item {\footnotesize\( \mathbb{C}^\mathsf{(s),F}_{k,\mathbbm{m}} = \{ B_{k,\mathbbm{m}}^\mathsf{(s)}, P_{k,\mathbbm{m}}^\mathsf{(s)}, p^\mathsf{(s)}_{i,\mathbbm{m}}, \mathbbm{c}^\mathsf{(s),PelS}_{k,j}, \mathbbm{c}^\mathsf{(s),PelU}_{k,j} \} \)}, where {\footnotesize\( B_{k,\mathbbm{m}}^\mathsf{(s)} \)} and {\footnotesize\( P_{k,\mathbbm{m}}^\mathsf{(s)} \)} represent the required bandwidth and power resources, and {\footnotesize\( p^\mathsf{(s)}_{k,\mathbbm{m}} \)} represents the bid for the current solution, initialized with $p_{k,j}^\mathsf{(s),min}$. The preference list \( \overrightarrow{L^\mathsf{(s)}_k} \) of an coalition \( \bm{c}_k \) regarding feasible solutions {\footnotesize\( \bm{F^{(s)}}_{k,\mathbbm{m}} \in \bm{\mathcal{C}^\mathsf{(s),F}}_{k} \)} is a vector of tuples, sorted in non-ascending order based on their expected utility:
\begin{equation}
\hspace{-2.5mm}
	\begin{aligned}
			\overrightarrow{L^\mathsf{(s)}_k} = [\bm{F^{(s)}}_{k,\mathbbm{m}}\in \bm{\mathcal{C}^\mathsf{(s),F}}_{k},  \text{sorted in non-ascending order  based on (\ref{equ. expected sensing utility})}].
	\end{aligned}
    \hspace{-2mm}
\end{equation}
\end{Defn}
\noindent~\textbf{Step 3. Proposal of clients} (lines 7-12, Alg.~\ref{Alg:1}): At round \( \mathcal{X} \), each individual MU \( u_i \) and sensing coalition \( \bm{c}_k \) select their \textit{preferred solutions} {\footnotesize\( F^\mathsf{(c),\star}_{i}\left\langle \mathcal{X} \right\rangle \)} and {\footnotesize\( F^\mathsf{(s),\star}_k\left\langle \mathcal{X} \right\rangle \)} from the first elements of their preference lists {\footnotesize\( \overrightarrow{L_i^\mathsf{(c)}} \)} and {\footnotesize\( \overrightarrow{L_k^\mathsf{(s)}} \)}, respectively. They then record the selected BS \( s_j \) in {\footnotesize\( \mathbb{Y}^\mathsf{(c)}(u_i) \)} for communication and {\footnotesize\( \mathbb{Y}^\mathsf{(s)}(\bm{c}_k) \)} for sensing service, and obtain the request information regarding the required bandwidth, power resources, and payments (e.g., for communication service, {\footnotesize$B_{i,j}^\mathsf{(c)}\left\langle \mathcal{X} \right\rangle\leftarrow B_{i,\mathbbm{n}}^\mathsf{(c)}, P_{i,j}^\mathsf{(c)}\left\langle \mathcal{X} \right\rangle\leftarrow P_{i,\mathbbm{n}}^\mathsf{(c)}, \mathbbm{c}^\mathsf{(c),Pay}_{i,j}\left\langle \mathcal{X} \right\rangle\leftarrow p_{i,\mathbbm{n}}^\mathsf{(c)}$}). Each client then transmits its solution to the BSs in {\footnotesize\( \mathbb{Y}^\mathsf{(c)}(u_i)\)} and {\footnotesize\( \mathbb{Y}^\mathsf{(s)}(\bm{c}_k) \)}, initiating the resource negotiations.

\noindent~\textbf{Step 4. Client selection on BSs' side} (lines 13-16, Alg.~\ref{Alg:1}): After collecting the information of individual MUs and sensing coalitions in set ${\widetilde{\mathbb{Y}}}\left(s_j\right)$, each BS $s_j$ solves a two-dimensional 0-1 knapsack problem, which can be solved using dynamic programming (DP) \cite{MY tsc}, to determine a temporary selection of MUs denoted by {\footnotesize$\mathbb{Y}^\mathsf{(c)}(s_j)$} and sensing coalitions {\footnotesize$\mathbb{Y}^\mathsf{(s)}(s_j)$}, where {\footnotesize$\mathbb{Y}^\mathsf{(c)}(s_j)$} and {\footnotesize$\mathbb{Y}^\mathsf{(s)}(s_j)$} belong to {\footnotesize$ {\widetilde{\mathbb{Y}}}\left(s_j\right)$}, maximizing the expected utility of BS \( s_j \) while satisfying constraints (\ref{equ. PF BS C3}), (\ref{equ. PF BS C4}), (\ref{equ. PF BS C5}), and (\ref{equ. PF BS C6}). Then, each $s_j$ reports its decisions to the  MUs and sensing coalitions for the current round.

\noindent~\textbf{Step 5. Decision-making on clients' side} (lines 17-27, Alg.~\ref{Alg:1}): After receiving decisions from BSs, each MU {\footnotesize\( u_i \in \mathbb{Y}^\mathsf{(c)}(s_j) \)} and each coalition{\footnotesize \( \bm{c}_k \in \mathbb{Y}^\mathsf{(s)}(s_j) \)} evaluate their current solutions {\footnotesize$F^\mathsf{(c),\star}_i\left\langle \mathcal{X} \right\rangle$} and {\footnotesize$F^\mathsf{(s),\star}_k\left\langle \mathcal{X} \right\rangle$} . The payment for an individual MU \( u_i \) and coalition \( \bm{c}_k \) remains unchanged if any of the following conditions are met:
\textit{(i)} \( u_i \) or \( \bm{c}_k \) is accepted by \( s_j \); 
\textit{(ii)} the current payment {\footnotesize\( \mathbbm{c}^\mathsf{(c),Pay}_{i,j} \)} or {\footnotesize\( \mathbbm{c}^\mathsf{(s),Pay}_{k,j} \)} equals its valuation {\footnotesize\( V^\mathsf{(c)}_{i,j} \)}; 
\textit{(iii)} constraints (\ref{equ. PF MU C4}), (\ref{equ. PF MU C5}), (\ref{equ. PF MU C6}), and (\ref{equ. PF MU C7}) are not met. 
Otherwise, \( u_i \) or \( \bm{c}_k \) will increase its bid associated with the current solution {\footnotesize$F^\mathsf{(c),\star}_i\left\langle \mathcal{X} \right\rangle$ and $F^\mathsf{(s),\star}_k\left\langle \mathcal{X} \right\rangle$} for \( s_j \) in the next round by $\Delta p$ to enhance its competitiveness in the market.

\noindent~\textbf{Step 6. Repeat} (lines 4-30, Alg.~\ref{Alg:1}): If all the preferred solutions stay unchanged in two consecutive rounds, the matching process terminates. We use $ \sum_{u_i\in\bm{\mathcal{U}}}{flag}_{i}={\bf False} $ and $\sum_{\bm{c}_k\in\bm{\mathcal{C}}}{flag}_{k}={\bf False} $ to capture this (line 5). Otherwise, the next round starts, re-iterating the above steps (lines 4-30).

\noindent~\textbf{Computational complexity:} The computational complexity of our proposed offRFW$^2$M depends on the total number of rounds involved in Alg. 1 (denoted by \( \mathcal{X}^{\mathsf{max}} \)), the overbooked resources {\footnotesize\( (1 + O_j^\mathsf{B}) B_j \)} and {\footnotesize\( (1 + O_j^\mathsf{Pow}) P_j \)}, as well as the number of clients sending requests to BS \( s_j \) in the $ \mathcal{X}^\mathsf{\text{th}} $ round, denoted as {\footnotesize\( |{\widetilde{\mathbb{Y}}}\left(s_j\right)|_{\mathcal{X}} \)}. In particular, the overall complexity of offRFW$^2$M for each BS \( s_j \) is
{\footnotesize $\sum_{\mathcal{X}=1}^\mathsf{\mathcal{X}^{\mathsf{max}}} \mathcal{O}\left( |{\widetilde{\mathbb{Y}}}\left(s_j\right)|_{\mathcal{X}} \times (1 + O_j^\mathsf{B})B_j \times (1 + O_j^\mathsf{Pow}) P_j \right)$}.


\subsection{Solution Characteristics and Key Properties}
As offRFW$^2$M is deployed prior to practical transactions, our focus is on maximizing the utilities of clients and BSs while controlling potential risks. This differentiates our approach significantly from conventional matching mechanisms, which primarily emphasize immediate resource allocation without considering long-term contractual stability and risks. {We next define the concept of \textit{blocking pair}\footnote{A Blocking pair refers to any two entities/participants in a trading market that are not matched to each other but would prefer to be paired together, rather than remaining in their current matches. This concept is a cornerstone for evaluating the stability and efficiency of a given matching solution\cite{MY tsc}.}, representing a key factor that may lead to the instability/inefficiency of a matching.}

{
\begin{Defn}(Blocking Pairs for Communication and Sensing Services in offRFW$^2$M)
Under a given matching \( \varphi^\mathfrak{(X)} \), a client \( \mathsf{a} \) (representing either a MU \( u_i \) for communication when \( \mathfrak{(X)} = \mathsf{(c)} \) or a coalition \( \bm{c}_k \) for sensing when \( \mathfrak{(X)} = \mathsf{(s)} \)), a BS set \( \mathbb{S} \subseteq \bm{\mathcal{S}} \), and a contract set \( \mathbb{C}^\prime \), denoted by the triplet \( \left(\mathsf{a}; \mathbb{S}; \mathbb{C}^\prime\right) \), may form one of the following two types of blocking pairs.

\noindent \textbf{Type 1 blocking pair:} The pair satisfies the following two conditions:

    \noindent $\bullet$ The client $\mathsf{a}$ prefers the BS set $ \mathbb{S} $ over its currently matched set $ \varphi^\mathfrak{(X)}(\mathsf{a}) $:
    \begin{equation}\label{equ. 41}
    \resizebox{0.38\textwidth}{!}{$
        \mathbb{E}\left[ u^{\mathfrak{(X)},\mathsf{U}}(\mathsf{a},\mathbb{S},\mathbb{C}^\prime) \right] >
        \mathbb{E}\left[ u^{\mathfrak{(X)},\mathsf{U}}(\mathsf{a}, \varphi^\mathfrak{(X)}(\mathsf{a}), \mathbb{C}^\mathfrak{(X)}_{i,j}) \right]$}
    \end{equation}
    
    \noindent $\bullet$ Every BS $s_j \in \mathbb{S}$ prefers to reallocate its service from current matches to include $\mathsf{a}$, i.e., there exists a subset $\varphi^\mathfrak{(X)\prime}(s_j)$ of currently matched agents to be evicted, such that:
    \begin{equation}\label{equ. 42}
    \resizebox{0.44\textwidth}{!}{$
        \mathbb{E}\left[u^{\mathfrak{(X)},\mathsf{S}}(s_j, \{\varphi^\mathfrak{(X)}(s_j) \setminus \varphi^\mathfrak{(X)\prime}(s_j)\} \cup \{\mathsf{a}\}, \mathbb{C}^\prime)\right] >
        \mathbb{E}\left[u^{\mathfrak{(X)},\mathsf{S}}(s_j, \varphi^\mathfrak{(X)}(s_j), \mathbb{C}^\mathfrak{(X)}_{i,j})\right]$}
    \end{equation}

\noindent \textbf{Type 2 blocking pair:} The pair satisfies the following two conditions:

\noindent $\bullet$ The client $\mathsf{a}$ prefers the BS set $ \mathbb{S} $ over its currently matched set, as shown in (\ref{equ. 41}).
    
\noindent $\bullet$  Each BS $s_j \in \mathbb{S}$ prefers to additionally serve $\mathsf{a}$ while maintaining its current matches:
    \begin{equation}\label{equ. 44}\resizebox{0.44\textwidth}{!}{$
        \mathbb{E}\left[u^{\mathfrak{(X)},\mathsf{S}}(s_j, \varphi^\mathfrak{(X)}(s_j) \cup \{\mathsf{a}\}, \mathbb{C}^\prime)\right] >
        \mathbb{E}\left[u^{\mathfrak{(X)},\mathsf{S}}(s_j, \varphi^\mathfrak{(X)}(s_j), \mathbb{C}^{\mathfrak{(X)}}_{i,j})\right]$}
    \end{equation}

\end{Defn}
\vspace{0 mm}}

{ In essence, a Type 1 blocking pair compromises matching stability by incentivizing a BS to reallocate its resources to an alternative set of clients that offer higher utility. A Type 2 blocking pair similarly leads to instability, as the BS retains residual resources that could be used to serve additional clients and  enhance its utility. These two types of blocking pairs are the pillars of characterizing the key properties of offRFW$^2$M in the subsequent analysis.}

\begin{Prop}\label{Prop 1}(Individual Rationality of offRFW$^2$M) The proposed offRFW$^2$M mechanism ensures individual rationality for BSs, individual MUs, and sensing coalitions as follows:
	
	\noindent
	$\bullet$ For each BS: \textit{(i)} the bandwidth and power resources of BS $s_j$ booked to matched clients $\varphi^\mathsf{(c)}\left(s_j\right)$ and coalitions $\varphi^\mathsf{(s)}\left(s_j\right)$ does not exceed $(1+O_j^\mathsf{B})B_j$ and $(1+O_j^\mathsf{Pow})P_j$ after applying overbooking, i.e., constraints (\ref{equ. PF BS C3}) and (\ref{equ. PF BS C4}) are met; \textit{(ii)} the risks associated with each BS are maintained within a certain acceptable range, i.e., constraint (\ref{equ. PF BS C5}) and (\ref{equ. PF BS C6}) are satisfied.
	
	\noindent
	$\bullet$ For each client (i.e., each MU and each coalition): \textit{(i)} The value obtained by each client is at least equal to the payment it makes, ensuring that (\ref{equ. PF MU C3}) is met; \textit{(ii)} the risks associated with each client are  acceptable, satisfying  (\ref{equ. PF MU C4})-(\ref{equ. PF MU C7}).
\end{Prop}

\begin{Prop}(Fairness of offRFW$^2$M): The proposed offRFW$^2$M ensures fairness by preventing the formation of Type 1 blocking pairs, ensuring that clients are satisfied with their matched BSs and no BS is incentivized to reallocate its resources to a different set of clients.\end{Prop}

\begin{Prop}(Non-wastefulness of offRFW$^2$M): offRFW$^2$M guarantees non-wastefulness by preventing the formation of Type 2 blocking pairs, ensuring that BSs efficiently utilize their resources without leaving a surplus.\end{Prop}
\begin{Prop}(Strong Stability of offRFW$^2$M)\label{Prop 4} The proposed offRFW$^2$M achieves strong stability by ensuring that the matching remains individually rational, fair, and non-wasteful.
\end{Prop}

\begin{Prop}(Stability of Sensing Coalitions in offRFW$^2$M)
In offRFW$^2$M, each sensing coalition $\bm{c}_k$ is stable, that is 
	
\noindent $\bullet$ For each MU $u_i$ in $\bm{c}_k$, its expected utility is above $u^\mathsf{(s)}_\mathsf{\min}$. 

\noindent $\bullet$ For each MU $u_i$ in sensing coalition $\bm{c}_k$, the expected utility obtained by joining the coalition \( \bm{c}_k \) (shown in (\ref{equ. expected sensing utility})) is lower than the expected utility when trading as an individual.
\end{Prop}

{Note that for the MOO problem defined by $ \bm{\mathcal{F}^\mathsf{U}} $ and $ \bm{\mathcal{F}^\mathsf{S}} $, a Pareto improvement occurs when the expected social welfare (i.e., the summation of expected utilities of clients and BSs)\cite{RW Matching3} can be increased with another matching result.}
{Subsequently, a matching is considered to be \textit{weakly Pareto optimal} if no other feasible matching leads to a strictly higher total expected social welfare. Unlike strong Pareto optimality, which is reached when it is impossible to improve one individual's utility without making at least one other individual worse off (i.e., any change that benefits one entity would necessarily harm another), this weaker version considers only aggregate utility. 
In other words, even if some participants experience no gain, or a slight loss, as long as the total utility cannot be improved, the solution is considered weakly Pareto optimal.

This weaker notion is appropriate for our setting because our model already ensures individual rationality (no participant is worse off by participating) and matching stability (no blocking pairs exist). As such, weak Pareto optimality serves as a meaningful, complementary criterion that captures global optimality, aligning with the objectives of many distributed, market-based mechanisms~\cite{RW Matching3}.}

\begin{Prop}(Weak Pareto Optimality of offRFW$^2$M)  offRFW$^2$M is weakly Pareto optimal, ensuring that no other matching can increase the social welfare of the system.
\end{Prop}
Detailed proofs of all the above results are presented in Appendix C.

\section{Effective Backup Win-Win Matching for Online Trading (onEBW$^2$M)}
\noindent The intermittent participation of MUs may hinder the seamless execution of pre-signed long-term contracts during practical transactions, potentially resulting in financial losses for both parties. To address this, we introduce an online trading mode as a complementary mechanism, handling two scenarios: \textit{(i)} When the resource demand at a BS exceeds its available supply, the BS strategically designates certain long-term contract clients as \textit{voluntary contributors} using a greedy-based approach. These clients forgo their services in exchange for compensations. \textit{(ii)} When there exist clients with unmet resource demands --- including voluntary clients and those without long-term contracts --- as well as BSs with surplus resources, we implement the onEBW$^2$M mechanism, which obtains two types of \textit{temporary contracts}: 
\textit{(i)} for communication service, a contract between an individual MU $u_i$ and a BS $s_j$, represented as {\footnotesize$\dot{\mathbb{C}}^\mathsf{(c)}_{i,j}$}, and 
\textit{(ii)} for sensing service, a contract between a sensing coalition $\bm{c}_k$ and a BS $s_j$, represented as {\footnotesize$\dot{\mathbb{C}}^\mathsf{(s)}_{k,j}$}. 
These temporary contracts consist of two terms: the quantity of trading resources {\footnotesize$\dot{\mathbbm{c}}^\mathsf{(c),B}_{i,j}, \dot{\mathbbm{c}}^\mathsf{(c),Pow}_{i,j}, \dot{\mathbbm{c}}^\mathsf{(s),B}_{k,j}, \dot{\mathbbm{c}}^\mathsf{(s),Pow}_{k,j}$}, and the service price {\footnotesize$\dot{\mathbbm{c}}^\mathsf{(c),Pay}_{i,j}, \dot{\mathbbm{c}}^\mathsf{(s),Pay}_{k,j}$} (super-scripts have the same meaning as those in Sec.~\ref{sec:over}).



To identify the MUs and BSs that participate in online trading (i.e., in onEBW$^2$M), we introduce the following notations:

\noindent $\bullet$ $\bm{\mathcal{U}^\prime}$: The set of MUs that trade with BSs during online trading (MUs without long-term contracts and volunteers), $\bm{\mathcal{U}^\prime}\subseteq \bm{\mathcal{U}}$;

\noindent $\bullet$ $\bm{\mathcal{C}^\prime}$: The set of coalitions (without long-term contracts and volunteers) that trade with BSs during online trading, $\bm{\mathcal{C}^\prime}\subseteq \bm{\mathcal{C}}$;

\noindent $\bullet$ $\bm{\mathcal{S}^\prime}$:The set of BSs with available resources beyond those specified in long-term contracts,  $\bm{\mathcal{S}^\prime}\subseteq \bm{\mathcal{S}}$;

\noindent $\bullet$ $\nu^\mathsf{(c)}(s_j)$ and $\nu^\mathsf{(s)}(s_j)$: The set of MUs and coalitions served by BS \( s_j \) for communication and sensing services in online trading mode, where $\nu^\mathsf{(c)}(s_j)$ and $\nu^\mathsf{(s)}(s_j)$$ \subseteq \bm{\mathcal{U}^\prime}$;

\noindent $\bullet$ $\nu^\mathsf{(c)}(u_i)$ and $\nu^\mathsf{(s)}(\bm{c}_k)$: The BS that serves MU \( u_i \) and coalition \( \bm{c}_k \) for communication and sensing services in online trading mode, where $\nu^\mathsf{(c)}(u_i)\subseteq \bm{\mathcal{S}^\prime}$ and $\nu^\mathsf{(s)}(u_i)$ $ \subseteq \bm{\mathcal{S}^\prime}$.

In the online trading mode, we define the practical utility of a client in communication and sensing services as the difference between its obtained valuation and its payments as follows:
\begin{equation}\label{eq:49_n}
	u^\mathsf{U^\prime,(c)}(u_i,\nu^\mathsf{(c)}(u_i),\dot{\mathbb{C}}^\mathsf{(c)}_{i,j})=V^\mathsf{(c)}_{i,j}-\dot{\mathbbm{c}}^\mathsf{(c),Pay}_{i,j},
\end{equation}
\begin{equation}\label{eq:49_n2}
	u^\mathsf{U^\prime,{(s)}}(\bm{c}_{k},\nu^\mathsf{(s)}(\bm{c}_{k}),\dot{\mathbb{C}}^\mathsf{(s)}_{k,j})=V^\mathsf{(s)}_{k,j}-\dot{\mathbbm{c}}^\mathsf{(s),Pay}_{k,j}.
\end{equation}
Similarly, the utility of BS $s_j \in \bm{\mathcal{S}^\prime}$ for communication and sensing services is calculated via its received payments as 
\begin{equation}\label{eq:50_n}
	u^\mathsf{S^\prime,{(c)}}(s_j,\nu^\mathsf{(c)}(s_j),\dot{\mathbb{C}}^\mathsf{(c)}_{i,j})=\sum_{u_i\in\nu^\mathsf{(c)}(s_j)}\dot{\mathbbm{c}}^\mathsf{(c),Pay}_{i,j},
\end{equation}
\begin{equation}\label{eq:50_n2}
	u^\mathsf{S^\prime,{(s)}}(s_j,\nu^\mathsf{(s)}(s_j),\dot{\mathbb{C}}^\mathsf{(s)}_{k,j})=\sum_{u_i\in\nu^\mathsf{(s)}(s_j)}\dot{\mathbbm{c}}^\mathsf{(s),Pay}_{k,j}.
\end{equation}

During online trading, the objective of each client and BS is \textit{to maximize its practical utility (e.g., \eqref{eq:49_n},\eqref{eq:49_n2} similar to~\eqref{equ. PF MU}, and summing \eqref{eq:50_n} and \eqref{eq:50_n} similar to \eqref{equ. PF BS}), resampling the MOO problem in Sec.~\ref{sec:probForm}.} Subsequently, for brevity, since onEBW$^2$M is similar to offRFW$^2$M, the detailed formulation of the optimization for onEBW$^2$M, the definitions of matchings $\nu^\mathsf{(c)}(.)$ and $\nu^\mathsf{(s)}(.)$, the solution design and its characteristics are provided in Appendix D.

\section{Evaluations}
\noindent {We conduct experiments under both synthetic  (Sec.~\ref{sec:Synt}) and real-world (Sec.~\ref{sec:real}) settings to evaluate the performance of our framework.} For clarity, our future resource bank framework for ISAC is abbreviated as "FRBank". {For synthetic simulations, BSs are deterministically deployed to form symmetric network topologies within a square area of size $800\ \mathrm{m} \times 800\ \mathrm{m}$. Specifically, we consider two configurations with 5 and 7 BSs, respectively, where BSs are placed at fixed locations such as the four corners, the center, and optionally the edge midpoints. MUs and sensing targets are uniformly distributed within the same area.
} For real-world experiments, we employ the EUA dataset\cite{EUA dataset}, which provides information about BSs and MUs in Melbourne metropolitan area of Australia, covering over 9,000 km$^2$. {Since the exact target locations are not specified in this dataset, sensing targets are distributed across the coverage area uniformly at random.} We also adopt the following parameters: $|\bm{\mathcal{Q}}|=8$, $N^\mathsf{T}_j\in[8,16]$, $N^\mathsf{R}_i\in[4,8]$, $B_j\in[80,120]$MHz, $P_j\in[10, 20]$dBW, $R_i^\mathsf{req}\in[0.01,10]$bit/s, $S_i^\mathsf{req}\in[1,100]$, $B_0=180$KHz\cite{RW ISAC1}, and $\mathbbm{a}_i\in[0.64,0.96]$\cite{MY tsc}.
 {We perform the Monte Carlo method in both Sec.~\ref{sec:Synt} and Sec.~\ref{sec:real}, where each data point in the figures represents the average result over 100 independent trials.}

\subsection{Benchmark Methods and Evaluation Metrics}
We first consider two conventional resource trading benchmarks, each relying on a single trading mode.

\noindent $\bullet$ \textbf{Conventional online resource trading (ConOnline)}\cite{RW Matching3}: This method relies solely on online resource trading. It conducts M2M matching between BSs and MUs in every practical transaction for communication and sensing services.

\noindent $\bullet$ \textbf{Conventional offline resource trading (ConOffline)}\cite{future 1}: This method relies on offline resource trading. It conducts M2M matching to obtain long-term contracts between BSs and MUs for communication and sensing services, which are executed during practical transactions without any backup mechanisms.

We also consider two hybrid trading benchmarks from the broader networking area  to verify the necessity of integrating different trading modes, sensing coalitions, and overbooking in our method, tailored for dynamic ISAC networks.

\noindent $\bullet$ \textbf{Hybrid trading for resource provision (Hybrid)}\cite{MY tsc}: This method integrates the ConOffline and ConOnline methods but does not incorporate overbooking or sensing coalitions.

\noindent $\bullet$ \textbf{Hybrid trading for resource provision with overbooking (HybridO)}: This method is similar to the above method but considers overbooking when determining long-term contracts.

We also incorporate a greedy resource provisioning method.

\noindent $\bullet$ \textbf{Greedy-driven resource provisioning (Greedy)}\cite{MY tsc}: This method is an online trading where BSs select MUs offering  highest payments, while satisfying its resource constraints.

To quantify the performance, we consider the metrics below:

\noindent $\bullet$ \textbf{Social welfare:} The summation of utilities of MUs and BSs for both communication and sensing services (i.e., summing (\ref{equ. comm utility}), (\ref{equ. sensing utility}), (\ref{equ. comm BS U}), and (\ref{equ. sensing BS U})). Here, the utility of MUs for sensing services captures the profit obtained by each MU in its coalition.

\noindent $\bullet$ \textbf{Utility of MUs and BSs}: The respective utilities received by MUs and BSs in the trading process.

{\noindent $\bullet$ \textbf{Running time (RT)}: The reported delay (in milliseconds) corresponds to the latency of the onsite resource trading  process, capturing the real-time execution efficiency of the methods. In our framework, as this process is triggered upon the actual arrival of service demands (in online trading mode), this metric captures the overall perceived latency. This is because the offline pre-matching stage (offRFW$^2$M) is executed in advance and thus excluded from the reported runtime measurements.}

\noindent $\bullet$ \textbf{Number of interactions (NI)}: The number of clients-to-BSs interactions to obtain matchings, measuring the  overhead.

\noindent $\bullet$ \textbf{Delay caused by interactions between BSs and clients (DIBC)}: The duration of decision-making (in ms). We presume the delay of each interaction to be $[1, 15]$ ms\cite{MY tsc,E2E1,E2E2}.

\noindent $\bullet$ \textbf{Energy consumption incurred by interactions between BSs and clients (ECIBC)}: The energy (in Watts) consumed during the decision-making process. To obtain it, we assume that the transmit power of MUs lies in $[0.2, 0.4]$ watts, while the transmit power of BSs lies in $[6, 20]$ Watts \cite{MY tsc,E2E1,E2E2}.

\subsection{Synthetic Experiments}\label{sec:Synt}
\subsubsection{Social Welfare and Utility of MUs/BSs}
\begin{figure}[]
	\centering
	\setlength{\abovecaptionskip}{-1 mm}
	\includegraphics[width=1\columnwidth]{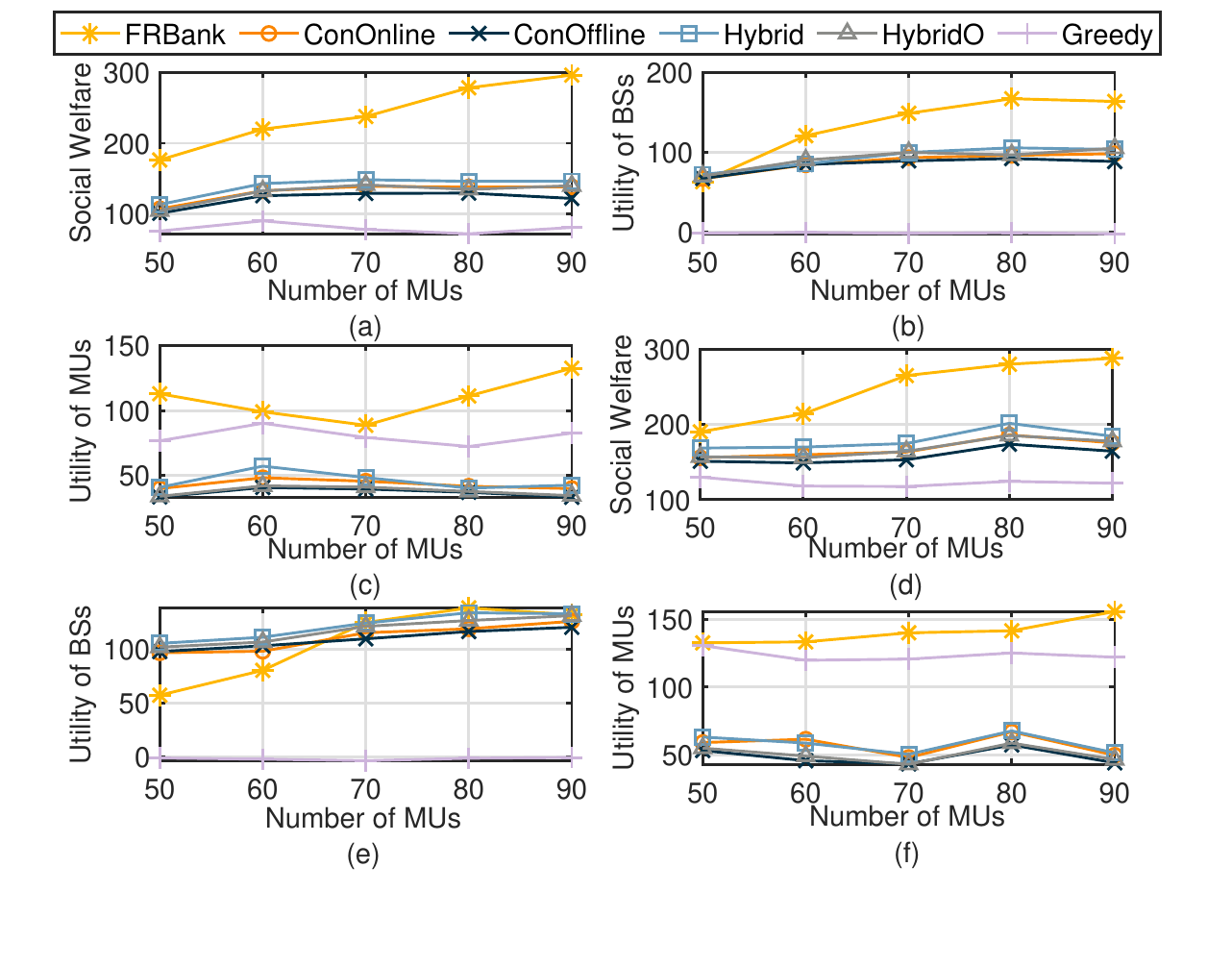}
	\caption{Performance comparisons in terms of social welfare, utility of BSs, and utility of MUs. We consider 5 BSs in (a)-(c), and 7 BSs in (d)-(f).}
	\label{SW}
\end{figure}
Social welfare, along with the utilities of MUs and BSs that constitute it, serve as crucial metrics for assessing the effectiveness of our method,  depicted in Fig. \ref{SW} (Figs. \ref{SW}(a)-(c) and \ref{SW}(d)-(f) consider 5 and 7 BSs, respectively, to capture different market scales).

In Fig. \ref{SW}(a), the curves of benchmark methods such as ConOnline, ConOffline, Hybrid, and HybridO, show an upward trend when the number of MUs is below 60. This is because an increase in the number of MUs leads to higher resource demand, thereby improving the utilities of both MUs and BSs (also evident in Figs. \ref{SW}(b) and \ref{SW}(c)). {However, when the number of MUs exceeds 60, these curves tend to stabilize as the client's payment approaches the maximum value that they can achieve. This trend is especially pronounced in Fig. \ref{SW}(c), where MUs' utilities decline and stabilize once the number of MUs surpasses 60.
In contrast, the performance curves of our proposed FRBank exhibit a generally increasing trend, primarily due to its ability to enable additional MUs to extract sensing value by joining sensing coalitions. Notably, in Fig. \ref{SW}(c), the MU utility curve for FRBank initially declines when the number of MUs is below approximately 70. This decrease is attributed to intensified competition, which elevates bidding prices and consequently diminishes individual MU utilities. However, as the number of MUs further increases and bidding approaches its upper limit, newly participating MUs begin to benefit from coalition-based sensing, resulting in a subsequent rise in MU utilities.
Furthermore, the overbooking mechanism embedded in HybridO alleviates bidding competition among MUs more effectively than the baseline Hybrid scheme, albeit at the cost of a slight reduction in overall social welfare. Nonetheless, this trade-off proves beneficial: by accommodating fluctuations in resource demand, HybridO significantly reduces interaction overhead, as further illustrated in Fig. \ref{NI}, thereby achieving a more favorable balance between efficiency and robustness.} The well-designed online trading mode, acting as a complementary mechanism, enables both Hybrid and HybridO methods to obtain higher social welfare than ConOffline by accommodating more MUs.
Also, Greedy exhibits the weakest performance, as its design prioritizes maximizing BS revenues over efficient resource utilization.
Further, qualitatively similar results are obtained for a larger market scale with 7 BSs in Figs. \ref{SW}(d)-(f).

\subsubsection{Evaluation on Interaction Overhead}
\begin{figure*}[]
	\centering
	\setlength{\abovecaptionskip}{-2 mm}
	\includegraphics[width=2\columnwidth]{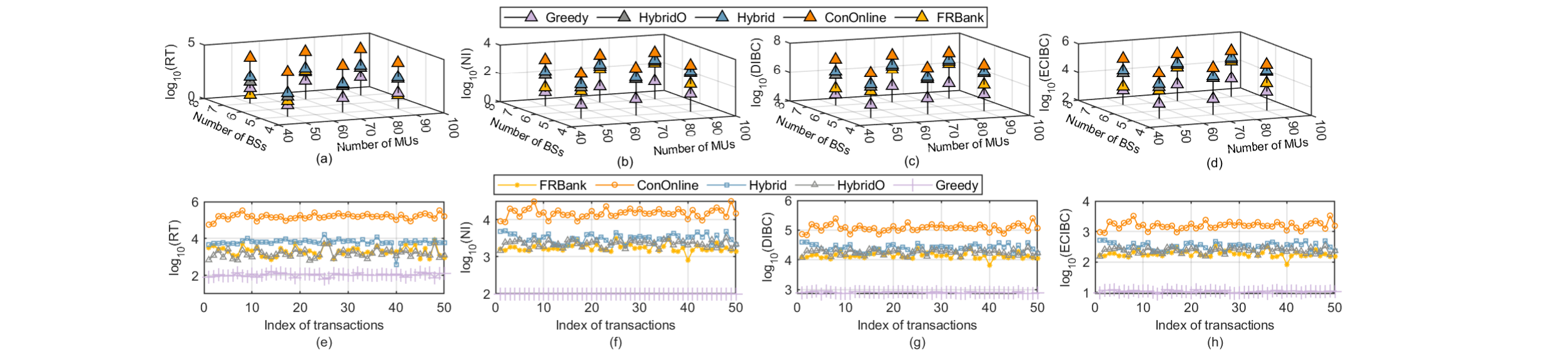}
	\caption{Performance comparisons in terms of RT, NI, DIBC, and ECIBC metrics.}
	\label{NI}
	\vspace{-0.5cm}
\end{figure*}
To provide a quantitative assessment of time and energy efficiency, we consider four key metrics illustrated in Fig.~\ref{NI}, each reflecting the decision-making overhead, including: RT (Figs. \ref{NI}(a) and (e)), NI (Figs. \ref{NI}(b) and (f)), DIBC (Figs. \ref{NI}(c) and (g)), and ECIBC (Figs. \ref{NI}(d) and (h)). The y-axis of subfigures in Fig. \ref{NI} is presented on a logarithmic scale to enhance the visualizations. Also, ConOffline is excluded from the comparison here, as it operates exclusively in the offline trading mode and therefore does not incur real-time decision-making delays or overhead. 

Figs. \ref{NI}(a) and \ref{NI}(e) depict RT performance across different market scales. ConOnline exhibits significantly higher RT compared to other methods, since BSs and MUs require substantial time to determine matching results during each transaction. This issue becomes more pronounced as resource demand intensifies with an increasing number of MUs.
Our FRBank significantly reduces RT compared to ConOnline, as many MUs bypass online trading due to the integration of the overbooking strategy, the sensing coalitions, and the pre-established long-term contracts. {Besides, the well-designed overbooking mechanism in HybridO contributes to a lower RT compared to Hybrid. By ensuring that most contracts are successfully fulfilled in advance, HybridO effectively reduces the reliance on real-time matching, thereby minimizing the number of participants engaged in real-time transactions.} In addition, while Greedy achieves a RT comparable to Hybrid due to the absence of price negotiations, its overall social welfare performance remains suboptimal (see Fig. \ref{SW}).

Figs. \ref{NI}(b)-(d) and Figs. \ref{NI}(f)-(h) illustrate the interaction overhead between BSs and clients, encompassing NI, as well as the time (DIBC) and energy consumption (ECIBC) incurred in reaching matching decisions. The results demonstrate that ConOnline exhibits the worst performance due to its reliance on online trading, leading to significant interaction overhead. In contrast, the structured offline trading mode in Hybrid, HybridO, and FRBank significantly reduces interaction overhead by limiting the number of active participants in each practical transaction. Notably, FRBank and HybridO surpass Hybrid due to the overbooking mechanism, which proactively mitigates the impact of uncertain MU availability and enhances the probability of contract fulfillment.
Moreover, our FRBank outperforms all other methods by incorporating sensing coalitions as strategic market participants. These coalitions act as single decision-making entities, effectively streamlining negotiation complexity and reducing the volume of interactions required in both offline and online trading phases. Such a coalition-based structure not only enhances trading efficiency but also ensures better resource allocation and task execution.
While Greedy exhibits minimal interaction overhead (as seen in Fig. \ref{NI}(a)), this efficiency comes at the expense of weakened social welfare and reduced BS revenue. In summary, FRBank establishes a well-balanced framework, minimizing interaction costs while maintaining economic viability for both BSs and MUs, ultimately leading to superior social welfare and system-wide efficiency.

\subsubsection{Individual Rationality}
\begin{figure} \centering 
\vspace{-3mm}
	\setlength{\abovecaptionskip}{0.1 cm}
	\subfigure[] {
		\label{fig:a} 
		\includegraphics[width=0.45\columnwidth]{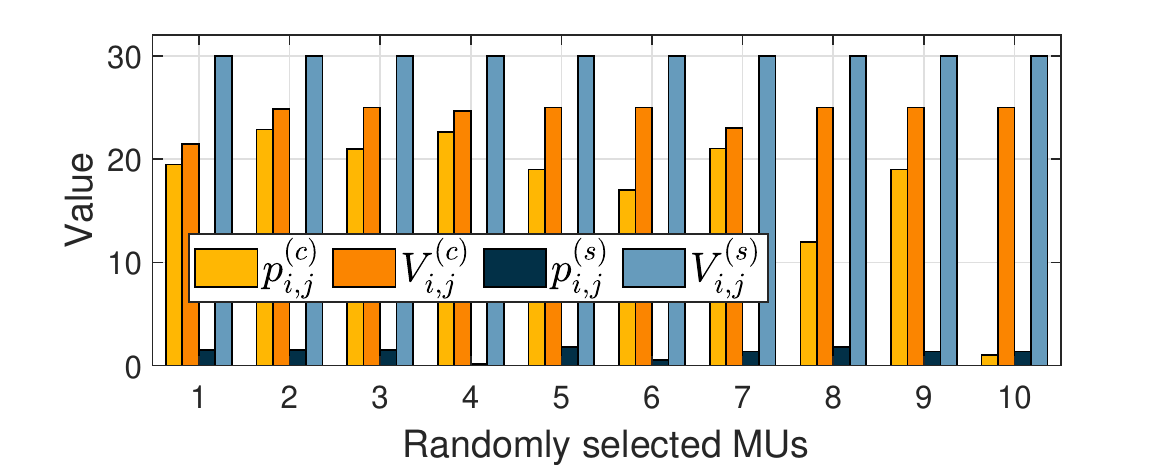} 
	} \hspace{-4mm}
	\subfigure[] { 
		\label{fig:b} 
		\includegraphics[width=0.45\columnwidth]{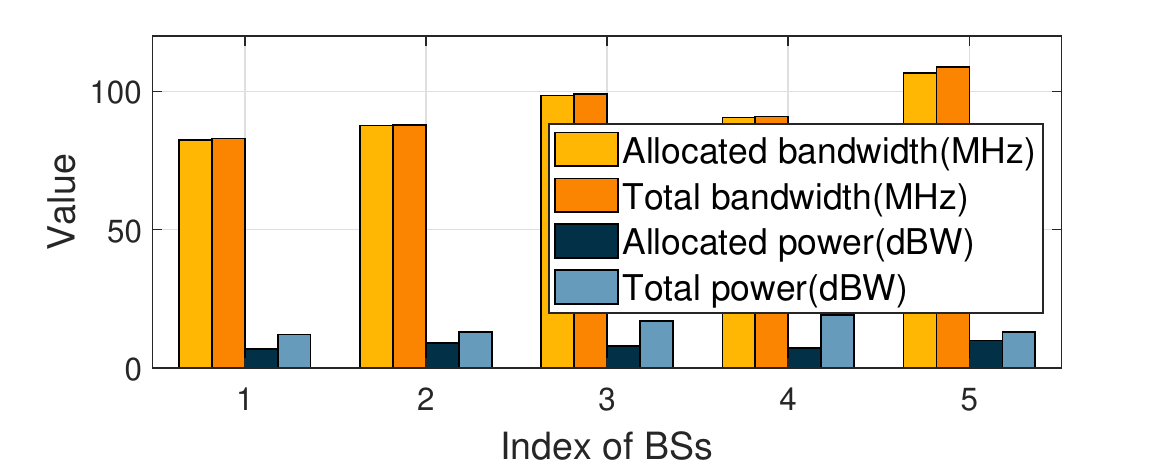} 
	} 
    \vspace{-2mm}
	\caption{Demonstration of individual rationality of MUs and BSs.} 
	\label{fig} 
	\vspace{-0.6 cm}
\end{figure} 
 Fig. \ref{fig:a} quantifies the \textit{value} accrued by MUs (15 MUs are randomely selected from a pool of 50) from accessing both communication and sensing services, along with their corresponding payment obligations. As can be seen, the paid payment of MUs for communication and sensing services never exceed the values they receive. This validates that our designed offRFW$^2$M and onEBW$^2$M in FRBank effectively uphold the individual rationality of MUs.
Additionally, Fig. \ref{fig:b} validates the individual rationality of the 5 BSs, demonstrating that the allocated bandwidth and power for each BS never exceed their total available resources. This is ensured by the \textit{voluntary} selection mechanism during online trading in FRBank, which prevents resource shortages. 

\subsubsection{Performance Analysis of Our Unique Considerations}

This paper introduces several novel aspects that are unexplored in existing ISAC research. Specifically, we highlight three key innovations: sensing coalition, overbooking, and risk analysis. To systematically evaluate their advantages and impact, we conduct dedicated analysis as follows.

\noindent~\textit{(a) Analysis of sensing coalitions (Figs. \ref{SW}-\ref{NI}):} As illustrated in Fig. \ref{SW}, incorporating sensing coalitions allows FRBank to achieve the highest MU utility and social welfare. Moreover, considering sensing coalitions can significantly improve interaction efficiency, as evidenced by superior performance in NI, time, and energy consumption (see Fig. \ref{NI}) during the bargaining and decision-making process. 
This improvement stems from the fact that each coalition appoints a representative MU to engage in negotiations, effectively reducing the number of active participants and simplifying the matching process. 

\begin{figure}[t]
\centering 
\vspace{-1mm}
	\setlength{\abovecaptionskip}{0 cm}
	\subfigure[] {
		\label{overbook} 
		\includegraphics[width=0.45\columnwidth]{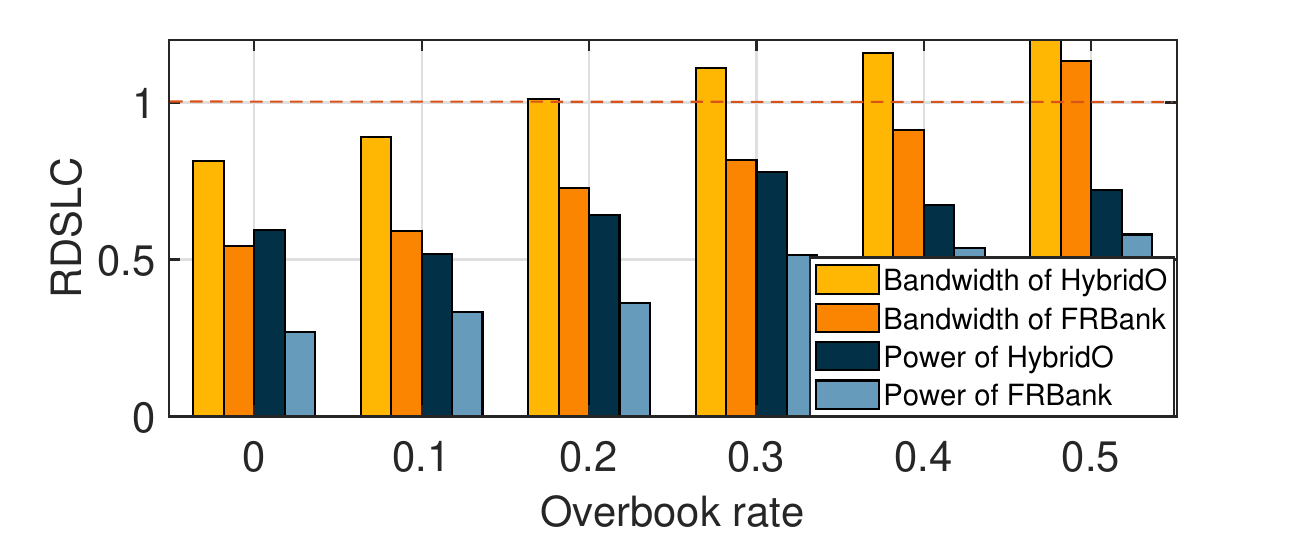} 
	} \hspace{-4mm}
	\subfigure[] { 
		\label{risk} 
		\includegraphics[width=0.45\columnwidth]{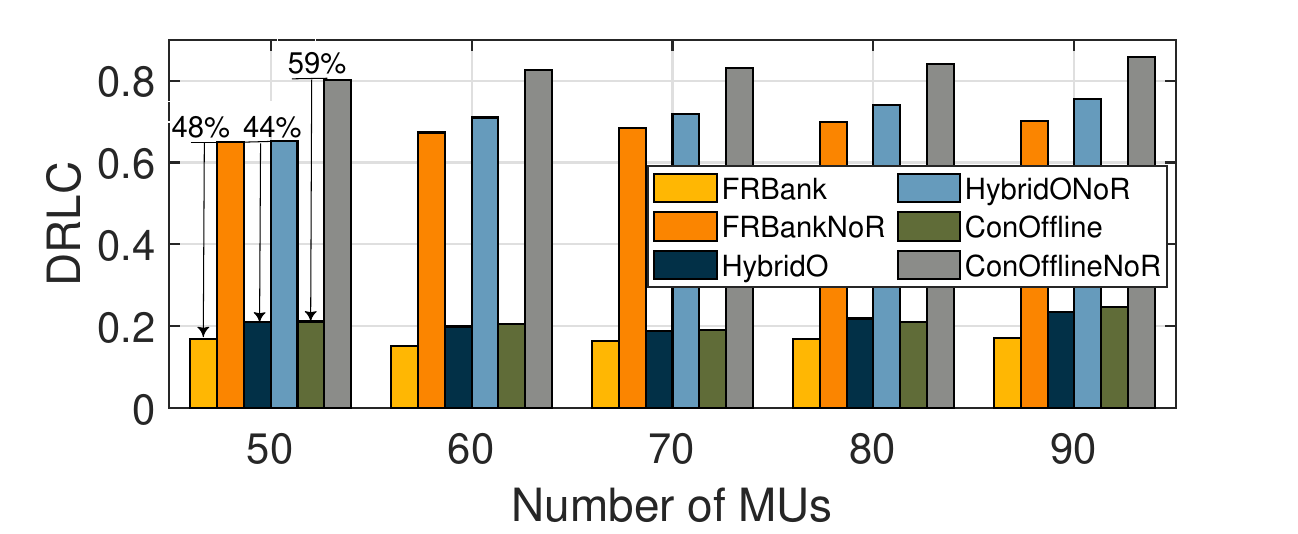} 
	} 
    \vspace{-0.5mm}
	\caption{Performance comparisons in terms of RDSLC and DRLC.} 
	\label{fig} 
\end{figure}
\noindent~\textit{(b) Analysis of overbooking (Fig. \ref{overbook}):} In our FRBank,  overbooking enables BSs to mitigate fluctuations in MU demand during practical transactions by strategically overbooking resources in long-term contracts.
To evaluate its impact, we analyze the \textbf{r}atio of BSs' resource \textbf{d}emand to resource \textbf{s}upply in practical transactions under \textbf{l}ong-term \textbf{c}ontracts (RDSLC) across different overbooking levels. Specifically, Fig. \ref{overbook} presents a comparative analysis between HybridO and FRBank for different resources, highlighting how overbooking optimizes resource allocation stability and adaptability in dynamic ISAC environments. In Fig. \ref{overbook}, as the overbooking rate increases, we observe a rise in RDSLC curves for both HybridO and FRBank in terms of bandwidth utilization. This trend highlights the effectiveness of the overbooking strategy in mitigating resource underutilization caused by the uncertainty of MU participation in transactions. Nevertheless, it is critical to carefully calibrate the overbooking rate, as excessive overbooking can lead to resource shortages. For instance, when the overbooking rate exceeds 0.3, the RDSLC curve for HybridO in terms of bandwidth surpasses 1, indicating that the BS should select voluntary clients to compensate because of the resource deficit. 
This underscores the importance of a proper overbooking rate, preventing both  resource underutilization and overcommitment in ISAC networks.

\noindent~\textit{(c) Analysis of risks (Fig. \ref{risk}):} To illustrate the importance of risk assessment and management, we introduce a metric called \textbf{d}efault \textbf{r}ate on \textbf{l}ong-term \textbf{c}ontract (DRLC)\footnote{The ratio of the number of failed transactions under long-term contracts to the total number of transactions specified in the contract during practical operations. This metric quantifies the contract violation rate.} upon having 5 BSs and different number of MUs. In Fig.~\ref{risk}, we conduct ablation experiments by introducing three comparative methods called FRBankNoR, HybridONoR, and ConOfflineNoR,  indicating FRBank, HybridO, and ConOffline methods excluding risk analysis. Note that if the utility of either the client or the BS becomes negative, or the compensation received is insufficient to offset the incurred loss, the affected party will opt to break contracts to minimize its loss and maximize its own utility. Additionally, due to overbooking and fluctuations in MUs' resource demand, BSs may encounter situations where resource supply is insufficient. In such cases, the BS may choose to breach long-term contracts and compensate volunteers to ensure sufficient resource availability. As shown in the figure, the value of DRLC of ConOfflineNoR, HybridONoR and FRBankNoR are significantly higher than those of ConOffline, HybridO and FRBank. This improvement is attributed to the implementation of effective risk management, which mitigates the risks of clients or BSs encountering unsatisfactory transactions and the risk of insufficient resource supply from the BS. These risk constraints allow our design to better adapt to the challenges posed by dynamic and uncertain ISAC network environments.

\subsection{Experiments on a Real-World Dataset}\label{sec:real}
\begin{figure}[t]
	\centering
	\includegraphics[width=1\columnwidth]{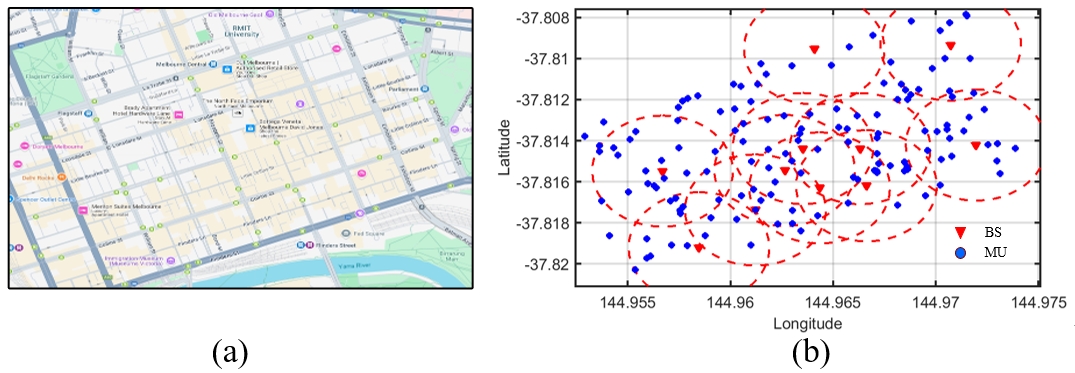}
        \vspace{-8mm}
	\caption{Locations of BSs and MUs in Melbourne  area as per EUA Dataset.}
	\label{EUA}
\end{figure}

To further evaluate the performance, we leverage the real-world EUA dataset\cite{EUA dataset}, focusing on Melbourne Central Business District  (see Fig. \ref{EUA}).
Within this region, we randomly selected 140 MUs, identified the geographic locations of 11 BSs, and defined 28 sensing targets. The results are summarized in Table \ref{table EUA}, which confirm that FRBank delivers the best balance across key metrics, including social welfare, RT, NI, DIBC, and ECIBC. Note that ConOffline operates exclusively in the offline mode and therefore does not incur real-time decision-making delays or overhead, reflected by N/A cells in this table.

  \begin{table}[h]

	{\small 
		\caption{Performance Evaluations for EUA Dataset \\(Alg. 1: FRBank, Alg. 2: ConOnline, Alg. 3: ConOffline, Alg. 4: Hybrid, Alg. 5: HybridO, Alg. 6: Greedy)}
        \vspace{-2mm}
		\begin{center}\vspace{-0.2cm}\label{table EUA}
			\setlength{\tabcolsep}{0.5mm}{
				\begin{tabular}{|c||c|c|c|c|c|c|}
					\hline \hline\rowcolor{gray!20}
					\textbf{Performance} & \textbf{Alg. 1}& \textbf{Alg. 2}&\textbf{Alg. 3}&\textbf{Alg. 4} & \textbf{Alg. 5} & \textbf{Alg. 6}\\ \hline\hline
					\textbf{Social welfare} &5014.70&3103.04&2851.26&3291.26&3074.83&1927.03
					\\ \hline
					\textbf{RT (ms)} &6705.7&389286.8&0.7 &27832.5&11712.7&1505.6\\ \hline
					\textbf{NI} &2878&37820&N/A&11803&8207&280
					\\ \hline
					\textbf{DIBC (ms)} &23291&302678&N/A&94909&64740 &2231	\\ \hline
					\textbf{ECIBC (W)} &304.53&3972.40&N/A&1251.95&868.92&29.89	\\ \hline \hline
			\end{tabular}}
	\end{center}}
         \vspace{-4mm}
\end{table}

\section{Conclusion and Future Work}
\noindent In this paper, we tackled the problem of resource allocation in multi-MU, multi-BS ISAC networks. We proposed a novel framework called the \textit{future resource bank for ISAC (FRBank)}, which models individual MUs and sensing coalitions as clients while integrating offline and online resource trading modes to enable efficient and mutually beneficial interactions. 
FRBank incorporates offRFW$^2$M with overbooking strategies to obtain long-term contracts between clients and BSs while mitigating potential risks. It also incorporates onEBW$^2$M to address challenges arising from intermittent MU participation.
Through simulations, we demonstrated the superior performance of FRBank in terms of time efficiency and network/user utility. 
{
Promising future directions are as follows.  
\textit{(i)} Adaptive overbooking strategies: Dynamically adjusting overbooking rates to better cope with real-time market fluctuations and reduce the cost of contract violations; 
\textit{(ii)} Sensing coalition value allocation: Investigating fair and stable value/cost distribution within sensing coalitions using cooperative game theory or bargaining models;} {
\textit{(iii)} Sensing utility modeling with RCS uncertainty: Considering RCS uncertainties and communication constraints to enhance the robustness of sensing utility estimation and decision-making;
\textit{(iv)} Incentive-compatible volunteer selection: Designing advanced compensation schemes (e.g., auction-based or contract-theoretic) to ensure fairness and efficiency in volunteer participation under limited information.}

{

	}
  
\newpage
\clearpage
\appendices
\section{Derivations associated with ISAC}
\subsection{Derivation of CRLB of ${\boldsymbol{\eta}_{i,j}}$}
Here, we first note that the device-free sensing model and the device-based sensing model are similar in structure, with the latter being derivable in a similar manner. As an example, we present the detailed derivation of the device-free sensing model. Denote
\begin{equation}
	\mathbf{H}_{i,j,n} = \Omega(\tau_{i,j,n}) \mathbf{A}_{i,j}^\mathsf{r} (\phi^\mathsf{(s)}_{i,j}) (\mathbf{A}_{i,j}^\mathsf{t})^\mathsf{H} \theta^\mathsf{(s)}_{i,j},
\end{equation}
where
\begin{equation}
	\Omega(\tau_{i,j}) = \frac{\sqrt{N_j^\mathsf{T} N_i^\mathsf{R} p_{i,j,n} h_{i,j,n}}}{\sqrt{p_{i,j,n}}} e^{j\frac{2\pi n \tau_{i,j,n}}{NT_s}}.
\end{equation}
Based on \cite{RW ISAC1,46}, when dealing with AWGN, the $(\mathbbm{p},\mathbbm{q})$-th elements in $\mathbf{J}_{\boldsymbol{\eta}_{i,j}}(\mathbbm{p},\mathbbm{q})$ can be derived as
\begin{equation}
	\mathbf{J}_{\boldsymbol{\eta}_{i,j}}(\mathbbm{p},\mathbbm{q}) = \sum_n \frac{2}{\sigma_s^2} \left\{ \frac{\partial (\mathbf{H}_{i,j,n} s_{ref})^\mathsf{H}}{\partial {\boldsymbol{\eta}_{i,j}}(\mathbbm{p})} \cdot \frac{\partial (\mathbf{H}_{i,j,n} s_{ref})}{\partial {\boldsymbol{\eta}_{i,j}}(\mathbbm{q})} \right\}^\mathsf{Re},
\end{equation}
where ${\boldsymbol{\eta}_{i,j}}(\mathbbm{q})$ is the $\mathbbm{q}$-th element of ${\boldsymbol{\eta}_{i,j}}$. By decomposing $\mathbf{H}_{i,j,n}$ across antennas, we have
\begin{equation}
	\begin{aligned}
		&\tilde{h}_{i,j,n,l_j^\mathsf{T},l_i^\mathsf{R}} =\\& \Omega(\tau_{i,j,n}) \times \exp\left( \frac{2\pi n}{\lambda_n} j \left( - \frac{N_j^\mathsf{T}-1}{2} + l_j^\mathsf{T} \right) \right)
		\times d \sin \theta_{i,j} \\&+ \frac{2\pi n}{\lambda_n} j \left( - \frac{N_i^\mathsf{R}-1}{2} + l_i^\mathsf{R} \right) d \sin \phi_{i,j},
	\end{aligned}
\end{equation}
where $l_j^\mathsf{T}\in \{1,2,...,N_j^\mathsf{T}\}$ and $l_i^\mathsf{R} \in \{1,2,...,N_i^\mathsf{R}\}$. Inspired by the method presented in \cite{46}, under AWGN, the $(\mathbbm{p},\mathbbm{q})$-th elements in $\mathbf{J}_{\boldsymbol{\eta}_{i,j}}$ can then be derived as
\begin{equation}{\small
	\begin{aligned}
			&\mathbf{J}_{\boldsymbol{\eta}_{i,j}}(\mathbbm{p},\mathbbm{q}) =\\& \frac{2}{\sigma_s^2} \sum_{n = \mathbbm{l}^\mathsf{(s)}_{i,j}}^{\mathbbm{l}^\mathsf{(s)}_{i,j}+N^\mathsf{(s)}_{i,j}} \sum_{l_j^\mathsf{T}=1}^{N_j^\mathsf{T}} \sum_{l_i^\mathsf{R}=1}^{N_i^\mathsf{R}} \left\{ \frac{\partial (\tilde{h}_{i,j,n,l_j^\mathsf{T},l_i^\mathsf{R}})^\mathsf{H}}{\partial {\boldsymbol{\eta}_{i,j}}(\mathbbm{p})} \cdot \frac{\partial \tilde{h}_{i,j,n,l_j^\mathsf{T},l_i^\mathsf{R}}}{\partial {\boldsymbol{\eta}_{i,j}}(\mathbbm{q})} \right\}^\mathsf{Re},
	\end{aligned}}
\end{equation}
where the first summation is among all the subcarriers in the subchannels. Due to the difficulty of deriving a closed-form with inverse processing of $\mathbf{J}_{\boldsymbol{\eta}_{i,j}}$, we adopt a similar approximation approach as also applied in \cite{RW ISAC1,47,48}. For the three variables in ${\boldsymbol{\eta}_{i,j}}$, assuming that the error source comes from the white noise, the estimation of each measurement is considered as independent, then, we have
\begin{equation}
	\text{var}(\bm{\hat{\eta}}_{i,j}(i)) \geq \mathbf{J}_{\boldsymbol{\eta}_{i,j}}^{-1}(i,i), \quad i \in \{1, 2, 3\},
\end{equation}
as shown in (\ref{equ. 59}).
\begin{figure*}[b!]
	\hrulefill
	\begin{equation}\label{equ. 59}
		\begin{aligned}
			\text{var}(\hat{\theta}_{i,j}) &\geq \frac{\zeta_1}{N^\mathsf{(s)}_{i,j} P_{i,j}} = \frac{3\sigma^2 \rho_{i,j}}{16\pi^2 h^2_{i,j}N^\mathsf{(s)}_{i,j} N_i^\mathsf{R} P^\mathsf{(s)}_{i,j} \cos^2 \theta_{i,j}N_j^\mathsf{T} (N_j^\mathsf{T} + 1)(2N_j^\mathsf{T} + 1)} \\
			\text{var}(\hat{\phi}_{i,j}) &\geq \frac{\zeta_2}{N^\mathsf{(s)}_{i,j} P_{i,j}} = \frac{3\sigma^2 \rho_{i,j}}{16\pi^2 h^2_{i,j}N^\mathsf{(s)}_{i,j} N_j^\mathsf{T} P^\mathsf{(s)}_{i,j} \cos^2 \theta_{i,j}N_i^\mathsf{R} (N_i^\mathsf{R} + 1)(2N_i^\mathsf{R} + 1)} \\
			\text{var}(\hat{\tau}_{i,j}) &\geq  \frac{\zeta_3}{N^\mathsf{(s)}_{i,j}(N^\mathsf{(s)}_{i,j}+1)(2N^\mathsf{(s)}_{i,j}+1)P^\mathsf{(s)}_{i,j}}, \zeta_3 = \frac{3\sigma^2 \rho_{i,j}(N_{i,j}^\mathsf{(s)})^2 T_s^2}{4\pi^2 h^2_{i,j} N_j^\mathsf{T} N_i^\mathsf{R}} 
		\end{aligned}
	\end{equation}
\end{figure*}

\subsection{Derivation of PEB from CRLB of ${\boldsymbol{\eta}_{i,j}}$}
With the FIM of channel parameters, the PEB$_{i,j}$ can be derived through the variable transformation tensor $\mathbb{T}$ from ${\boldsymbol{\eta}_{i,j}}$ to $l^\mathsf{Q}_{i}$\cite{RW ISAC1,49}, as given by
\begin{equation}\label{equ. 60}
	\mathbf{J}_{l^\mathsf{Q}_{i}} = \mathbb{T} \mathbf{J}_{\boldsymbol{\eta}_{i,j}} \mathbb{T}^\mathsf{T}, \quad \mathbb{T} = \frac{\partial {\boldsymbol{\eta}^\mathsf{T}_{i,j}}}{\partial l^\mathsf{Q}_{i}},
\end{equation}
where $\mathbb{T} \in \mathbb{R}^{3 \times 2}$ can be written as
\begin{equation}\label{equ. 61}
	\mathbb{T} = \left[ \frac{\partial \tau_{i,j}}{\partial l^\mathsf{Q}_{i}}, \frac{\partial \theta_{i,j}}{\partial l^\mathsf{Q}_{i}}, \frac{\partial \phi_k}{\partial l^\mathsf{Q}_{i}} \right].
\end{equation}

From (\ref{equ. transmission time})-(\ref{equ. 18}), we have
\begin{equation}\label{equ. 62}
	\begin{aligned}
			&\frac{\partial \tau_{i,j}}{\partial l^\mathsf{Q}_{i}} = \frac{ \left[ \cos(\theta_{i,j}) + \cos(\phi_{i,j}), \sin(\theta_{i,j}) + \sin(\phi_{i,j}) \right]^\mathsf{T}}{c},
		\\&\frac{\partial \theta_{i,j}}{\partial l^\mathsf{Q}_{i}} = \frac{\left[ - \sin(\theta_{i,j}), \cos(\theta_{i,j}) \right]^\mathsf{T}}{d^\mathsf{Q}_{i}} ,
		\\&\frac{\partial \phi_{i,j}}{\partial l^\mathsf{Q}_{i}} = \frac{-\left[ \cos(\pi - \phi_{i,j}), \cos(\pi - \phi_{i,j}) \right]^\mathsf{T}}{d^\mathsf{Q}_{i,j}} .
	\end{aligned}
\end{equation}
\begin{figure*}[t!]
	\begin{equation}\label{equ. 63}
		\begin{aligned}
			&\mathbf{J}_{l^\mathsf{Q}_{i,j}}(1,1) = P^\mathsf{(s)}_{i,j} N^\mathsf{(s)}_{i,j} \left( \frac{\sin^2 \theta_{i,j} }{\zeta_1 (d^\mathsf{Q}_{i,j})^2}+\frac{\cos^2 \theta_{i,j} }{\zeta_2 (d^\mathsf{Q}_{i})^2} + \frac{\left( \cos \theta_{i,j} + \cos \phi_{i,j} \right)^2}{\zeta_3 c^2} \left( N^\mathsf{(s)}_{i,j}+1 \right) \left( 2N^\mathsf{(s)}_{i,j}+1 \right)\right)
			\\&\mathbf{J}_{l^\mathsf{Q}_{i,j}}(2,2) = P^\mathsf{(s)}_{i,j} N^\mathsf{(s)}_{i,j} \left( \frac{\cos^2 \theta_{i,j} }{\zeta_1 (d^\mathsf{Q}_{i,j})^2}+\frac{\sin^2 \theta_{i,j} }{\zeta_2 (d^\mathsf{Q}_{i})^2} + \frac{\left( \sin \theta_{i,j} + \sin \phi_{i,j} \right)^2}{\zeta_3 c^2} \left( N^\mathsf{(s)}_{i,j}+1 \right) \left( 2N^\mathsf{(s)}_{i,j}+1 \right) \right)
		\end{aligned}
	\end{equation}\hrulefill
\end{figure*}
With (\ref{equ. 59}) and (\ref{equ. 62}), $\mathbf{J}_{l^\mathsf{Q}_{i}}$ can be expressed as (\ref{equ. 63}). Then, we can get
\begin{equation}
	\begin{aligned}
			&\text{PEB}_{i,j} = \sqrt{ \text{tr} \left\{ \mathbf{J}_{l^\mathsf{Q}_{i}}^{-1} \right\}} \\&\approx \sqrt{ \frac{1}{\mathbf{J}_{l^\mathsf{Q}_{i}}(1,1)} + \frac{1}{\mathbf{J}_{l^\mathsf{Q}_{i}}(2,2)} }
		= \frac{\zeta_{i,j}}{\sqrt{N^\mathsf{(s)}_{i,j}P^\mathsf{(s)}_{i,j}}} ,
	\end{aligned}
\end{equation}
where the expression for \( \zeta_{i,j} \), in terms of bandwidth, antennas, channel gain, noise power, and the position of sensing targets/MUs, can be derived by substituting (\ref{equ. 60}) into (\ref{equ. 61}).

\section{Derivations associated with offRFW$^2$M}
\noindent\textbf{Mathematical expectation of $ \alpha_i $.} $ \alpha_i $ refers to the MU participation uncertainty, following a Bernoulli random variable $\alpha_i$, $\alpha_i \sim {\bf{B}} \{(1, 0), (\mathbbm{a}_i, 1-\mathbbm{a}_i)\}$. The expectation of $\alpha_i$ can simply computed as $ \mathbb{E}[\alpha_i]=1\times\mathbbm{a}_i+0\times (1-\mathbbm{a}_i))=\mathbbm{a}_i $.

\noindent\textbf{Mathematical expectation of $ \beta_k $.} The variable \( \beta_k \) represents the uncertainty associated with the participation of \( \bm{c}_k \) in a practical transaction.
 The expectation of $\beta_k$ can simply computed as
\begin{equation}
	\begin{aligned}
	&\mathbb{E}[\beta_k]=1\times\Pr(\sum_{u_i\in \bm{c}_k}\alpha_i > 0)+0\times\Pr(\sum_{u_i\in \bm{c}_k}\alpha_i = 0)\\&=1-\Pr(\sum_{u_i\in \bm{c}_k}\alpha_i = 0)=1-\prod_{u_i\in\bm{c}_k}(1-\mathbbm{a}_i)
	\end{aligned}
\end{equation}

\noindent\textbf{Mathematical expectation of $ \mathbbm{v}^\mathsf{(c)}_{i,j} $ and $ \mathbbm{v}^\mathsf{(s)}_{k,j} $.}
As resource overbooking can lead to the case where a contractual client being selected as a volunteer, we use $ \mathbbm{v}^\mathsf{(c)}_{i,j} $ and $ \mathbbm{v}^\mathsf{(s)}_{k,j} $ to indicate whether MU $ u_i $ and $ \bm{c}_k$ is determined as a volunteer by BS $ s_j $ in each practical transaction.

Due to intermittent participation of MUs, assessing the close-form of the probability of an MU being a volunteer needs large calculations. We first use notation $\bm{\mathcal{M}}_j=\left\{\bm{M}_1,..,\bm{M}_n,...,\bm{M_{|\mathcal{M}_j|}}\right\}$ to collect all the possible cases of the participation of clients participation in the online trading market, where $\bm{M}_n=\left\{\alpha_1,...,\alpha_i,...,\alpha_{|\varphi^\mathsf{(c)}\left(s_j\right)|}, \beta_1,...,\beta_k, ..., \beta_{|\varphi^\mathsf{(s)}\left(s_j\right)|}\right\}$ is a vector and denoted as the ${n}^\text{th}$ case of clients' participation in a transaction. For example, suppose that $s_j$ has pre-signed long-term contracts with two clients, all the possible cases of these clients taking part in a practical transaction can be expressed as $\bm{\mathcal{M}_j}=\left\{\bm{M}_1,\bm{M}_2,\bm{M}_3,\bm{M}_4\right\}=\left\{\left\{0,0\right\},\left\{0,1\right\},\left\{1,0\right\},\left\{1,1\right\}\right\}$. For notational simplicity, we use a set $\bm{\mathcal{M}}^\mathsf{U}_{i,j}$ to denote all the cases that $u_i$ attends in a transaction, which, unfortunately, is selected as a volunteer by $s_j$.
Accordingly, the probability of MU $ u_i $ being determined by $s_j$ as a volunteer is given by
\begin{equation}\label{key}{\small
		\begin{aligned}
			\text{Pr}(\mathbbm{v}^\mathsf{(c)}_{i,j}=1)=\sum_{\bm{M}_n\in\bm{\mathcal{M}}^\mathsf{U}_{i,j}}{\prod_{\alpha_i,\beta_k\in \bm{M}_n}}{\mathbb{E}[\alpha_i]\mathbb{E}[\beta_k]},		
	\end{aligned}}
\end{equation}
and the expected value of $\mathbbm{v}^\mathsf{(c)}_{i,j}$ can be calculated by
\begin{equation}{\small
		\begin{aligned}
			&\mathbb{E}[\mathbbm{v}^\mathsf{(c)}_{i,j}]=\text{Pr}\left(\mathbbm{v}^\mathsf{(c)}_{i,j}=0\right)\times 0+\text{Pr}\left(\mathbbm{v}^\mathsf{(c)}_{i,j}=1\right)\times 1\\&=\text{Pr}(\mathbbm{v}^\mathsf{(c)}_{i,j}=1)=\sum_{\bm{M}_n\in\bm{\mathcal{M}}^\mathsf{U}_{i,j}}{\prod_{\alpha_i,\beta_k\in \bm{M}_n}}{{\mathbb{E}[\alpha_i]\mathbb{E}[\beta_k]}}.
	\end{aligned}}
\end{equation}
Similarly, we use the set \( \bm{\mathcal{M}}^\mathsf{C}_{k,j} \) to represent all the cases where \( \bm{c}_k \) participates in a transaction, which, however, has become a volunteer by \( s_j \). Therefore, the expected value of \( \mathbbm{v}^\mathsf{(s)}_{i,j} \) is given by
\begin{equation}{\small
		\begin{aligned}
			&\mathbb{E}[\mathbbm{v}^\mathsf{(s)}_{k,j}]=\text{Pr}\left(\mathbbm{v}^\mathsf{(s)}_{k,j}=0\right)\times 0+\text{Pr}\left(\mathbbm{v}^\mathsf{(s)}_{k,j}=1\right)\times 1\\&=\text{Pr}(\mathbbm{v}^\mathsf{(s)}_{k,j}=1)=\sum_{\bm{M}_n\in\bm{\mathcal{M}}^\mathsf{C}_{k,j}}{\prod_{\alpha_i,\beta_k\in \bm{M}_n}}{{\mathbb{E}[\alpha_i]\mathbb{E}[\beta_k]}}.
	\end{aligned}}
\end{equation}

\noindent\textbf{Mathematical expectation of $ V^\mathsf{(s),max }_{k,j} $.} Due to the uncertainty in the participation of MUs in the practical transaction, and \( V^{\mathsf{(s), max}}_{k,j} = \max \{V^{\mathsf{(s)}}_{i,j}\} , u_i \in \bm{c}_k \), the expected value of \( V^{\mathsf{(s), max}}_{k,j} \) can be defined as \( \mathbb{E}[V^{\mathsf{(s), max}}_{k,j}] = \max\{ \mathbbm{a}_i V^{\mathsf{(s)}}_{i,j} \}, u_i \in \bm{c}_k \).

\noindent\textbf{Derivations associated with (\ref{equ. PF BS C5}) and (\ref{equ. PF BS C6}).} We obtain an upper bound for the left-hand side of equation (\ref{equ. PF BS C5}) using the Markov inequality \cite{RW Matching3}, as given by
\begin{equation}\label{key}{\small
		\begin{aligned}
			&R_1^\mathsf{S}(s_j,\varphi^\mathsf{(c)}(s_j),\mathbb{C}^\mathsf{(c)}_{i,j},\varphi^\mathsf{(s)}(s_j),\mathbb{C}^\mathsf{(s)}_{k,j}) \leq \rho_1 \Rightarrow \\& \Pr\left(\sum_{u_i\in\varphi^\mathsf{(c)}(s_j)}\alpha_i\mathbbm{c}^\mathsf{(c),B}_{i,j}+\sum_{\bm{c}_k\in\varphi^\mathsf{(s)}(s_j)}\beta_k\mathbbm{c}^\mathsf{(s),B}_{k,j}> B_j\right)\leq\\& \frac{\mathbb{E}\left[\sum_{u_i\in\varphi^\mathsf{(c)}(s_j)}\alpha_i\mathbbm{c}^\mathsf{(c),B}_{i,j}+\sum_{\bm{c}_k\in\varphi^\mathsf{(s)}(s_j)}\beta_k\mathbbm{c}^\mathsf{(s),B}_{k,j}\right]}{B_j} \leq\rho_1 \Rightarrow\\&  
            \frac{\sum_{u_i\in\varphi^\mathsf{(c)}(s_j)}\mathbb{E}[\alpha_i]\mathbbm{c}^\mathsf{(c),B}_{i,j}+\sum_{\bm{c}_k\in\varphi^\mathsf{(s)}(s_j)}\mathbb{E}[\beta_k]\mathbbm{c}^\mathsf{(s),B}_{k,j}}{B_j} \leq\rho_1.
	\end{aligned}}
\end{equation}
Similarly, constraint (\ref{equ. PF BS C6}) can be reformulated as
\begin{equation}\label{key}{\small
		\begin{aligned}
			&R_2^\mathsf{S}(s_j,\varphi^\mathsf{(c)}(s_j),\mathbb{C}^\mathsf{(c)}_{i,j},\varphi^\mathsf{(s)}(s_j),\mathbb{C}^\mathsf{(s)}_{k,j}) \leq \rho_2 \Rightarrow \\& \Pr\left(\sum_{u_i\in\varphi^\mathsf{(c)}(s_j)}\alpha_i\mathbbm{c}^\mathsf{(c),Pow}_{i,j}+\sum_{\bm{c}_k\in\varphi^\mathsf{(s)}(s_j)}\beta_k\mathbbm{c}^\mathsf{(s),Pow}_{k,j}> P_j\right)\leq \rho_2 \\& \Rightarrow  
            \frac{\sum_{u_i\in\varphi^\mathsf{(c)}(s_j)}\mathbb{E}[\alpha_i]\mathbbm{c}^\mathsf{(c),Pow}_{i,j}+\sum_{\bm{c}_k\in\varphi^\mathsf{(s)}(s_j)}\mathbb{E}[\beta_k]\mathbbm{c}^\mathsf{(s),Pow}_{k,j}}{P_j} \leq\rho_2.
	\end{aligned}}
\end{equation}

\noindent\textbf{Derivations associated with (\ref{equ. PF MU C6}) and (\ref{equ. PF MU C7}).} 
Constraint (\ref{equ. PF MU C6}) represents a probabilistic expression, making its close form non-trivial to be obtained. Therefore, we define the variable $ \hat{u}^\mathsf{(c),U}\left(u_i,\varphi^\mathsf{(c)}(u_i),\mathbb{C}^\mathsf{(c)}_{i,j}\right)=u^\mathsf{(c)}_\mathsf{\max}-u^\mathsf{(c),U}\left(u_i,\varphi^\mathsf{(c)}(u_i),\mathbb{C}^\mathsf{(c)}_{i,j}\right) $, where $u^\mathsf{(c)}_\mathsf{\max}$ represents the maximum value of $ u^\mathsf{(c),U}\left(u_i,\varphi^\mathsf{(c)}(u_i),\mathbb{C}^\mathsf{(c)}_{i,j}\right)$, and transform (\ref{equ. PF MU C6}) into a tractable one by exploiting a set of bounding techniques. First, (\ref{equ. PF MU C6}) can be rewritten as
\begin{equation}\label{DR 38d}{\small
		\begin{aligned}
			&R_1^\mathsf{U}(u_i,\varphi^\mathsf{(c)}(u_i),\mathbb{C}^\mathsf{(c)}_{i,j})
            \\&=\Pr\left(u^\mathsf{(c),U}\left(u_i,\varphi^\mathsf{(c)}(u_i),\mathbb{C}^\mathsf{(c)}_{i,j}\right)\leq u^\mathsf{(c)}_\mathsf{\min}\right)
            \\&=\Pr\left(\hat{u}^\mathsf{(c),U}\left(u_i,\varphi^\mathsf{(c)}(u_i),\mathbb{C}^\mathsf{(c)}_{i,j}\right)\geq u^\mathsf{(c)}_\mathsf{\max}-u^\mathsf{(c)}_\mathsf{\min}\right)\leq\rho_3 .
	\end{aligned}}
\end{equation}

To obtain a tractable form for (\ref{DR 38d}), we can have the upper-bound of its left-hand side by using Markov inequality \cite{RW Matching3}:
\begin{equation}\label{DR 38d1}{\small
		\begin{aligned}
		&\Pr\left(\hat{u}^\mathsf{(c),U}\left(u_i,\varphi^\mathsf{(c)}(u_i),\mathbb{C}^\mathsf{(c)}_{i,j}\right)\geq u^\mathsf{(c)}_\mathsf{\max}-u^\mathsf{(c)}_\mathsf{\min}\right)
        \\&\leq \frac{\mathbb{E}\left[\hat{u}^\mathsf{(c),U}\left(u_i,\varphi^\mathsf{(c)}(u_i),\mathbb{C}^\mathsf{(c)}_{i,j}\right)\right]}{u^\mathsf{(c)}_\mathsf{\max}-u^\mathsf{(c)}_\mathsf{\min}}
        \\&=\frac{u^\mathsf{(c)}_\mathsf{\max}-\mathbb{E}\left[u^\mathsf{(c),U}\left(u_i,\varphi^\mathsf{(c)}(u_i),\mathbb{C}^\mathsf{(c)}_{i,j}\right)\right]}{u^\mathsf{(c)}_\mathsf{\max}-u^\mathsf{(c)}_\mathsf{\min}}.
	\end{aligned}}
\end{equation}
Combining (\ref{DR 38d}) and (\ref{DR 38d1}), we can then get a tractable form for (\ref{equ. PF MU C6}):
\begin{equation}\label{key}{\small
		\begin{aligned}
			\frac{u^\mathsf{(c)}_\mathsf{\max}-\mathbb{E}\left[u^\mathsf{(c),U}\left(u_i,\varphi^\mathsf{(c)}(u_i),\mathbb{C}^\mathsf{(c)}_{i,j}\right)\right]}{u^\mathsf{(c)}_\mathsf{\max}-u^\mathsf{(c)}_\mathsf{\min}}\leq\rho_3	,
	\end{aligned}}
\end{equation}
where the value of $\mathbb{E}\left [u^\mathsf{(c),U}\left(u_i,\varphi^\mathsf{(c)}(u_i),\mathbb{C}^\mathsf{(c)}_{i,j}\right)\right ]$ is given by (\ref{equ. expected comm utility}).
Similarly, we can get a tractable form for (\ref{equ. PF MU C7}):
\begin{equation}\label{key}{\small
		\begin{aligned}
			\frac{u^\mathsf{(s)}_\mathsf{\max}-\mathbb{E}\left [u^\mathsf{(s),U}\left(\bm{c}_k,\varphi^\mathsf{(s)}(\bm{c}_k),\mathbb{C}^\mathsf{(s)}_{k,j}\right)\right ]}{u^\mathsf{(s)}_\mathsf{\max}-u^\mathsf{(s)}_\mathsf{\min}}\leq\rho_4	,
	\end{aligned}}
\end{equation}
where $u^\mathsf{(s)}_\mathsf{\max}$ represents the maximum value of $ u^\mathsf{(s),U}\left(\bm{c}_k,\varphi^\mathsf{(s)}(\bm{c}_k),\mathbb{C}^\mathsf{(s)}_{k,j}\right)$, and the value of $\mathbb{E}\left [u^\mathsf{(s),U}\left(\bm{c}_k,\varphi^\mathsf{(s)}(\bm{c}_k),\mathbb{C}^\mathsf{(s)}_{k,j}\right)\right ]$ is given by (\ref{equ. expected sensing utility}).

\section{Property Analysis of offRFW$^2$M}
\begin{Prop}\label{Prop 7}
	(Convergence of a set of matching in offRFW$^2$M) Alg. 1 converges within finite rounds.
\end{Prop}
\begin{proof}
	As the offRFW$^2$M refers to a set of M2M matching (matching between BSs and individual MUs, as well as matching between BSs and coalitions), we utilize the DP algorithm to transform the problem into a two-dimensional 0-1 knapsack problem \cite{RW Matching3}. After a finite number of rounds, each client's payment can either be accepted or reach its maximum value while considering constraints (\ref{equ. PF BS C3}), (\ref{equ. PF BS C4}), (\ref{equ. PF BS C5}), and (\ref{equ. PF BS C6}) (e.g., lines 17-26, Alg. 1), ensuring the convergence.
\end{proof}

\begin{Prop}
	(Individual rationality of offRFW$^2$M) The proposed offRFW$^2$M mechanism ensures individual rationality for All the BSs, individual MUs, and sensing coalitions are individual rational in the offRFW$^2$M.
\end{Prop}
\begin{proof}
	We offer the analysis on proving the individual rationality of both BSs and clients.
	
		\textbf{Individual rationality of BSs.} Owing to overbooking, each BS $s_j$ regards $(1+O_j^\mathsf{B})B_j$ and $(1+O_j^\mathsf{Pow})P_j$ as up limit of resources for serving MUs, and the actual number of matched clients of $s_j$ will definitely not exceed its overbooked resource supply (e.g., line 15, Alg. 1). In addition, thanks to the risk analysis, the risk of BS \(s_{j}\) having actual resource demand exceeding supply is controlled within a reasonable range (e.g., ensuring that (\ref{equ. PF BS C5}) and (\ref{equ. PF BS C6}) are satisfied, see line 15, Alg. 1).
	
	\textbf{Individual rationality of clients.} Lines 17-26 of Alg. 1 ensure that the value obtained by each client is at least equal to the payment it makes, thereby satisfying constraint (\ref{equ. PF MU C3}). Furthermore, lines 6, 18 and 23 of Alg. 1 guarantee that risks associated with each client are controlled within acceptable limits, satisfying constraints (\ref{equ. PF MU C4}), (\ref{equ. PF MU C5}), (\ref{equ. PF MU C6}), and (\ref{equ. PF MU C7}).
	 
	As a summary, clients and BSs are individual rationality in our proposed offRFW$^2$M.
\end{proof}

\begin{Prop}
	No blocking pair can exist in the Resource Trading for Communication Services in offRFW$^2$M.
\end{Prop}
\begin{proof}
	We show there is no blocking pair of either Type 1 or Type 2, as following:
	
	\noindent 
	$\bullet$ \textbf{There is no Type 1 blocking pair related to communication services of offRFW$^2$M.} We offer the proof by considering contradiction.
	
	Under a given matching $ \varphi^\mathsf{(c)} $, MU $ u_i $ and BS $ s_j $ form a Type 1 blocking pair $ \left(u_i;s_j;\mathbb{C}^\prime\right) $.
	If MU $ u_i $ does not sign a long-term contract with BS $ s_j $, when any of the following conditions is met: \textit{(i)} the final payment offered by MU $ u_i $ equals to its expected valuation; and \textit{(ii)} the risk is out of control (e.g., constraint (\ref{equ. PF MU C6})). For analytical simplicity, we use $ p_{i,j}^\mathsf{(c),max} $ to denoted the maximum payment from $u_i$ to $s_j$ under an accepted risk $ R_1^\mathsf{U} $. Thus, the final payment $ \mathbbm{c}^\mathsf{(c),Pay}_{i,j}$ can only refer to $\mathbb{E}[V^\mathsf{(c)}_{i,j}]$ or $p_{i,j}^\mathsf{(c),max} $, shown by (\ref{59}) and (\ref{60}).
	\begin{equation}\label{59}{\small
			\begin{aligned}
				\mathbbm{c}^\mathsf{(c),Pay}_{i,j} = \text{min}\left\{\mathbb{E}[V^\mathsf{(c)}_{i,j}],p_{i,j}^\mathsf{(c),max}\right\},
		\end{aligned}}
	\end{equation}
		\begin{equation}\label{60}
		\begin{aligned}
			&\mathbb{E}\left[u^\mathsf{(c),S}\left(s_j,\left\{\varphi\left( s_j \right)\backslash\widetilde{\varphi^\mathsf{(c)\prime}}\left( s_j \right)\right\} \cup \left\{ u_i \right\},\mathbb{C}^\prime \right)\right ] \\&< \mathbb{E}\left [u^\mathsf{(c),S}\left(s_j,\varphi\left( s_j\right),\mathbb{C}^\mathsf{(c)}_{i,j} \right)\right].\\
		\end{aligned}
	\end{equation} 
	
	If BS $ s_j $ selects MU $ u_i $, we have $ \mathbbm{c}^\mathsf{(c),Pay}_{i,j}\left\langle \mathcal{X}^\mathsf{*} \right\rangle\leq \mathbbm{c}^\mathsf{(c),Pay}_{i,j}\left\langle \mathcal{X} \right\rangle = \text{min}\left\{\mathbb{E}[V^\mathsf{(c)}_{i,j}],p_{i,j}^\mathsf{(c),max}\right\} $ and the following (\ref{81})
	\begin{equation}\label{81}{\small
			\begin{aligned}
				&\mathbb{E}\left[u^\mathsf{(c),S}\left(s_j,\left\{\varphi\left( s_j \right)\backslash\widetilde{\varphi^\mathsf{(c)\prime}}\left( s_j \right)\right\} \cup \left\{ u_i \right\},\mathbb{C}^\prime \right)\right ] \geq\\& \mathbb{E}\left[u^\mathsf{(c),S}\left(s_j,\left\{\varphi\left( s_j \right)\backslash\widetilde{\varphi^\mathsf{(c)\prime\prime}}\left( s_j \right)\right\} \cup \left\{ u_i \right\},\mathbb{C}^\prime \right)\right ],\\
		\end{aligned}}
	\end{equation}
	where $ 
	\widetilde{\varphi^\mathsf{(c)\prime\prime}}\left(s_j\right) \subseteq \widetilde{\varphi^\mathsf{(c)\prime}}\left(s_j\right) $. From (\ref{60}) and (\ref{81}), we can get
	\begin{equation}\label{key}\small{
			\begin{aligned}
				&\mathbb{E}\left [u^\mathsf{(c),S}\left(s_j,\varphi\left( s_j\right),\mathbb{C}^\mathsf{(c)}_{i,j} \right)\right]> \\&\mathbb{E}\left[u^\mathsf{(c),S}\left(s_j,\left\{\varphi\left( s_j \right)\backslash\widetilde{\varphi^\mathsf{(c)\prime\prime}}\left( s_j \right)\right\} \cup \left\{ u_i \right\},\mathbb{C}^\prime \right)\right ],
		\end{aligned}}
	\end{equation}
	which is contrary to (\ref{equ. 42}), thus ensuring the inexistence of Type 1 blocking pairs.
	
	\noindent 
	$\bullet$ \textbf{There is no Type 2 blocking pair related to communication services of offRFW$^2$M.}
	We conduct the proof by considering cases of contradiction. 
	
	Under a given matching $ \varphi^\mathsf{(s)} $, MU $ u_i $ and BS $ s_j $ form a Type 2 blocking pair $ \left(u_i;s_j;\mathbb{C}^\prime\right) $, as shown by (\ref{equ. 44}).
	If MU $ u_i $ is rejected by BS $ s_j $, the final payment of $ u_i $ can be set by $ \mathbbm{c}^\mathsf{(c),Pay}_{i,j} = \text{min}\left\{\mathbb{E}[V^\mathsf{(c)}_{i,j}],p_{i,j}^\mathsf{(c),max}\right\} $, where the only reason of such a rejection is that $ s_j $ has no surplus resources. However, the coexistence of (\ref{equ. 44}) shows that BS $ s_j $ has adequate resource supply to serve MUs, which contradicts our previous assumption. Therefore, we prove that there is no Type 2 blocking pair.
	
	As a summary, no blocking pair can exist during the matching related to communication services in offRFW$^2$M. 
\end{proof}

\begin{Prop}\label{Prop 10}
	No blocking pair can exist in the Resource Trading for Sensing Services in offRFW$^2$M
\end{Prop}
\begin{proof}
	We show there is no blocking pair of either Type 1 or Type 2, as following:
	
	\noindent 
	$\bullet$ \textbf{There is no Type 1 blocking pair related to sensing services of offRFW$^2$M.} We offer the proof by considering contradiction.
	
	Under a given matching $ \varphi^\mathsf{(s)} $, coalition $ \bm{c}_k $ and BS $ s_j $ form a Type 1 blocking pair $ \left(\bm{c}_k;s_j;\mathbb{C}^\prime\right) $.
	If coalition $ \bm{c}_k $ does not sign a long-term contract with BS $ s_j $, when any of the following conditions is met: \textit{(i)} the final payment offered by coalition $ \bm{c}_k $ equals to its expected valuation; and \textit{(ii)} the risk is out of control (e.g., constraint (\ref{equ. PF MU C7})). For analytical simplicity, we use $ p_{k,j}^\mathsf{(s),max} $ to denoted the maximum payment from $\bm{c}_k$ to $s_j$ under an accepted risk $ R_2^\mathsf{U} $. Thus, the final payment $ \mathbbm{c}^\mathsf{(s),Pay}_{k,j}$ can only refer to $\mathbb{E}[V^\mathsf{(s)}_{k,j}]$ or $p_{k,j}^\mathsf{(s),max} $, shown by (\ref{59a}) and (\ref{60a}).
	\begin{equation}\label{59a}{\small
			\begin{aligned}
				\mathbbm{c}^\mathsf{(s),Pay}_{k,j} = \text{min}\left\{\mathbb{E}[V^\mathsf{(s)}_{k,j}],p_{k,j}^\mathsf{(s),max}\right\},
		\end{aligned}}
	\end{equation}
	\begin{equation}\label{60a}
		\begin{aligned}
			&\mathbb{E}\left[u^\mathsf{(s),S}\left(s_j,\left\{\varphi\left( s_j \right)\backslash\widetilde{\varphi^\mathsf{(s)\prime}}\left( s_j \right)\right\} \cup \left\{ \bm{c}_k \right\},\mathbb{C}^\prime \right)\right ] \\&< \mathbb{E}\left [u^\mathsf{(s),S}\left(s_j,\varphi\left( s_j\right),\mathbb{C}^\mathsf{(s)}_{k,j} \right)\right].\\
		\end{aligned}
	\end{equation} 
	
	If BS $ s_j $ selects coalition $ \bm{c}_k $, we have $ \mathbbm{c}^\mathsf{(s),Pay}_{k,j}\left\langle \mathcal{X}^\mathsf{*} \right\rangle\leq \mathbbm{c}^\mathsf{(s),Pay}_{k,j}\left\langle \mathcal{X} \right\rangle = \text{min}\left\{\mathbb{E}[V^\mathsf{(s)}_{k,j}],p_{k,j}^\mathsf{(s),max}\right\} $ and the following (\ref{81a})
	\begin{equation}\label{81a}{\small
			\begin{aligned}
				&\mathbb{E}\left[u^\mathsf{(s),S}\left(s_j,\left\{\varphi\left( s_j \right)\backslash\widetilde{\varphi^\mathsf{(s)\prime}}\left( s_j \right)\right\} \cup \left\{ \bm{c}_k \right\},\mathbb{C}^\prime \right)\right ] \geq\\& \mathbb{E}\left[u^\mathsf{(s),S}\left(s_j,\left\{\varphi\left( s_j \right)\backslash\widetilde{\varphi^\mathsf{(s)\prime\prime}}\left( s_j \right)\right\} \cup \left\{ \bm{c}_k \right\},\mathbb{C}^\prime \right)\right ],\\
		\end{aligned}}
	\end{equation}
	where $ 
	\widetilde{\varphi^\mathsf{(s)\prime\prime}}\left(s_j\right) \subseteq \widetilde{\varphi^\mathsf{(s)\prime}}\left(s_j\right) $. From (\ref{60a}) and (\ref{81a}), we can get
	\begin{equation}\label{key}\small{
			\begin{aligned}
				&\mathbb{E}\left [u^\mathsf{(s),S}\left(s_j,\varphi\left( s_j\right),\mathbb{C}^\mathsf{(s)}_{k,j} \right)\right]> \\&\mathbb{E}\left[u^\mathsf{(s),S}\left(s_j,\left\{\varphi\left( s_j \right)\backslash\widetilde{\varphi^\mathsf{(s)\prime\prime}}\left( s_j \right)\right\} \cup \left\{ \bm{c}_k \right\},\mathbb{C}^\prime \right)\right ],
		\end{aligned}}
	\end{equation}
	which is contrary to (\ref{equ. 42}), and thus proving the inexistence of Type 1 blocking pairs.
	
	\noindent 
	$\bullet$ \textbf{There is no Type 2 blocking pair related to sensing services of offRFW$^2$M.}
	We conduct the proof by considering cases of contradiction. 
	
	Under a given matching $ \varphi^\mathsf{(s)} $, coalition $ \bm{c}_k $ and BS $ s_j $ form a Type 2 blocking pair $ \left(\bm{c}_k;s_j;\mathbb{C}^\prime\right) $, as shown by (\ref{equ. 44}).
	If coalition $ \bm{c}_k $ is rejected by BS $ s_j $, the final payment of $ \bm{c}_k $ can be set by $ \mathbbm{c}^\mathsf{(s),Pay}_{k,j} = \text{min}\left\{\mathbb{E}[V^\mathsf{(s)}_{k,j}],p_{k,j}^\mathsf{(s),max}\right\} $, where the only reason of such a rejection is that $ s_j $ has no surplus resources. However, the coexistence of (\ref{equ. 44}) shows that BS $ s_j $ has adequate resource supply to serve coalitions, which contradicts our previous assumption. Therefore, we prove that there is no Type 2 blocking pair.
	
	As a summary, no blocking pair can exist during the matching related to sensing services in offRFW$^2$M. 
\end{proof}

\begin{Prop}\label{Prop 11}
	(Fairness, Non-wastefulness, Strong Stability of offRFW$^2$M) offRFW$^2$M is fair, non-wasteful, strong stable.
\end{Prop}
\begin{proof}
	Since the matching result of Alg. 1 holds Propositions \ref{Prop 7}-\ref{Prop 10}, according to Propositions \ref{Prop 1}-\ref{Prop 4}, our proposed offRFW$^2$M is strongly fairness, non-wastefulness, strong stability.
\end{proof}

\begin{Prop}(Stability of Sensing Coalitions in offRFW$^2$M) The proposed offRFW$^2$M ensures that each sensing coalition $\bm{c}_k$ is stable.\end{Prop}
\begin{proof}
Due to line 24 in Alg. 1, each MU \( u_i \) in the sensing coalition \( \bm{c}_k \) ensures its expected utility to exceed \( u^\mathsf{(s)}_\mathsf{\min} \). Furthermore, MUs within a coalition share both costs and profits, leading to a lower expected utility per MU compared to trading individually. Also, even in the extreme case where only one MU from the coalition engages in a transaction, the incurred cost remains equal to the expected utility when trading independently. Therefore, we can conclude that joining coalition $\bm{c}_k$ will not result in a lower expected utility than trading as an individual.
\end{proof}

\begin{Prop}
	(Weak Pareto optimality of offRFW$^2$M) The proposed offRFW$^2$M provides a weak Pareto optimality.
\end{Prop}
\begin{proof}
	Reviewing our design of offRFW$^2$M, each participant (e.g., client, BS) makes decisions according to its preference list to determine the trading counterpart and the specific terms of the long-term contract (e.g., resource trading volume, transaction price, compensation price). If the alternative choice ranks higher in the participant's preference list, they will switch their matching target and long-term contract in the following round. Such a switch indicates that returning to the previous choice would not result in a higher expected utility. For an MU \( u_i \), if there exists a BS \( s_j \) that can offer a higher expected utility than its currently matched BS, \( u_i \) and \( s_j \) are more inclined to establish a matching relationship. This, however, forms a blocking pair. Since Proposition \ref{Prop 11} confirms that our proposed offRFW$^2$M is stable and free of blocking pairs, there is no possibility of Pareto improvement when the procedure of matching \( \varphi^\mathsf{(c)} \) terminates. Similarly, we can infer that there is no Pareto improvement in matching \( \varphi^\mathsf{(s)} \) (e.g., Propositions \ref{Prop 10} and \ref{Prop 11}). In conclusion, the offRFW$^2$M we study is said to be weak Pareto optimal.
\end{proof}

\section{Details of onEBW$^2$M}
\subsection{Key Definitions}
\begin{Defn}(M2M Matching for Communication Services in onEBW$^2$M)
	An M2M matching \( \nu^\mathsf{(c)} \) designed for communication services in onEBW$^2$M constitutes a two-way function/mapping between the BS set \( \bm{\mathcal{S}^\prime} \) and the MU set \(\bm{\mathcal{U}^\prime} \), satisfying the following properties:
	
	\noindent
	$\bullet$ For each BS $ s_{j} \in \bm{\mathcal{S}^\prime},\nu^\mathsf{(c)}\left( s_j \right) \subseteq \bm{\mathcal{U}^\prime} $, meaning that a BS can provide communication services to multiple MUs simultaneously based on its available resources.
	
	\noindent
	$\bullet$ For each MU $ u_i \in \bm{\mathcal{U}^\prime}, \nu^\mathsf{(c)}\left( u_i \right) \subseteq \bm{\mathcal{S}^\prime} $, and $|\nu^\mathsf{(c)}\left( u_i \right)|=1$, ensuring that each MU is assigned to exactly one BS, maintaining a structured association for stable temporary contracts.
	
	\noindent
	$\bullet$ For each BS $ s_j $ and MU $ u_i $, $ s_j\in\nu^\mathsf{(c)}(u_i)$ if and only if $ u_i\in\nu^\mathsf{(c)}\left(s_j\right) $, indicating that a valid matching occurs only when both the MU and the BS mutually accept the contract, ensuring reciprocal agreement.
\end{Defn}
\begin{Defn}(M2M Matching for Sensing Services in onEBW$^2$M)
	An M2M matching \( \nu^\mathsf{(s)} \) designed for sensing services in onEBW$^2$M constitutes a two-way function/mapping between the BS set \( \bm{\mathcal{S}^\prime} \) and the sensing coalition set \( \bm{\mathcal{C}^\prime} \), satisfying the following properties:
	
	\noindent
	$\bullet$ For each BS $ s_{j} \in \bm{\mathcal{S}^\prime},\nu^\mathsf{(s)}\left( s_j \right) \subseteq \bm{\mathcal{C}^\prime} $, meaning that a BS can provide sensing resources to multiple sensing coalitions simultaneously based on its available bandwidth and power.
	
	\noindent
	$\bullet$ For each MUs' coalition $ \bm{c}_k \in \bm{\mathcal{C}^\prime}, \nu^\mathsf{(s)}\left( \bm{c}_k \right) \subseteq \bm{\mathcal{S}^\prime} $, and $|\nu^\mathsf{(s)}\left( \bm{c}_k \right)|=1$, ensuring that each coalition is assigned to one or more BSs, allowing cooperative resource provisioning for enhanced sensing accuracy.
	
	\noindent
	$\bullet$ For each BS $ s_j $ and sensing coalition $ \bm{c}_k $, $ s_j\in\nu^\mathsf{(s)}(\bm{c}_k)$ if and only if $ \bm{c}_k\in\nu^\mathsf{(s)}\left(s_j\right) $, ensuring that a valid matching occurs only when both the BS and the coalition mutually accept the resource-sharing agreement, fostering stable and efficient sensing operations.
\end{Defn}
\subsection{Problem Formulation}
We formulate the bandwidth and power resource trading in the designed online mode as obtaining M2M matching between level-wise clients with unmet demands (individual MUs and sensing coalitions) and BSs with surplus supply, while simultaneously determining their temporary contracts. Similar to offline trading mode, the objective of each BS \( s_j \in \bm{\mathcal{S^\prime}} \) is to maximize its overall practical utility, as formulated by
\begin{subequations}{\small
		\begin{align}	\hspace{-3mm}\bm{\mathcal{F}^\mathsf{S^\prime}}\hspace{-1mm}{:}&\hspace{-2mm}\underset{{\dot{\mathbb{C}}^\mathsf{(c)}_{i,j},\dot{\mathbb{C}}^\mathsf{(s)}_{k,j}}}{\max}\hspace{-1mm}u^\mathsf{(c),S}(s_j,\nu^\mathsf{(c)}(s_j),\dot{\mathbb{C}}^\mathsf{(c)}_{i,j})\hspace{-1mm}+ \hspace{-1mm}u^\mathsf{(s),S}(s_j,\nu^\mathsf{(s)}(s_j),\dot{\mathbb{C}}^\mathsf{(s)}_{k,j}) \label{equ. PF BSa} \tag{78}\\
			\text{s.t.}~~~
			&\nu^\mathsf{(c)}\left(s_j\right)\subseteq\bm{\mathcal{U}^\prime},\nu^\mathsf{(s)}\left(s_j\right)\subseteq\bm{\mathcal{C}^\prime}, \mu^\prime\left(\bm{c}_k\right)\subseteq\bm{\mathcal{U}^\prime}, \tag{78a}\label{equ. PF BS C1a}\\
			&u_i\in \nu^\mathsf{(c)}(s_j), \bm{c}_k\in \nu^\mathsf{(s)}(s_j), u_i\in\mu^\prime(\bm{c}_k), \tag{78b}\label{equ. PF BS C2a}\\
			&\sum_{u_i\in\nu^\mathsf{(c)}(s_j)}B_{i,j}^\mathsf{(c)}+\sum_{\bm{c}_k\in\nu^\mathsf{(s)}(s_j)}\dot{\mathbbm{c}}^\mathsf{(s),B}_{k,j}\leq B^\prime_j, \tag{78c}\label{equ. PF BS C3a}\\
			&\sum_{u_i\in\nu^\mathsf{(c)}(s_j)}P_{i,j}^\mathsf{(c)}+\sum_{\bm{c}_k\in\nu^\mathsf{(s)}(s_j)}\dot{\mathbbm{c}}^\mathsf{(s),Pow}_{k,j}\leq P^\prime_j, \tag{78d}\label{equ. PF BS C4a}
	\end{align}}
\end{subequations}
In $ \bm{\mathcal{F}^\mathsf{S^\prime}} $, constraint (\ref{equ. PF BS C1a}) and (\ref{equ. PF BS C2a}) enforce that the set of MUs $\nu^\mathsf{(c)}(s_j)$ should belong to set $\bm{\mathcal{U}^\prime}$, the coalition set $\nu^\mathsf{(s)}(s_j)$ for sensing services must be covered by $\bm{\mathcal{C}^\prime}$, and the MU set $\mu^\prime(\bm{c}_k)$ of coalition $\bm{c}_k$ has to be within $\bm{\mathcal{U}^\prime}$. Constraints (\ref{equ. PF BS C3a}) and (\ref{equ. PF BS C4a}) ensure that the bandwidth and power resources sold by BS $s_j$ do not exceed its available supply $B_j^\prime$ and $P_j^\prime$. 
Furthermore, each client (i.e., $u_i $ or $\bm{c}_k$) also aims \textit{to maximize its utility}, as described by the following optimization problem
\begin{subequations}
		\begin{align}
			\bm{\mathcal{F}^\mathsf{U^\prime}}:~&
			\left\{ \begin{matrix}
				\underset{{\dot{\mathbb{C}}^\mathsf{(c)}_{i,j}}}{\max}~u^\mathsf{(c),U}(u_i,\nu^\mathsf{(s)}(u_i),\dot{\mathbb{C}}^\mathsf{(c)}_{i,j}) \\
				\underset{{\dot{\mathbb{C}}^\mathsf{(s)}_{k,j}}}{\max}~u^\mathsf{(s),U}(\bm{c}_k,\nu^\mathsf{(s)}(\bm{c}_k),\dot{\mathbb{C}}^\mathsf{(s)}_{k,j})
			\end{matrix}\right\}, \tag{79}\label{equ. PF MUa}\\
			\text{s.t.}~~~
			&\nu^\mathsf{(c)}\left(u_i\right)\subseteq\bm{\mathcal{S}^\prime}, \mu^\prime\left(u_i\right)\subseteq\bm{\mathcal{C}^\prime},\nu^\mathsf{(s)}\left(\bm{c}_k\right)\subseteq\bm{\mathcal{S}^\prime} ,\label{equ. PF MU C1a}\tag{79a}\\
			&s_j\in \nu^\mathsf{(c)}(u_i), \bm{c}_k\in \mu^\prime(u_i), s_j\in\nu^\mathsf{(s)}(\bm{c}_k), \tag{79b}\label{equ. PF MU C2a}\\
			&V^\mathsf{(c)}_{i,j}\ge \dot{\mathbbm{c}}^\mathsf{(c),Pay}_{i,j},V^\mathsf{(s)}_{i,j}=V^\mathsf{(s),max }_{k,j}\ge \dot{\mathbbm{c}}^\mathsf{(s),Pay}_{i,j}, \tag{79c}\label{equ. PF MU C3a}\\
			&V^\mathsf{(c)}_{i,j}\geq R^\mathsf{req},\tag{79d}\label{equ. PF MU C4a}\\
			&V^\mathsf{(s),max }_{k,j}\geq S^\mathsf{req},\tag{79e}\label{equ. PF MU C5a}\\
                        & B_{\min} \leq B^\mathsf{(c)}_{i,j}, B^\mathsf{(s)}_{k,j}\leq B_{\max},\tag{79f}\label{equ. PF MU C8a}\\
		&P_{\min} \leq P^\mathsf{(c)}_{i,j}, P^\mathsf{(s)}_{k,j}\leq P_{\max},\tag{79g}\label{equ. PF MU C9a}
	\end{align}
\end{subequations}
In $ \bm{\mathcal{F}^\mathsf{U^\prime}} $, constraints (\ref{equ. PF MU C1a}) and (\ref{equ. PF MU C2a}) are similar to constraints (\ref{equ. PF BS C1a}) and (\ref{equ. PF BS C2a}). Constraint (\ref{equ. PF MU C3a}) ensures that the obtained valuation of $u_i$ benefit from $s_j$ or $\bm{c}_k$ can cover its individual payment, while constraints (\ref{equ. PF MU C4a}) and (\ref{equ. PF MU C5a}) guarantee that the communication and sensing service quality of each MU or sensing coalition meets the corresponding requirements. Constraints (\ref{equ. PF MU C8a}) and (\ref{equ. PF MU C9a}) guarantee the bandwidth and power resources requested by each client for services are constrained within a certain range.

The online trading mode thus presents an MOO problem involving both $\bm{\mathcal{F}^\mathsf{S^\prime}}$ and $\bm{\mathcal{F}^\mathsf{U^\prime}}$, where the conflicting utilities of different parties make designing a win-win solution for them a complex task. To address this MOO problem, we propose onEBW$^2$M, which facilitates temporary contracts while achieving mutually -- beneficial practical utilities for both parties. The following sections discuss the detailed implementation of onEBW$^2$M.

\subsection{Solution Design}
\begin{algorithm}[t!] 
	{\footnotesize\setstretch{0.4}\caption{{Proposed Effective Backup Win-Win Matching for Online Trading}}
		\LinesNumbered 
		\textbf{Initialization:} $ \mathcal{X} \leftarrow 1 $, $ \dot{\mathbbm{c}}^\mathsf{(c),Pay}_{i,j}\left\langle 1 \right\rangle \leftarrow p^\mathsf{\min}_{i,j}$, $ \dot{\mathbbm{c}}^\mathsf{(s),Pay}_{k,j}\left\langle 1 \right\rangle \leftarrow p^\mathsf{\min}_{k,j}$, ${flag}_{j} \leftarrow 1 $, $\mathbb{Y}^\mathsf{(c)}\left( u_i \right)\leftarrow \varnothing$, $\mathbb{Y}^\mathsf{(c)}\left( s_{j} \right)\leftarrow \varnothing$, $\mathbb{Y}^\mathsf{(s)}\left( \bm{c}_k \right)\leftarrow \varnothing$, $\mathbb{Y}^\mathsf{(s)}\left( s_{j} \right)\leftarrow \varnothing$\ 
		
		\For{$\forall u_i\in\bm{\mathcal{U}^\prime}$}{
			$\bm{c}_k\leftarrow u_i$ forms coalitions based on shared sensing target, where $\bm{c}_k\in\bm{\mathcal{C}^\prime}$
		}
		\While{$ \sum_{u_i\in\bm{\mathcal{U}^\prime}}{flag}_{i} $ and $\sum_{\bm{c}_k\in\bm{\mathcal{C}^\prime}}{flag}_{k}$}{
			\textbf{$ {flag}_{i} \leftarrow {\bf False} $, $ {flag}_{k} \leftarrow {\bf False} $}
			
			\textbf{Calculate:} $\overrightarrow{L^\mathsf{(c)}_i}$ and $\overrightarrow{L^\mathsf{(s)}_k}$ under constraints (\ref{equ. PF MU C4a}) and (\ref{equ. PF MU C5a})
			
			$ F^\mathsf{(c),\star}_i\left\langle \mathcal{X} \right\rangle\leftarrow \overrightarrow{L^\mathsf{(c)}_i}$, $ F^\mathsf{(s),\star}_k\left\langle \mathcal{X} \right\rangle\leftarrow \overrightarrow{L^\mathsf{(s)}_k}$      
			 $\mathbb{Y}^\mathsf{(c)}\left( u_i \right), B_{i,j}^\mathsf{(c)}\left\langle \mathcal{X} \right\rangle, P_{i,j}^\mathsf{(c)}\left\langle \mathcal{X} \right\rangle, \dot{\mathbbm{c}}^\mathsf{(c),Pay}_{i,j}\left\langle \mathcal{X} \right\rangle  \leftarrow F^\mathsf{(c),\star}_i\left\langle \mathcal{X} \right\rangle $, $ \mathbb{Y}^\mathsf{(s)}\left( \bm{c}_k \right), B_{k,j}^\mathsf{(s)}\left\langle \mathcal{X} \right\rangle, P_{k,j}^\mathsf{(s)}\left\langle \mathcal{X} \right\rangle, \dot{\mathbbm{c}}^\mathsf{(s),Pay}_{k,j}\left\langle \mathcal{X} \right\rangle \} \leftarrow F^\mathsf{(s),\star}_k\left\langle \mathcal{X} \right\rangle $

			\If{$ \forall\mathbb{Y}^\mathsf{(c)}\left( u_i \right) \neq \varnothing $ or $ \forall\mathbb{Y}^\mathsf{(s)}\left( \bm{c}_k \right) \neq \varnothing $}{
				\For{$\forall u_i \in \bm{\mathcal{U}^\prime}$ }{
					$ u_i $ sends a proposal about its information to $ s_j $, where $s_j\in\mathbb{Y}^\mathsf{(c)}\left( u_i \right)$}
				\For{$\forall \bm{c}_k \in \bm{\mathcal{C}^\prime}$ }{
					$ \bm{c}_k $ sends a proposal about its information to $ s_j $, where $s_j\in\mathbb{Y}^\mathsf{(s)}\left( \bm{c}_k \right)$}
				
				\While{
					$ \Sigma_{u_i\in \bm{\mathcal{U}^\prime}}{flag}_{i} > 0 $}{
					$ {\widetilde{\mathbb{Y}}}\left(s_j\right) \leftarrow$ collect proposals from clients
					
					$ \mathbb{Y}^\mathsf{(c)}(s_j)$, $\mathbb{Y}^\mathsf{(s)}(s_j) \leftarrow $ choose MUs from $ {\widetilde{\mathbb{Y}}}\left(s_j\right) $ that can achieve the maximization of the utility of BS $s_j$ (i.e., (\ref{equ. PF BSa})) by using DP under constraints (\ref{equ. PF BS C3a}) and (\ref{equ. PF BS C4a})
					
					$ s_j $ temporally accepts the clients in $ \mathbb{Y}^\mathsf{(c)}(s_j) $ and $ \mathbb{Y}^\mathsf{(s)}(s_j) $, and rejects the others
				}
				
				\For{
					$ \forall u_i \in \mathbb{Y}^\mathsf{(c)}\left( s_j \right) $
				}{
					\If{$ u_i $ is rejected by $ s_j $, $V^\mathsf{(c)}_{i,j}\ge \dot{\mathbbm{c}}^\mathsf{(c),Pay}_{i,j}$ and constraint (\ref{equ. PF MU C4a}) is met}{
						$ \dot{\mathbbm{c}}^\mathsf{(c),Pay}_{i,j}\left\langle {\mathcal{X} + 1} \right\rangle \leftarrow \min\left\{ \dot{\mathbbm{c}}^\mathsf{(c),Pay}_{i,j}\left\langle \mathcal{X} \right\rangle + \mathrm{\Delta}p~,{ V}^\mathsf{(c)}_{i,j} \right\} $}
					\Else{$ \dot{\mathbbm{c}}^\mathsf{(c),Pay}_{i,j}\left\langle {\mathcal{X} + 1} \right\rangle \leftarrow \dot{\mathbbm{c}}^\mathsf{(c),Pay}_{i,j}\left\langle \mathcal{X} \right\rangle $}
				}
				
				\For{
					$ \forall \bm{c}_k \in \mathbb{Y}^\mathsf{(s)}\left( s_j \right) $
				}{
					\If{$ u_i $ is rejected by $ s_j $, $V^\mathsf{(s)}_{k,j}\ge \dot{\mathbbm{c}}^\mathsf{(s),Pay}_{k,j}$ and constraint (\ref{equ. PF MU C5a}) is met}{
						$ \dot{\mathbbm{c}}^\mathsf{(s),Pay}_{k,j}\left\langle {\mathcal{X} + 1} \right\rangle \leftarrow \min\left\{ \dot{\mathbbm{c}}^\mathsf{(s),Pay}_{k,j}\left\langle \mathcal{X} \right\rangle + \mathrm{\Delta}p~,{ V}^\mathsf{(s)}_{k,j} \right\} $}
					\Else{$ \dot{\mathbbm{c}}^\mathsf{(s),Pay}_{k,j}\left\langle {\mathcal{X} + 1} \right\rangle \leftarrow \dot{\mathbbm{c}}^\mathsf{(s),Pay}_{k,j}\left\langle \mathcal{X} \right\rangle $}
				}

                $ p_{i,\mathbbm{n}}^\mathsf{(c)}\leftarrow \dot{\mathbbm{c}}^\mathsf{(c),Pay}_{i,j}\left\langle \mathcal{X}+1 \right\rangle, p_{i,\mathbbm{n}}^\mathsf{(c)} \in F^\mathsf{(c),\star}_{i}\left\langle \mathcal{X} \right\rangle$, $
                    p_{k,\mathbbm{m}}^\mathsf{(s)}\leftarrow \dot{\mathbbm{c}}^\mathsf{(c),Pay}_{i,j}\left\langle \mathcal{X}+1 \right\rangle, p_{i,\mathbbm{n}}^\mathsf{(c)} \in F^\mathsf{(c),\star}_{i}\left\langle \mathcal{X} \right\rangle$
                
				\If{$\mathcal{X}\le2$ and there exists $F^\mathsf{(c),\star}_{i}\left\langle \mathcal{X}-1 \right\rangle \neq F^\mathsf{(c),\star}_{i}\left\langle \mathcal{X} \right\rangle $ or $F^\mathsf{(s),\star}_{k}\left\langle \mathcal{X}-1 \right\rangle \neq F^\mathsf{(s),\star}_{k}\left\langle \mathcal{X} \right\rangle $}{
					$ {flag}_{i} \leftarrow {\bf True} $, $ {flag}_{k} \leftarrow {\bf True} $,}	\
					$ \mathcal{X}\leftarrow \mathcal{X}+1 $	
			}

		}

		$\nu^\mathsf{(c)}(s_j)\leftarrow\mathbb{Y}^\mathsf{(c)}(s_j)$, $\nu^\mathsf{(c)}(u_i)\leftarrow \mathbb{Y}^\mathsf{(c)}(u_i)$,
		$\nu^\mathsf{(s)}(s_j)\leftarrow\mathbb{Y}^\mathsf{(s)}(s_j)$, $\nu^\mathsf{(s)}(\bm{c}_k)\leftarrow \mathbb{Y}^\mathsf{(s)}(\bm{c}_k)$ , $\mathcal{X} \leftarrow \mathcal{X}-1$

		\textbf{Return:} $\dot{\mathbb{C}}_{i,j}^\mathsf{(c)} =\{ B_{i,j}^\mathsf{(c)}\left\langle \mathcal{X} \right\rangle, P_{i,j}^\mathsf{(c)}\left\langle \mathcal{X} \right\rangle, \dot{\mathbbm{c}}^\mathsf{(c),Pay}_{i,j}\left\langle \mathcal{X} \right\rangle \}$, $\dot{\mathbb{C}}_{k,j}^\mathsf{(s)} =\{ B_{k,j}^\mathsf{(s)}\left\langle \mathcal{X} \right\rangle, P_{k,j}^\mathsf{(s)}\left\langle \mathcal{X} \right\rangle, \mathbbm{s}^\mathsf{(s),Pay}_{k,j}\left\langle \mathcal{X} \right\rangle \} $}
\end{algorithm}
Our proposed onEBW$^2$M mechanism enables BSs with surplus resources and clients to negotiate the quantity and pricing of bandwidth and power resources for two distinct service types, similar with the offline mode: individual MUs engage in resource trading for communication services, while coalitions participate in resource trading for sensing services. Note that when there exist clients with unmet resource demands—including voluntary clients and those without long-term contracts—as well as BSs with surplus resources, we implement the onEBW$^2$M mechanism to establish temporary contracts for real-time resource allocation. Specifically, onEBW$^2$M is similar to the offRFW$^2$M mechanism we proposed. The main difference is that in onEBW$^2$M, transactions use current market/resource information to negotiate acceptable terms for temporary contracts, without considering expected utilities and risks introduced by dynamic networks. Thus, we omit its details here as constrained by limited space.

\noindent\textbf{Computational complexity of onEBW$^2$M:} The computational complexity of our proposed onEBW$^2$M depends on the number of rounds involved in Alg. 2, denoted by \( \mathcal{X}^{\mathsf{max}} \), the remain resources \( B^\prime_j \) and \( P^\prime_j \), as well as the number of clients sending requests to BS \( s_j \) in the $ \mathcal{X}^\mathsf{\text{th}} $ round, denoted as \( |\widetilde{\mathbb{Y}}(s_j)|_{\mathcal{X}} \). In particular, the overall complexity of onEBW$^2$M for each BS \( s_j \) is:
$\sum_{\mathcal{X}=1}^\mathsf{\mathcal{X}^{\mathsf{max}}} \mathcal{O}\left( |\widetilde{\mathbb{Y}}(s_j)|_{\mathcal{X}} \times B^\prime_j \times P^\prime_j \right)$.

\noindent \textbf{Solution Characteristics:} This work provides a novel perspective on online trading-driven temporary contract determination by designing a \textit{effective backup win-win matching mechanism that achieves mutually beneficial utilities} for both parties. From the clients' perspective, the preference list in practical transactions for each client is similar to  Definitions \ref{def 5} and \ref{def 6}, under constraints (\ref{equ. PF MU C4a}) and (\ref{equ. PF MU C5a}), ensuring that the selected solution maximizes their utility. From the BSs' perspective, the appropriate contract terms are selected from the feasible solutions reported by clients, utilizing a DP algorithm to maximize the utility of the BS (line 15, Alg. 2).

\subsection{Design Targets and Property Analysis}
In the following, we analyze the key properties of our unique matching mechanism within this online trading framework.

{\begin{Defn}(Blocking Pairs for Communication and Sensing Services in onEBW$^2$M)
Under a given matching \( \nu^\mathfrak{(X)} \), a client \( \mathsf{a} \) (representing either a MU \( u_i \) for communication when \( \mathfrak{(X)} = \mathsf{(c)} \) or a coalition \( \bm{c}_k \) for sensing when \( \mathfrak{(X)} = \mathsf{(s)} \) ), a BS set \( \mathbb{S} \subseteq \bm{\mathcal{S}} \), and a contract set \( \dot{\mathbb{C}}^\prime \), denoted by the triplet \( \left(\mathsf{a}; \mathbb{S}; \dot{\mathbb{C}}^\prime\right) \), may form one of the following two types of blocking pairs.

\noindent \textbf{Type 1 blocking pair:} The pair satisfies the following two conditions:

    \noindent $\bullet$ The client $\mathsf{a}$ prefers the BS set $ \mathbb{S} $ over its currently matched set $ \nu^\mathfrak{(X)}(\mathsf{a}) $:
    \begin{equation}\label{Nequ. 87}
  \begin{aligned}
        \mathbb{E}\left[ u^{\mathfrak{(X)},\mathsf{U}}(\mathsf{a},\mathbb{S},\dot{\mathbb{C}}^\prime) \right] >
        \mathbb{E}\left[ u^{\mathfrak{(X)},\mathsf{U}}(\mathsf{a}, \nu^\mathfrak{(X)}(\mathsf{a}), \mathbb{C}^\mathfrak{(X)}_{i,j}) \right]
        \end{aligned}
    \end{equation}
    
    \noindent $\bullet$ Every BS $s_j \in \mathbb{S}$ prefers to reallocate its service from current matches to include $\mathsf{a}$, i.e., there exists a subset $\nu^\mathfrak{(X)\prime}(s_j)$ of currently matched agents to be evicted, such that:
    \begin{equation}\label{Nequ. 88}
     \begin{aligned}
        &\mathbb{E}\left[u^{\mathfrak{(X)},\mathsf{S}}(s_j, \{\nu^\mathfrak{(X)}(s_j) \setminus \nu^\mathfrak{(X)\prime}(s_j)\} \cup \{\mathsf{a}\}, \dot{\mathbb{C}}^\prime)\right]>\\& 
        \mathbb{E}\left[u^{\mathfrak{(X)},\mathsf{S}}(s_j, \nu^\mathfrak{(X)}(s_j), \mathbb{C}^\mathfrak{(X)}_{i,j})\right]
        \end{aligned}
    \end{equation}

\noindent \textbf{Type 2 blocking pair:} The pair satisfies the following two conditions:

\noindent $\bullet$ The client $\mathsf{a}$ prefers the BS set $ \mathbb{S} $ over its currently matched set, as shown in (\ref{Nequ. 88}).
    
\noindent $\bullet$  Each BS $s_j \in \mathbb{S}$ prefers to additionally serve $\mathsf{a}$ while maintaining its current matches:
    \begin{equation}\label{Nequ. 90}\
     \begin{aligned}
        &\mathbb{E}\left[u^{\mathfrak{(X)},\mathsf{S}}(s_j, \nu^{\mathfrak{(X)}}(s_j) \cup \{\mathsf{a}\}, \dot{\mathbb{C}}^\prime)\right] >\\&
        \mathbb{E}\left[u^{\mathfrak{(X)},\mathsf{S}}(s_j, \nu^{\mathfrak{(X)}}(s_j), \mathbb{C}^{\mathfrak{(X)}}_{i,j})\right]
        \end{aligned}
    \end{equation}
\end{Defn}}

Building on the above two definitions, a Type 1 blocking pair undermines the stability of the matching by incentivizing a BS to reallocate its resources to a different set of clients that yield a higher utility. Similarly, a Type 2 blocking pair introduces instability, as the BS possesses residual resources that could be allocated to additional clients, thereby further maximizing its utility. These blocking pairs are used in the following to define the major characteristics of our proposed matching methodology.

\begin{Prop}\label{Prop 14}(Individual rationality of onEBW$^2$M) The proposed onEBW$^2$M mechanism ensures individual rationality for BSs, individual MUs, and sensing coalitions under the following conditions:
	
	\noindent
	$\bullet$ For each BS: the bandwidth and power resources of a BS $s_j$ booked to matched clients $\nu^\mathsf{(c)}\left(s_j\right)$ and coalitions $\nu^\mathsf{(s)}\left(s_j\right)$ does not exceed $B^\prime_j$ and $P^\prime_j$, i.e., constraints (\ref{equ. PF BS C3a}) and (\ref{equ. PF BS C4a}) are met.
	
	\noindent
	$\bullet$ For each client (i.e., each MU and each coalition): \textit{(i)} The value obtained by each client is at least equal to the payment it makes, ensuring that constraint (\ref{equ. PF MU C3a}) is met; \textit{(ii)} each client are satisfying constraints (\ref{equ. PF MU C4a}) and (\ref{equ. PF MU C5a}).
\end{Prop}

\begin{Prop}(Fairness of onEBW$^2$M): The proposed onEBW$^2$M mechanism ensures fairness by preventing the formation of Type 1 blocking pairs, ensuring that clients are satisfied with their matched BSs and no BS is incentivized to reallocate its resources to a different set of clients at the expense of existing agreements.\end{Prop}
\begin{Prop}(Non-wastefulness of onEBW$^2$M): onEBW$^2$M guarantees non-wastefulness by preventing the formation of Type 2 blocking pairs, ensuring that BSs efficiently utilize their resources without leaving surplus allocations that could accommodate additional clients.\end{Prop}

\begin{Prop}(Strong Stability of onEBW$^2$M) \label{Prop 17}The proposed onEBW$^2$M mechanism achieves strong stability by ensuring that the matching remains individually rational, fair, and non-wasteful.
\end{Prop}

\begin{Prop}(Stability of Sensing Coalitions in onEBW$^2$M)
	In onEBW$^2$M, each sensing coalition $\bm{c}_k$ is stable when the following conditions are satisfied: 
	
	\noindent $\bullet$ For each MU $u_i$ in sensing coalition $\bm{c}_k$, its utility onEBW$^2$M above $u^\mathsf{(s)}_\mathsf{\min}$; 
	
	\noindent $\bullet$ For each MU $u_i$ in sensing coalition $\bm{c}_k$, the utility obtained by joining the coalition \( \bm{c}_k \) is greater than the utility when trading as an individual.
\end{Prop}

Further, for the MOO problem collectively defined by $ \bm{\mathcal{F}^\mathsf{U^\prime}} $ and $ \bm{\mathcal{F}^\mathsf{S^\prime}} $, a Pareto improvement occurs when the social welfare (i.e., the summation of utilities of clients and BSs in the considered market)\cite{RW Matching3} can be increased with another feasible matching result. A matching is thus weakly Pareto optimal when no further Pareto improvement is possible, which is a desired property.


\begin{Prop}(Weak Pareto optimality of onEBW$^2$M) The proposed onEBW$^2$M mechanism ensures weak Pareto optimality by preventing further Pareto improvements, ensuring that no alternative matching configuration can increase the social welfare of the system.
\end{Prop}

We next examine the aforementioned property of onEBW$^2$M, as outlined below:

\begin{Prop}\label{Prop 19}
	(Convergence of a set of matching in onEBW$^2$M) Alg. 2 converges within finite rounds.
\end{Prop}
\begin{proof}
	As the onEBW$^2$M refers to a set of M2M matching (matching between BSs and individual MUs, as well as matching between BSs and coalitions), we utilize the DP algorithm to transform the problem into a two-dimensional 0-1 knapsack problem \cite{RW Matching3}. After a finite number of rounds, each client's payment can either be accepted or reach its maximum value while considering constraints (\ref{equ. PF BS C3a}) and (\ref{equ. PF BS C4a}) (e.g., lines 17-26, Alg. 2), which ensuring the convergence.
\end{proof}

\begin{Prop}
	(Individual rationality of onEBW$^2$M) The proposed onEBW$^2$M mechanism ensures individual rationality for All the BSs, individual MUs, and sensing coalitions are individual rational in the onEBW$^2$M.
\end{Prop}
\begin{proof}
	We offer the analysis on proving the individual rationality of both BSs and clients.
	
	\textbf{Individual rationality of BSs.} each BS $s_j$ regards $B^\prime_j$ and $P^\prime_j$ as up limit of resources for serving MUs, and the actual number of matched clients of $s_j$ will definitely not exceed its remain resource supply (e.g., line 15, Alg. 2).
	
	\textbf{Individual rationality of clients.} Lines 17-26 of Alg. 2 ensure that the value obtained by each client is at least equal to the payment it makes, thereby satisfying constraint (\ref{equ. PF MU C3a}). Furthermore, lines 6, 18 and 23 of Alg. 2 guarantee that each client are satisfying constraints (\ref{equ. PF MU C4a}) and (\ref{equ. PF MU C5a}),.
	
	As a summary, clients and BSs are individual rationality in our proposed onEBW$^2$M.
\end{proof}

\begin{Prop}
	No blocking pair can exist in the Resource Trading for Communication Services in onEBW$^2$M.
\end{Prop}
\begin{proof}
	We show there is no blocking pair of either Type 1 or Type 2, as following:
	
	\noindent 
	$\bullet$ \textbf{There is no Type 1 blocking pair related to communication services of onEBW$^2$M.} We offer the proof by considering contradiction.
	
	Under a given matching $ \nu^\mathsf{(c)} $, MU $ u_i $ and BS $ s_j $ form a Type 1 blocking pair $ \left(u_i;s_j;\dot{\mathbb{C}}^\prime\right) $.
	If MU $ u_i $ does not trading with BS $ s_j $, the payment of MU $ u_i $ during the last round can only be its valuation $V^\mathsf{(c)}_{i,j}$, as shown by (\ref{59A}) and (\ref{60A}).
	\begin{equation}\label{59A}{\small
			\begin{aligned}
				\dot{\mathbbm{c}}^\mathsf{(c),Pay}_{i,j} = V^\mathsf{(c)}_{i,j},
		\end{aligned}}
	\end{equation}
	\begin{equation}\label{60A}
		\begin{aligned}
			&u^\mathsf{(c),S}\left(s_j,\left\{\nu\left( s_j \right)\backslash\widetilde{\nu^\mathsf{(c)\prime}}\left( s_j \right)\right\} \cup \left\{ u_i \right\},\dot{\mathbb{C}}^\prime \right) \\&< u^\mathsf{(c),S}\left(s_j,\nu\left( s_j\right),\dot{\mathbb{C}}^\mathsf{(c)}_{i,j} \right).\\
		\end{aligned}
	\end{equation} 
	
	If BS $ s_j $ selects MU $ u_i $, we have $ \dot{\mathbbm{c}}^\mathsf{(c),Pay}_{i,j}\left\langle \mathcal{X}^\mathsf{*} \right\rangle\leq \dot{\mathbbm{c}}^\mathsf{(c),Pay}_{i,j}\left\langle \mathcal{X} \right\rangle =V^\mathsf{(c)}_{i,j} $ and the following (\ref{81aa})
	\begin{equation}\label{81aa}{\small
			\begin{aligned}
				&u^\mathsf{(c),S}\left(s_j,\left\{\nu\left( s_j \right)\backslash\widetilde{\nu^\mathsf{(c)\prime}}\left( s_j \right)\right\} \cup \left\{ u_i \right\},\dot{\mathbb{C}}^\prime \right) \geq\\& u^\mathsf{(c),S}\left(s_j,\left\{\nu\left( s_j \right)\backslash\widetilde{\nu^\mathsf{(c)\prime\prime}}\left( s_j \right)\right\} \cup \left\{ u_i \right\},\dot{\mathbb{C}}^\prime \right),\\
		\end{aligned}}
	\end{equation}
	where $ 
	\widetilde{\nu^\mathsf{(c)\prime\prime}}\left(s_j\right) \subseteq \widetilde{\nu^\mathsf{(c)\prime}}\left(s_j\right) $. From (\ref{60a}) and (\ref{81aa}), we can get
	\begin{equation}\label{key}\small{
			\begin{aligned}
				&u^\mathsf{(c),S}\left(s_j,\nu\left( s_j\right),\dot{\mathbb{C}}^\mathsf{(c)}_{i,j} \right)> \\&u^\mathsf{(c),S}\left(s_j,\left\{\nu\left( s_j \right)\backslash\widetilde{\nu^\mathsf{(c)\prime\prime}}\left( s_j \right)\right\} \cup \left\{ u_i \right\},\dot{\mathbb{C}}^\prime \right),
		\end{aligned}}
	\end{equation}
	which is contrary to (\ref{Nequ. 88}),thus ensuring the inexistence of Type 1 blocking pairs.
	
	\noindent 
	$\bullet$ \textbf{There is no Type 2 blocking pair related to communication services of onEBW$^2$M.}
	We conduct the proof by considering cases of contradiction. 
	
	Under a given matching $ \nu^\mathsf{(s)} $, MU $ u_i $ and BS $ s_j $ form a Type 2 blocking pair $ \left(u_i;s_j;\dot{\mathbb{C}}^\prime\right) $, as shown by (\ref{Nequ. 90}).
	If MU $ u_i $ is rejected by BS $ s_j $, the final payment of $ u_i $ can be set by $ \dot{\mathbbm{c}}^\mathsf{(c),Pay}_{i,j} = V^\mathsf{(c)}_{i,j} $, where the only reason of such a rejection is that $ s_j $ has no surplus resources. However, the coexistence of (\ref{Nequ. 90}) shows that BS $ s_j $ has adequate resource supply to serve MUs, which contradicts our previous assumption. Therefore, we prove that there is no Type 2 blocking pair.
	
	As a summary, no blocking pair can exist during the matching related to communication services in onEBW$^2$M. 
\end{proof}

\begin{Prop}\label{Prop 22}
	No blocking pair can exist in the Resource Trading for Sensing Services in onEBW$^2$M.
\end{Prop}
\begin{proof}
	We show there is no blocking pair of either Type 1 or Type 2, as following:
	
	\noindent 
	$\bullet$ \textbf{There is no Type 1 blocking pair related to sensing services of onEBW$^2$M.} We offer the proof by considering contradiction.
	
	Under a given matching $ \nu^\mathsf{(s)} $, coalition $ \bm{c}_k $ and BS $ s_j $ form a Type 1 blocking pair $ \left(\bm{c}_k;s_j;\dot{\mathbb{C}}^\prime\right) $.
	If $ \bm{c}_k $ does not trading with BS $ s_j $, the payment of MU $ u_i $ during the last round can only be its valuation $V^\mathsf{(s)}_{i,j}$, as shown by (\ref{59A1}) and (\ref{60A1}).
	\begin{equation}\label{59A1}{\small
			\begin{aligned}
				\dot{\mathbbm{c}}^\mathsf{(s),Pay}_{k,j} = V^\mathsf{(s)}_{k,j},
		\end{aligned}}
	\end{equation}
	\begin{equation}\label{60A1}
		\begin{aligned}
			&u^\mathsf{(s),S}\left(s_j,\left\{\nu\left( s_j \right)\backslash\widetilde{\nu^\mathsf{(s)\prime}}\left( s_j \right)\right\} \cup \left\{ \bm{c}_k \right\},\dot{\mathbb{C}}^\prime \right) \\&<u^\mathsf{(s),S}\left(s_j,\nu\left( s_j\right),\dot{\mathbb{C}}^\mathsf{(s)}_{k,j} \right).\\
		\end{aligned}
	\end{equation} 
	
	If BS $ s_j $ selects coalition $ \bm{c}_k $, we have $ \dot{\mathbbm{c}}^\mathsf{(s),Pay}_{k,j}\left\langle \mathcal{X}^\mathsf{*} \right\rangle\leq \dot{\mathbbm{c}}^\mathsf{(s),Pay}_{k,j}\left\langle \mathcal{X} \right\rangle =V^\mathsf{(s)}_{k,j} $ and the following (\ref{81aaa})
	\begin{equation}\label{81aaa}{\small
			\begin{aligned}
				&u^\mathsf{(s),S}\left(s_j,\left\{\nu\left( s_j \right)\backslash\widetilde{\nu^\mathsf{(s)\prime}}\left( s_j \right)\right\} \cup \left\{ \bm{c}_k \right\},\dot{\mathbb{C}}^\prime \right) \geq\\& u^\mathsf{(s),S}\left(s_j,\left\{\nu\left( s_j \right)\backslash\widetilde{\nu^\mathsf{(s)\prime\prime}}\left( s_j \right)\right\} \cup \left\{ \bm{c}_k \right\},\dot{\mathbb{C}}^\prime \right),\\
		\end{aligned}}
	\end{equation}
	where $ 
	\widetilde{\nu^\mathsf{(s)\prime\prime}}\left(s_j\right) \subseteq \widetilde{\nu^\mathsf{(s)\prime}}\left(s_j\right) $. From (\ref{60A1}) and (\ref{81aaa}), we can get
	\begin{equation}\label{key}\small{
			\begin{aligned}
				&u^\mathsf{(s),S}\left(s_j,\nu\left( s_j\right),\dot{\mathbb{C}}^\mathsf{(s)}_{k,j} \right)> \\&u^\mathsf{(s),S}\left(s_j,\left\{\nu\left( s_j \right)\backslash\widetilde{\nu^\mathsf{(s)\prime\prime}}\left( s_j \right)\right\} \cup \left\{ \bm{c}_k \right\},\dot{\mathbb{C}}^\prime \right),
		\end{aligned}}
	\end{equation}
	which is contrary to (\ref{Nequ. 88}), and thus proving the inexistence of Type 1 blocking pairs.
	
	\noindent 
	$\bullet$ \textbf{There is no Type 2 blocking pair related to sensing services of onEBW$^2$M.}
	We conduct the proof by considering cases of contradiction. 
	
	Under a given matching $ \nu^\mathsf{(s)} $, coalition $ \bm{c}_k $ and BS $ s_j $ form a Type 2 blocking pair $ \left(\bm{c}_k;s_j;\dot{\mathbb{C}}^\prime\right) $, as shown by (\ref{Nequ. 90}).
	If coalition $ \bm{c}_k $ is rejected by BS $ s_j $, the final payment of $ \bm{c}_k $ can be set by $ \dot{\mathbbm{c}}^\mathsf{(s),Pay}_{k,j} =V^\mathsf{(s)}_{k,j}$, where the only reason of such a rejection is that $ s_j $ has no surplus resources. However, the coexistence of (\ref{Nequ. 90}) shows that BS $ s_j $ has adequate resource supply to serve coalitions, which contradicts our previous assumption. Therefore, we prove that there is no Type 2 blocking pair.
	
	As a summary, no blocking pair can exist during the matching related to sensing services in onEBW$^2$M. 
\end{proof}

\begin{Prop}\label{Prop 24}
	(Fairness, Non-wastefulness, Strong Stability of onEBW$^2$M) onEBW$^2$M is fair, non-wasteful, and strongly stable.
\end{Prop}
\begin{proof}
	Since the matching result of Alg. 2 holds Propositions \ref{Prop 19}-\ref{Prop 22}, according to Propositions \ref{Prop 14}-\ref{Prop 17}, our proposed onEBW$^2$M is strongly fairness, non-wastefulness, strong stability.
\end{proof}

\begin{Prop}(Stability of Sensing Coalitions in onEBW$^2$M) The proposed onEBW$^2$M ensures that each sensing coalition $\bm{c}_k$ is stable.\end{Prop}
\begin{proof}
	Due to line 24 in Alg. 2, each MU \( u_i \) in the sensing coalition \( \bm{c}_k \) ensures its utility to exceed \( u^\mathsf{(s)}_\mathsf{\min} \). Furthermore, MUs within a coalition share both costs and profits, leading to a lower utility per MU compared to trading individually. Therefore, we can conclude that joining coalition $\bm{ c}_k$ will not result in a lower utility than trading as an individual.
\end{proof}

\begin{Prop}
	(Weak Pareto optimality of onEBW$^2$M) The proposed onEBW$^2$M provides a weak Pareto optimality.
\end{Prop}
\begin{proof}
Reviewing our design of onEBW$^2$M, each participant (e.g., client, BS) makes decisions according to its preference list to determine the trading counterpart and the specific terms of the temporary contract. If the alternative choice ranks higher in the participant's preference list, they will switch their matching target and contract in the following round. Such a switch indicates that returning to the previous choice would not result in a higher utility. For an MU \( u_i \), if there exists a BS \( s_j \) that can offer a higher utility than its currently matched BS, \( u_i \) and \( s_j \) are more inclined to establish a matching relationship. This, however, forms a blocking pair. Since Proposition \ref{Prop 24} confirms that our proposed onEBW$^2$M is stable and free of blocking pairs, there is no possibility of Pareto improvement when the procedure of matching \( \nu^\mathsf{(c)} \) terminates. Similarly, we can infer that there is no Pareto improvement in matching \( \nu^\mathsf{(s)} \) (e.g., Propositions \ref{Prop 22} and \ref{Prop 24}). In conclusion, the onEBW$^2$M game we study is said to be weak Pareto optimal.
\end{proof}

\end{document}